\documentclass[11pt]{article}

\usepackage[utf8]{inputenc} 
\usepackage[T1]{fontenc}    
\usepackage{hyperref}       
\usepackage{url}            
\usepackage{booktabs}       
\usepackage{amsfonts}       
\usepackage{nicefrac}       
\usepackage{microtype}      
\usepackage{setspace}
\usepackage[round]{natbib}
\usepackage[left=1.25in,top=1.25in,right=1.25in,bottom=1.25in,head=1.25in]{geometry}

\usepackage{enumerate}
\usepackage{titling}

\newcommand{\ones}{\mathbf 1}
\newcommand{\reals}{{\mbox{\bf R}}}

\newcommand{\prox}{\mathbf{prox}}
\newcommand{\spann}{\mathbf{span}} 

\newcommand{\Expect}{\mathop{\bf E{}}}

\newcommand{\prob}{\mathop{\bf Pr}} 

\newcommand{\eg}{{\it e.g.}}
\newcommand{\ie}{{\it i.e.}}

\newcommand{\nb}{{\it n.b.}} 

\usepackage{graphicx,psfrag,amsmath,amsfonts,verbatim} 
\usepackage[small,bf]{caption}
\usepackage{url}

\ifdefined\aaslides
\else

\usepackage{amsthm}
\newtheorem{theorem}{Theorem}[section]
\newtheorem{lemma}[theorem]{Lemma}

\fi

\newcommand{\minimize}{\mathop{\mbox{minimize}}} 
\newcommand{\minimizewrt}[1]{\underset{#1}{\minimize}}

\newcommand{\subjectto}{\mbox{subject to}}

\usepackage{xcolor}

\newcommand{\optprobstart}{\begin{equation}}
\newcommand{\optprobend}{\end{equation}}

\newcommand{\optprobstartnn}{\begin{equation*}}
\newcommand{\optprobendnn}{\end{equation*}}

\usepackage{listings}

\usepackage{color}

\newcommand{\mtxstart}{\left[\begin{array}}
\newcommand{\mtxend}{\end{array}\right]}	

\usepackage{tikz}

\newcommand{\var}{\mathop{\bf var}}

\usepackage{bm}
\usepackage{wrapfig}
\usepackage{bbm}
\usepackage{algorithm}
\usepackage{algorithmic}
\usepackage{subcaption}

\usepackage{floatpag}
\floatpagestyle{empty}

\newcommand{\mqgm}{MQGM}
\newcommand{\mb}{MB}
\newcommand{\glasso}{GLasso}
\newcommand{\spacejam}{SpaceJam}
\newcommand{\tiger}{TIGER}
\newcommand{\laplace}{Laplace}
\newcommand{\nonparas}{nonparanormal skeptic}
\newcommand{\nonpara}{Nonpara.}

\def\P{\prob}                 
\def\R{\reals}                 
\def\E{\Expect}              
\def\cF{\mathcal{F}}      

\title{The Multiple Quantile Graphical Model}
\author{Alnur Ali$^1$ \and J. Zico Kolter$^{1,2}$ \and Ryan J.~Tibshirani$^{1,3}$} 
\date{$^1$Machine Learning Department, Carnegie Mellon University \\ 
$^2$Computer Science Department, Carnegie Mellon University \\ 
$^3$Department of Statistics, Carnegie Mellon University}

\begin{document}
\maketitle

\begin{abstract}
We introduce the Multiple Quantile Graphical Model (\mqgm),  
which extends the {\it neighborhood selection} approach of
Meinshausen and B{\"u}hlmann for learning sparse graphical models.
The latter is defined by the basic subproblem of modeling the
conditional mean of one variable as a sparse function of all others.
Our approach models a set of conditional quantiles of
one variable as a sparse function of all others, and hence offers 
a much richer, more expressive class of conditional distribution
estimates.  We establish that,
under suitable regularity conditions, the \mqgm~identifies the exact  
conditional independencies with probability tending to one as the
problem size grows, even outside of the usual homoskedastic Gaussian
data model. We develop an efficient algorithm for fitting the
\mqgm~using the alternating direction method of multipliers.   
We also describe a strategy for sampling from the joint distribution
that underlies the \mqgm~estimate. Lastly, 
we present detailed experiments that demonstrate the flexibility and
effectiveness of the \mqgm~in modeling hetereoskedastic 
non-Gaussian data.
\end{abstract}

\newcommand{\fsw}{0.31\textwidth}

\section{Introduction}
\label{sec:intro}

We consider modeling the joint distribution $\P(y_1,\ldots,y_d)$ of $d$ random variables, given $n$ independent draws from this distribution \smash{$y^{(1)},\ldots,y^{(n)} \in \reals^d$}, where possibly $d \gg n$.  Later, we generalize this setup and consider modeling the conditional distribution $\P(y_1,\ldots,y_d | x_1,\ldots,x_p)$, given $n$ independent pairs \smash{$(x^{(1)},y^{(1)}),\ldots,(x^{(n)},y^{(n)}) \in \reals^{p+d}$}. Our starting point is the {\it neighborhood selection} method of \cite{meinshausen2006high}, which is typically considered in the context of multivariate Gaussian data, and seen as a tool for {\it covariance selection} \citep{dempster1972covariance}: when $\P(y_1,\ldots,y_d)$ is a multivariate Gaussian distribution, it is a well-known fact $y_j$ and $y_k$ are conditionally independent given the remaining variables if and ony if the coefficent corresponding to $y_k$ is zero in the (linear) regression of $y_j$ on all other variables (\eg, \cite{lauritzen1996graphical}).  Therefore, in  neighborhood selection we compute, for each $k=1,\ldots,d$, a lasso regression
 --- in order to obtain a small set of conditional dependencies --- of $y_k$ on the remaining variables, \ie,
\begin{equation}
\minimize_{\theta_k \in {\bf R}^d} \; \sum_{i=1}^n \bigg( y_k^{(i)} - \sum_{j \neq k} \theta_{kj} y_j^{(i)} \bigg)^2 + \lambda \| \theta_k \|_1,
\label{eq:mb}
\end{equation}
for a tuning parameter $\lambda > 0$. This strategy can be seen as a {\it pseudolikelihood} approximation \citep{besag1974}, 
\begin{equation}
\P(y_1,\ldots,y_d) \approx \prod_{k=1}^d \P(y_k | y_{\neg k}),
\label{eq:pseudo}
\end{equation}
where $y_{\neg k}$ denotes all variables except $y_k$.  Under the multivariate Gaussian model for $\P(y_1,\ldots,y_d)$, the conditional distributions $\P(y_k | y_{\neg k})$, $k=1,\ldots,d$ here are (univariate) Gaussians, and maximizing the pseudolikelihood in \eqref{eq:pseudo} is equivalent to separately maximizing the conditionals, as is precisely done in \eqref{eq:mb} (with induced sparsity), for $k=1,\ldots,d$.  

Following the pseudolikelihood-based approach traditionally means carrying out three steps: (i) we write down a suitable family of joint distributions for $\P(y_1,\ldots,y_d)$, (ii) we derive the conditionals $\P(y_k | y_{\neg k})$, $k=1,\ldots,d$, and then (iii) we maximize each conditional likelihood by (freely) fitting the parameters. Neighborhood selection, and a number of related approaches that came after it (see Section \ref{sec:mb}), can be all thought of in this workflow. In many ways, step (ii) acts as the bottleneck here, and to derive the conditionals, we are usually limited to a homoskedastic and parameteric family for the joint distribution.

The approach we take in this paper differs somewhat substantially, as we {\it begin} by directly modeling the conditionals in \eqref{eq:pseudo}, without any preconceived model for the joint distribution --- in this sense, it may be seen a type of {\it dependency network} \citep{heckerman2000dependency} for continuous data. We also employ heteroskedastic, nonparametric models for the conditional distributions, which allows us great flexibility in learning these conditional relationships. Our method, called the Multiple Quantile Graphical Model (\mqgm), is a marriage of ideas in high-dimensional, nonparametric, multiple quantile regression with those in the dependency network literature (the latter is typically focused on discrete, not continuous, data).

An outline for this paper is as follows.  Section \ref{sec:bkgd} reviews background material, and Section \ref{sec:model} develops the \mqgm~estimator.  Section \ref{sec:theory} studies basic properties of the \mqgm, and establishes a structure recovery result under appropriate regularity conditions, even for heteroskedastic, non-Gaussian data.  Section \ref{sec:algos} describes an efficient ADMM algorithm for estimation, and Section \ref{sec:exps} presents empirical examples comparing the \mqgm~versus common alternatives. Section \ref{sec:disc} concludes with a discussion.

\section{Background}
\label{sec:bkgd}

\subsection{Neighborhood selection and related methods}
\label{sec:mb}

Neighborhood selection has motivated a number of methods for learning sparse graphical models. The literature here is vast; we do not claim to give a complete treatment, but just mention some relevant approaches. Many pseudolikelihood approaches have been proposed, see, \eg, \cite{rocha2008path,peng2009partial,friedman2010applications,liu2012tiger,khare2014convex,alnur}.  These works exploit the connection between estimating a sparse inverse covariance matrix and regression, and they vary in terms of the optimization algorithms they use and the theoretical guarantees they offer.

In a related but distinct line of research, \cite{yuan2007,BEA:08,friedman2008sparse,rothman2008} proposed $\ell_1$-penalized likelihood estimation in the Gaussian graphical model, a method now generally termed the {\it graphical lasso} (\glasso).  Following this, several recent papers have extended the \glasso~in various ways.  \cite{finegold2011} examined a modification based on the multivariate Student $t$-distribution, for robust graphical modeling. \cite{sohn2012joint,yuan2014partial,wytock2013sparse} considered conditional distributions of the form $\P(y_1,\ldots,y_d|x_1,\ldots,x_p)$.  \cite{lee2013} proposed a model for mixed (both continuous and discrete) data types, generalizing both \glasso~and pairwise Markov random fields.  \cite{liu2009nonparanormal,liu2012high} used copulas for learning non-Gaussian graphical models.

A strength of neighborhood-based (\ie, pseudolikelihood-based) approaches lies in their simplicity; because they essentially reduce to a collection of univariate probability models, they are in a sense much easier to study outside of the typical homoskedastic, Gaussian data setting.  \cite{hoefling2009,yang2012,yang2015} elegantly studied the implications of using univariate exponential family models for the conditionals in \eqref{eq:pseudo}. 
Closely related to pseudoliklihood approaches are dependency networks \citep{heckerman2000dependency}.  Both frameworks focus on the conditional distributions of one variable given all the rest; the difference lies in whether or not the model for conditionals stems from first specifying some family of joint distributions (pseudolikelihood methods), or not (dependency networks).  Dependency networks have been thoroughly studied for discrete data, \eg, \cite{heckerman2000dependency,neville2004dependency}.
For continuous data, \cite{voorman2014} proposed modeling the mean in a Gaussian neighborhood regression as a nonparametric, additive function of the remaining variables, yielding flexible relationships --- this is a type of dependency network for continuous data (though it is not described by the authors in this way).  Our method, the \mqgm, also deals with continuous data, and is the first to our knowledge that allows for fully nonparametric conditional distributions, as well as nonparametric contributions of the neighborhood variables, in each local model.

\subsection{Quantile regression}
\label{sec:qr}

In linear regression, we estimate the conditional mean of $y | x_1,\ldots,x_p$ based on data samples. Similarly, in {\it $\alpha$-quantile regression} \citep{koenker1978}, we estimate the conditional $\alpha$-quantile of $y | x_1,\ldots,x_p$, formally \smash{$Q_{y|x_1,\ldots,x_p} (\alpha) = \inf \{ t : \P(y \leq t | x_1,\ldots,x_p) \geq \alpha \}$}, for a given $\alpha \in [0,1]$, by solving the convex optimization problem:
\begin{equation*}
\minimize_{\theta} \; \sum_{i=1}^n \psi_\alpha \bigg( y^{(i)} - \sum_{j=1}^q \theta_j x_j^{(i)} \bigg),
\end{equation*}
where \smash{$\psi_\alpha(z) = \max \{ \alpha z, (\alpha - 1) z \}$} is the quantile loss (also referred to as the ``pinball'' or ``tilted absolute'' loss).
Quantile regression can be useful when the conditional distribution in question is suspected to be heteroskedastic and/or non-Gaussian, \eg, heavy-tailed, or if we wish to understand properties of the distribution other than the mean, \eg, tail behavior.  In multiple quantile regression, we solve several quantile regression problems simultaneously, each corresponding to a different quantile level; these problems can be coupled somehow to increase efficiency in estimation (see details in the next section). Again, the literature on quantile regression is quite vast (especially that from econometrics), and we only give a short review here. A standard text is \cite{koenker2005quantile}.  Nonparametric modeling of quantiles is a natural extension from the (linear) quantile regression approach outlined above; in the univariate case (one conditioning variable), \cite{koenker1994} suggested a method using smoothing splines, and \cite{takeuchi2006nonparametric} described an approach using kernels.  More recently, \cite{koenker2011} studied the multivariate nonparametric case (more than one conditioning variable), using additive models.  In the high-dimensional setting, where $p$ is large, \cite{belloni2011,kato2011,fan2014adaptive} studied $\ell_1$-penalized quantile regression and derived estimation and recovery theory
for non-(sub-)Gaussian data.  We extend results in \cite{fan2014adaptive} to prove structure recovery guarantees for the \mqgm~(in Section \ref{sec:recovery}). 

\section{The multiple quantile graphical model}
\label{sec:model}

Many choices can be made with regards to the final form of the \mqgm, and to help in understanding these options, we break down our presentation in parts.
First fix some ordered set $\mathcal{A}=\{\alpha_1,\ldots,\alpha_r\}$ of quantile levels, \eg, $\mathcal{A}=\{0.05, 0.10, \ldots, 0.95\}$.  For each variable $y_k$, and each level $\alpha_\ell$, we model the $\alpha_\ell$-conditional quantile given the other variables, using an additive expansion of the form:
\begin{equation}
Q_{y_k | y_{\neg k}}(\alpha_\ell) = b^*_{\ell k} + \sum_{j \neq k}^d f_{\ell k j}^*(y_j), 
\label{eq:invcdf}
\end{equation}
where \smash{$b_{\ell k}^* \in \R$} is an intercept term, and \smash{$f_{\ell k j}^*$}, $j=1,\ldots,d$ are smooth, but not parametric in form. 

\smallskip
\paragraph{Generic functional form of the \mqgm} In its most general form, the \mqgm~estimator is defined as a collection of optimization problems, over $k=1,\ldots,d$ and $\ell=1,\ldots,r$:
\begin{equation}
\label{eq:mqgm}
\minimize_{b_{\ell k}, \, f_{\ell k j} \in \cF_{\ell k j}, \, j=1,\ldots,d} \;
\sum_{i=1}^n \psi_{\alpha_\ell} \bigg( y_k^{(i)} - b_{\ell k} - \sum_{j \not= k} f_{\ell k j} (y_j^{(i)}) \bigg) +
\sum_{j \not= k} \Big( \lambda_1 P_1 (f_{\ell k j}) + 
\lambda_2 P_2 (f_{\ell k j}) \Big)^\omega.
\end{equation}
Here $\lambda_1,\lambda_2 \geq 0$ are tuning parameters.  Also, $\cF_{\ell k j}$, $j=1,\ldots,d$ are spaces of univariate functions, $\omega>0$ is a fixed exponent, and $P_1,P_2$ are sparsity and smoothness penalty functions, respectively, all to be decided as part of the modeling process. We give three examples below; several other variants are possible outside of what we describe.

\paragraph{Example 1: basis expansion model} Consider taking \smash{$\cF_{\ell k j}=\spann\{\phi_1^j,\ldots,\phi_m^j\}$}, the span of $m$ basis functions, \eg, radial basis functions (RBFs) with centers placed at appropriate locations across the domain of variable $j$, for each $j=1,\ldots,d$.  This means that each $f_{\ell k j} \in \cF_{\ell k j}$ can be expressed as \smash{$f_{\ell k j}(x) = \theta_{\ell k j}^T \phi^j(x)$}, for a coefficient vector \smash{$\theta_{\ell k j} \in \R^m$}, where
\smash{$\phi^j(x) = (\phi^j_1(x),\ldots,\phi^j_m(x))$}. Also consider an exponent $\omega=1$, and the sparsity and smoothness penalties
\begin{equation*}
P_1(f_{\ell k j}) = \|\theta_{\ell k j}\|_2 \quad \text{and} \quad 
P_2(f_{\ell k j}) = \|\theta_{\ell k j}\|_2^2,
\end{equation*}
respectively, which are group lasso and ridge penalties, respectively.  With these choices in place, the \mqgm~problem in \eqref{eq:mqgm} can be rewritten in finite-dimensional form:
\begin{equation}
\label{eq:mqgm_theta}
\minimize_{b_{\ell k}, \, \theta_{\ell k} = (\theta_{\ell k 1}, \ldots, 
\theta_{\ell k d})} \;
\psi_{\alpha_\ell} \Big( Y_k - b_{\ell k} \mathbf{1} - \Phi \theta_{\ell k} \Big) + 
\sum_{j \not= k} \Big( \lambda_1 \|\theta_{\ell k j}\|_2 +
\lambda_2 \|\theta_{\ell k j}\|_2^2 \Big).
\end{equation}
Above, we used the abbreviation \smash{$\psi_{\alpha_{\ell}}(z) = \sum_{i=1}^n \psi_{\alpha_\ell} (z_i)$} for a vector  \smash{$z=(z_1,\ldots,z_n) \in \R^n$}, and also \smash{$Y_k = (y_k^{(1)},\ldots,y_k^{(n)}) \in \R^n$} for the observations along variable $k$, \smash{$\mathbf{1} = (1,\ldots,1) \in \R^n$}, and \smash{$\Phi \in \R^{n \times dm}$} for the basis matrix, with blocks of columns \smash{$\Phi_{ij} = \phi(y_j^{(i)})^T \in \R^m$}.

The basis expansion model is simple and tends to works well in practice. For the majority of the paper, we will focus on this model; in principle, everything that follows (methodologically, theoretically, algorithmically) extends to the next two models we describe, as well as many other variants.

\smallskip
\paragraph{Example 2: smoothing splines model} Now consider taking \smash{$\cF_{\ell k j}=\spann\{g_1^j,\ldots,g_n^j\}$}, the span of $m=n$ natural cubic splines with knots at \smash{$y_j^{(1)},\ldots,y_j^{(n)}$}, for $j=1,\ldots,d$.  As before, we can then write \smash{$f_{\ell k j}(x) = \theta_{\ell k j}^T g^j(x)$} with coefficients \smash{$\theta_{\ell k j} \in \R^n$}, for $f_{\ell k j} \in \cF_{\ell k j}$.  The work of \cite{meier2009highdimensional}, on high-dimensional additive smoothing splines, suggests a choice of exponent $\omega=1/2$, and penalties
\begin{equation*}
P_1(f_{\ell k j}) = \|G^j \theta_{\ell k j}\|_2^2 \quad \text{and} \quad 
P_2(f_{\ell k j}) = \theta_{\ell k j}^T \Omega^j \theta_{\ell k j},
\end{equation*}
for sparsity and smoothness,
respectively, where \smash{$G^j \in \R^{n\times n}$} is a spline basis matrix with entries \smash{$G^j_{i i'} = g^j_{i'}(y_j^{(i)})$}, and \smash{$\Omega^j$} is the smoothing spline penalty matrix containing integrated products of pairs of twice differentiated basis functions.  The \mqgm~problem in \eqref{eq:mqgm} can be translated into a finite-dimensional form, very similar to what we have done in \eqref{eq:mqgm_theta}, but we omit this for brevity.

\smallskip
\paragraph{Example 3: RKHS model} Consider taking $\cF_{\ell k j}= \mathcal{H}_j$, a univariate reproducing kernel Hilbert space (RKHS), with kernel function \smash{$\kappa^j(\cdot,\cdot)$}.  The representer theorem allows us to express each function $f_{\ell k j} \in \mathcal{H}_j$ in terms of the representers of evaluation, \ie, \smash{$f_{\ell k j}(x) = \sum_{i=1}^n (\theta_{\ell k j})_i \kappa^j(x,y^{(i)}_j)$}, for a coefficient vector \smash{$\theta_{\ell k j} \in \R^n$}.  The work of \cite{raskutti2012minimax}, on high-dimensional additive RKHS modeling, suggests a choice of exponent $\omega=1$, and sparsity and smoothness penalties
\begin{equation*}
P_1(f_{\ell k j}) = \|K^j \theta_{\ell k j}\|_2 \quad \text{and} \quad 
P_2(f_{\ell k j}) = \sqrt{\theta_{\ell k j}^T K^j \theta_{\ell k j}},
\end{equation*}
respectively, where \smash{$K^j \in \R^{n\times n}$} is the kernel matrix with entries \smash{$K^j_{i i'} = \kappa^j(y_j^{(i)}, y_j^{(i')})$}. Again, the \mqgm~problem in \eqref{eq:mqgm} can be written in finite-dimensional form, now an SDP, omitted for brevity.

\paragraph{Structural constraints} Different kinds of structural constraints can be placed on top of the \mqgm~optimization problem in order to guide the estimated component functions to meet particular shape requirements.  An important example are {\it non-crossing constraints} (commonplace in nonparametric, multiple quantile regression \citep{koenker2005quantile,takeuchi2006nonparametric}): here, we optimize \eqref{eq:mqgm} jointly over $\ell=1,\ldots,r$, subject to
\begin{equation}
\label{eq:noncrossing}
b_{\ell k} + \sum_{j\not=k} f_{\ell k j}(y_j^{(i)}) \leq 
b_{\ell' k} + \sum_{j\not=k} f_{\ell' k j}(y_j^{(i)}),
\quad \text{for all $\alpha_\ell < \alpha_{\ell'}$, and $i=1,\ldots,n$}.
\end{equation}
This ensures that the estimated quantiles obey the proper ordering, at the observations. For concreteness, we consider the implications for the basis regression model, in Example 1 (similar statements hold for the other two models). For each $\ell=1,\ldots,r$, denote by \smash{$F_{\ell k}(b_{\ell k}, \theta_{\ell k})$} the criterion in \eqref{eq:mqgm_theta}.  Introducing the non-crossing constraints requires coupling \eqref{eq:mqgm_theta} over $\ell=1,\ldots,r$, so that we now have the following optimization problems, for each target variable $k=1,\ldots,d$:
\begin{equation}
\minimize_{B_k, \Theta_k} \; 
\sum_{\ell=1}^r F_{\ell k} (b_{\ell k}, \theta_{\ell k}) 
\quad \subjectto \quad (\mathbf{1} B_k^T + \Phi \Theta_k) D^T \geq 0, 
\label{eq:mqgm_theta_nc}
\end{equation}
where we denote \smash{$B_k=(b_{1k},\ldots,b_{rk}) \in \R^r$}, \smash{$\Phi \in \R^{n \times dm}$} the basis matrix as before, \smash{$\Theta_k \in \R^{dm \times r}$} given by column-stacking $\theta_{\ell k} \in \R^{dm}$, $\ell=1,\ldots,r$, and $D \in \reals^{(r-1) \times r}$ is the usual discrete difference operator, \ie,
\[
D = \left[\begin{array}{rrrrr}
-1 & 1 & 0 & \ldots &  0\\
0 & -1 & 1 & \ldots&  0\\
\vdots &  & \ddots & \ddots & \\ 
0 & 0 & \ldots & -1 & 1 
\end{array}\right].
\]
(The inequality in \eqref{eq:mqgm_theta_nc} is to be interpreted componentwise.) Computationally, coupling the subproblems across $\ell=1,\ldots,r$ clearly adds to the overall difficulty of the \mqgm, but statistically this coupling acts as a regularizer, by constraining the parameter space in a useful way, thus increasing our efficiency in fitting multiple quantile levels from the given data.

For a triplet $\ell,k,j$, {\it monotonicity constraints} are easy to add, \ie, \smash{$f_{\ell k j}(y_j^{(i)}) \leq  f_{\ell k j}(y_j^{(i')})$} for all \smash{$y_j^{(i)} < y_j^{(i')}$}. {\it Convexity constraints}, where we require \smash{$f_{\ell k j}$} to be convex over the observations, for a particular $\ell,k,j$, are also straightforward.  Lastly, {\it strong non-crossing constraints}, where we enforce \eqref{eq:noncrossing} but over all inputs \smash{$z \in \R^d$} (not just over the observations) are also possible with positive basis functions.

\paragraph{Exogenous variables and conditional random fields} So far, we have considered modeling the joint distribution $\P(y_1,\ldots,y_d)$, corresponding to learning a Markov random field (MRF).  It is not hard to extend our framework to model the conditional distribution $\prob(y_1,\ldots,y_d|x_1,\ldots,x_p)$ given some exogenous variables $x_1,\ldots,x_p$, corresponding to learning a conditional random field (CRF).  To extend the basis regression model, we introduce the additional parameters \smash{$\theta^x_{\ell k} \in \R^p$} in \eqref{eq:mqgm_theta}, and the loss now becomes \smash{$\psi_{\alpha_\ell} ( Y_k - b_{\ell k} \mathbf{1}^T - \Phi \theta_{\ell k} - X \theta^x_{\ell k})$}, where \smash{$X \in \reals^{n \times q}$} is filled with the exogenous observations \smash{$x^{(1)},\ldots,x^{(n)} \in \reals^q$}; the other models are changed similarly. 

\section{Basic properties and theory}
\label{sec:theory}


\subsection{Quantiles and conditional independence}
\label{sec:indep}

In the model \eqref{eq:invcdf}, if a particular variable $y_j$ has no contribution, \ie, satisfied \smash{$f_{\ell k j}^*=0$} across all quantile levels $\alpha_\ell$, $\ell=1,\ldots,r$, what does this imply about the conditional independence between $y_k$ and $y_j$, given the rest?  Outside of the multivariate normal model (where the feature transformations need only be linear), nothing can be said in generality.  But we argue that conditional independence can be understood in a certain {\it approximate sense} (\ie, in a projected approximation of the data generating model).  We begin with a simple lemma. Its proof is elementary, and given in the supplement.

\begin{lemma}
\label{lem:indep}
Let $U,V,W$ be random variables, and suppose that all conditional quantiles of $U|V,W$ do not depend on $V$, \ie, $Q_{U|V,W}(\alpha)=Q_{U|W}(\alpha)$ for all $\alpha \in [0,1]$.  Then $U$ and $V$ are conditionally independent given $W$. 
\end{lemma}

By the lemma, if we knew that $Q_{U|V,W}(\alpha)=h(\alpha,U,W)$ for a function $h$, then it would follow that $U,V$ are conditionally independent given $W$ (\nb, the converse is true, as well).  The \mqgm~problem in \eqref{eq:mqgm}, with sparsity imposed on the coefficients, essentially aims to achieve such a representation for the conditional quantiles; of course we cannot use a {\it fully nonparametric} representation of the conditional distribution $y_k | y_{\neg k}$ and instead we use an {\it $r$-step approximation} to the conditional cumulative distribution function (CDF) of $y_k | y_{\neg k}$ (corresponding to estimating $r$ conditional quantiles), and (say) in the basis regression model, limit the dependence on conditioning variables to be in terms of an additive function of RBFs in $y_j$, $j \neq k$. Thus, if at the solution in \eqref{eq:mqgm_theta} we find that \smash{$\hat\theta_{kj\ell}=0$}, $\ell=1,\ldots,r$, we may interpret this to mean that $y_k$ and $y_j$ are conditionally independent given the remaining variables, but according to the distribution defined by the {\it projection} of $y_k|y_{\neg k}$ onto the space of models considered in \eqref{eq:mqgm_theta} ($r$-step conditional CDFs, which are additive expansions in $y_j$, $j \neq k$).  This interpretation is no more tenuous (arguably, less so, as the model space here is much larger) than that needed when applying standard neighborhood selection to non-Gaussian data.



\subsection{Gibbs sampling and the ``joint'' distribution}
\label{sec:joint}

When specifying a form for the conditional distributions in a pseudolikelihood approximation as in \eqref{eq:pseudo}, it is natural to ask: what is the corresponding joint distribution?  Unfortunately, for a general collection of conditional distributions, there need not exist a compatible joint distribution, even when all conditionals are continuous \citep{wang2008}.  Still, pseudolikelihood approximations (a special case of composite likelihood approximations), possess solid theoretical backing, in that maximizing the pseudolikelihood relates closely to minimizing a certain (expected composite) Kullback-Leibler divergence, measured to the true conditionals \citep{varin2005}.  Recently, \cite{chen2015,yang2015} made nice progress in describing specific conditions on conditional distributions that give rise to a valid joint distribution, though their work was specific to exponential families.  A practical answer to the question of this subsection is to use Gibbs sampling, which attempts to draw samples consistent with the fitted conditionals; this is precisely the observation of \cite{heckerman2000dependency}, who show that Gibbs sampling from discrete conditionals converges to a unique stationary distribution, although this distribution may not actually be compatible with the conditionals.  The following result establishes the analogous claim for continuous conditionals; its proof is in the supplement. We demonstrate the practical value of Gibbs sampling through various examples in Section \ref{sec:exps}.


\begin{lemma}
\label{lem:gibbs}
Assume that the conditional distributions $\prob(y_k | y_{\neg k})$, $k=1,\ldots,d$ take only positive values on their domain.  Then, for any given ordering of the variables, Gibbs sampling converges to a unique stationary distribution that can be reached from any initial point. (This stationary distribution depends on the ordering.)
\end{lemma}

\subsection{Graph structure recovery}
\label{sec:recovery}

When \smash{$\log{d}=O(n^{2/21})$}, and we assume somewhat standard regularity conditions (listed as A1--A4 in the supplement), we will show that the \mqgm~estimate recovers the underlying conditional independencies with high probability (interpreted in the projected model space, as explained in Section \ref{sec:indep}).  Importantly, we do not require a Gaussian, sub-Gaussian, or even parametric assumption on the data generating process; instead, we assume i.i.d.\ draws \smash{$y^{(1)},\ldots,y^{(n)} \in \reals^d$}, where the conditional distributions $y_k|y_{\neg k}$ have quantiles that are specified by the model in \eqref{eq:invcdf} for $k=1,\ldots,d$, $\ell=1,\ldots,r$, and further, each \smash{$f_{\ell k j}^*(x) = \theta_{\ell k j}^T \phi^j(x)^*$} for coefficients \smash{$\theta_{\ell k j}^* \in \R^m$}, $j=1,\ldots,d$, as in the basis expansion model.  

Let $E^*$ denote the corresponding edge set of conditional dependencies from these neighborhood models, \ie, 
\smash{$\{k,j\} \in E^* \iff \max_{\ell=1,\ldots,r} \max\{\|\theta_{\ell kj}^*\|_2, |\theta_{\ell jk}^*\|_2\}> 0$}. We define the estimated edge set \smash{$\hat{E}$} in the analogous way, based on the solution in \eqref{eq:mqgm_theta}. Without a loss of generality, we assume the features have been scaled to satisfy $\|\Phi_j\| \leq \sqrt{n}$ for $j=1,\ldots,dm$. The following is our recovery result; its proof is provided in the supplement.

\begin{theorem}
\label{thm:recovery}
Assume \smash{$\log{d}=O(n^{2/21})$}, and conditions A1--A4 in the supplement. Assume that the tuning parameters $\lambda_1,\lambda_2$ satisfy
\begin{equation*}
\lambda_1 \asymp \sqrt{mn \log( d^2mr / \delta) \log^3 n} 
\quad \text{and} \quad
\lambda_2 = o(n^{41/42} / \theta^*_{\max}),
\end{equation*}
where \smash{$\theta^*_{\max}=\max_{\ell,k,j} \|\theta^*_{\ell k j}\|_2$}. Then for $n$ large enough, the  \mqgm~estimate in \eqref{eq:mqgm_theta} exactly recovers the underlying conditional dependencies, \ie, \smash{$\hat{E}=E^*$}, with probability at least $1-\delta$.
\end{theorem}

The theorem shows that the nonzero pattern in the \mqgm~estimate identifies, with high probability, the underlying conditional independencies. But to be clear, we emphasize that the \mqgm~estimate is {\it not} an estimate of the inverse covariance matrix itself (this is also the case with neighborhood regression, \spacejam~of \cite{voorman2014}, and many other methods for learning graphical models).

\section{Computational approach}
\label{sec:algos}
By design, the \mqgm~problem in \eqref{eq:mqgm_theta} separates into $d$ 
subproblems, across $k=1,\ldots,d$ (it therefore suffices to consider only a
single subproblem, so we omit notational dependence on $k$ for auxiliary variables). While these subproblems are
challenging for off-the-shelf solvers (even for only moderately-sized
graphs), the key terms here all admit efficient
{\it proximal operators} \citep{parikh2013proximal}, which makes
operator splitting methods like the alternating direction method of
multipliers \citep{boyd2011distributed} a natural choice.  As an illustration, we consider the non-crossing constraints in the basis regression model below.
Reparameterizing so that we may apply ADMM:
\begin{equation}
\begin{array}{ll}
\minimizewrt{\Theta_k,B_k,V,W,Z} & 
\psi_{\mathcal{A}}(Z) + \lambda_1 \sum_{\ell=1}^r
\sum_{j=1}^d \|W_{\ell j}\|_2 + \frac{\lambda_2}{2} \|W\|_F^2 + I_+(V D^T) \smallskip \\ 
\subjectto & V = \mathbf{1} B_k^T + \Phi \Theta_k , \; W =
\Theta_k, \; Z = Y_k \mathbf{1}^T - \mathbf{1} B_k^T 
- \Phi \Theta_k,
\end{array}
\label{eq:admmprob}
\end{equation}
where for brevity \smash{$\psi_{\mathcal{A}}(A) = \sum_{\ell=1}^{r} \sum_{j=1}^d \psi_{\alpha_\ell}(A_{\ell j})$}, and $I_+(\cdot)$ is the indicator function of the space of elementwise nonnegative matrices. The augmented Lagrangian associated with \eqref{eq:admmprob} is:  
\begin{equation}
\begin{aligned}
&L_\rho(\Theta_k, B_k, V, W, Z, U_V, U_W, U_Z) = \psi_{\mathcal{A}}(Z) +
\lambda_1 \sum_{\ell=1}^r \sum_{j=1}^d \|W_{\ell j}\|_2  + 
\frac{\lambda_2}{2} \|W\|_F^2 +  I_+(VD^T) \\ 
&\;+\, \frac{\rho}{2} \Big(\|\mathbf{1} B_k^T +
\Phi \Theta_k - V + U_V\|_F^2 + 
\|\Theta_k - W + U_W\|_F^2 + \|Y_k \mathbf{1}^T - \mathbf{1} B_k^T     
- \Phi \Theta_k -Z + U_Z\|_F^2 \Big),
\end{aligned}
\label{eq:auglag} 
\end{equation}
where $\rho > 0$ is the augmented Lagrangian parameter, and
$U_V,U_W,U_Z$ are dual variables corresponding to the
equality constraints on $V,W,Z$, respectively.  Minimizing \eqref{eq:auglag} over $V$ yields:
\begin{equation}
\label{eq:V}
V \leftarrow P_{\textrm{iso}}\left(\mathbf{1} B_k^T + 
\Phi \Theta_k + U_V\right),
\end{equation}
where $P_{\textrm{iso}}(\cdot)$ denotes the row-wise projection operator onto the isotonic cone (the space of componentwise nondecreasing vectors), an $O(nr)$ operation here \citep{nickdp}.
Minimizing \eqref{eq:auglag} over $W_{\ell j}$ yields the update: 
\begin{equation}
\label{eq:W}
W_{\ell j} \leftarrow \frac{(\Theta_k)_{\ell j} + (U_W)_{\ell j}}{1+\lambda_2/\rho} \bigg(1-\frac{\lambda_1/\rho}{\|(\Theta_k)_{\ell j} + (U_W)_{\ell j}\|_2} \bigg)_+,
\end{equation}
where $(\cdot)_+$ is the positive part operator. This
can be seen by deriving the proximal operator of the function 
$f(x)=\lambda_1 \|x\|_2+(\lambda_2/2)\|x\|_2^2$. 
Minimizing \eqref{eq:auglag} over $Z$ yields the update:
\begin{equation}
\label{eq:Z}
Z \leftarrow \prox_{(1/\rho)\psi_{\mathcal{A}}}
(Y_k \mathbf{1}^T - \mathbf{1} b_k^T     
- \Phi \Theta_k + U_Z),
\end{equation}
where $\prox_f(\cdot)$ denotes the proximal operator of
a function $f$.  For the multiple quantile loss function 
\smash{$\psi_{\mathcal{A}}$}, this is
a kind of generalized soft-thresholding. The proof is given in
the supplement. 
\begin{lemma}
\label{lem:prox}
Let $P_+(\cdot)$ and $P_-(\cdot)$ be the elementwise positive and negative part operators, respectively,
and let $a = (\alpha_1,\ldots,\alpha_r)$. Then \smash{$\prox_{t\psi_{\mathcal{A}}}(A) = P_+(A - t \mathbf{1} a^T) + P_-(A - t \mathbf{1} a^T).$}
\end{lemma}
Finally, differentiation in \eqref{eq:auglag} with respect to $B_k$ and $\Theta_k$ yields the simultaneous updates:
\begin{equation}
\begin{split}
\left [\begin{array}{c}\Theta_k \\ B_k^T \end{array} \right] 
& \leftarrow 
\frac{1}{2}
\left[
\begin{array}{cc}
\Phi^T \Phi + \frac{1}{2} I & \Phi^T \ones \\
\ones^T \Phi & \ones^T \ones
\end{array}
\right]^{-1} \left (
\left [ I \; \mathbf{0} \right ]^T(W - U_W) +
\left [ \Phi \; \mathbf{1}  \right ]^T(Y_k 
\mathbf{1}^T - Z + U_Z + V - U_V) \right ). 
\end{split}
\label{eq:theta}
\end{equation}

A complete description of our ADMM algorithm for solving the \mqgm~problem is in the supplement.

\paragraph{Gibbs sampling}
Having fit the conditionals $y_k | y_{\neg k}$,
$k=1,\ldots d$, we may want to make predictions 
or extract joint distributions over subsets of variables.  As
discussed in Section \ref{sec:joint}, there is no general analytic
form for these joint distributions, but the pseudolikelihood
approximation  
underlying the \mqgm~suggests a natural Gibbs sampler.  A careful
implementation that respects the additive model in \eqref{eq:invcdf}
yields a highly efficient Gibbs sampler, especially for CRFs; the 
supplement gives details.

\section{Empirical examples}
\label{sec:exps}

\subsection{Synthetic data}
\label{sec:synth}
We consider synthetic examples, comparing to neighborhood selection (\mb), the graphical lasso (\glasso), \spacejam~\cite{voorman2014}, the \nonparas~\cite{liu2012high}, \tiger~\cite{liu2012tiger}, and neighborhood selection using the absolute loss (\laplace).

\paragraph{Ring example}
As a simple but telling example, we drew $n=400$ samples from a ``ring'' distribution in $d=4$ dimensions.  We used $m=10$ expanded features and $r=20$ quantile levels. Data were generated by first drawing a random angle $\nu \sim \textrm{Uniform}(0,1)$, then a random radius $R \sim \mathcal{N}(0,0.1)$, and finally computing the coordinates $y_1 = R \cos \nu$, $y_2 = R \sin \nu$ and $y_3, y_4 \sim \mathcal{N}(0,1)$, \ie, $y_1$ and $y_2$ are the only dependent variables here.
Figure \ref{fig:ring:dist:left} plots samples (blue) of the coordinates $(y_1,y_2)$ as well as new samples (red) from the \mqgm, \mb, \glasso, and \spacejam~fitted to these same (blue) samples; the samples from the \mqgm, obtained by using our Gibbs sampler (see the supplement for further details), appear to closely match the samples from the underlying ring.
\begin{figure}
\centering
\includegraphics[width=0.3\textwidth]{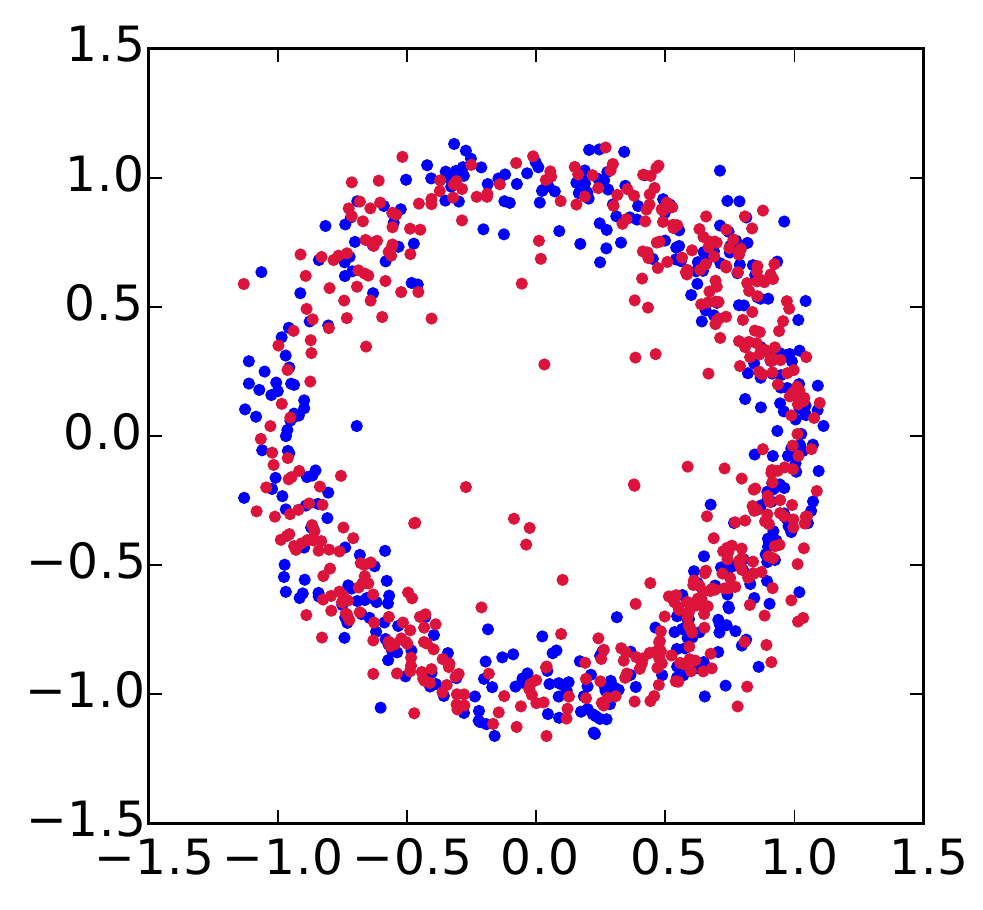}
\hfill
\includegraphics[width=0.3\textwidth]{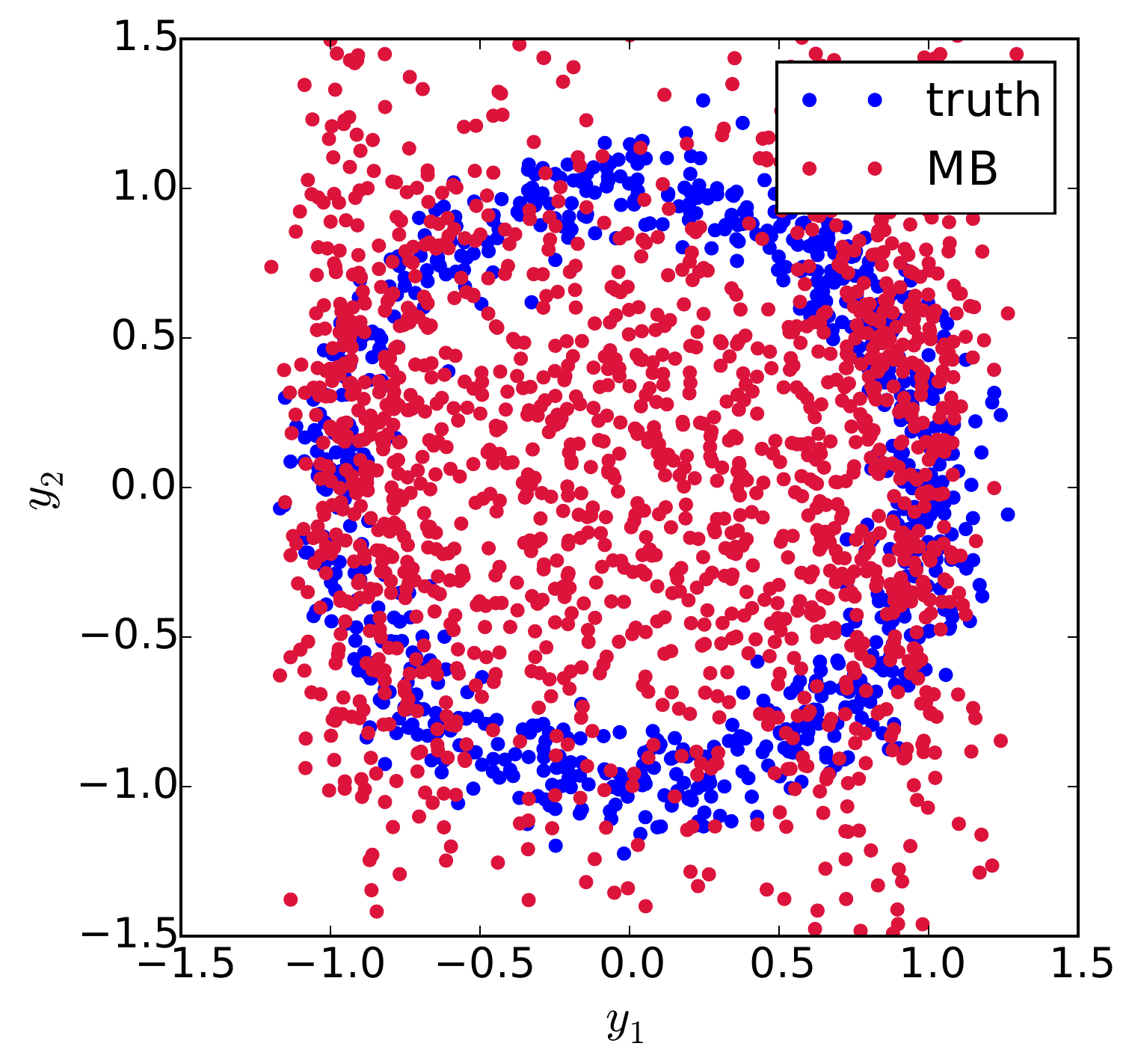}
\hfill
\includegraphics[width=0.3\textwidth]{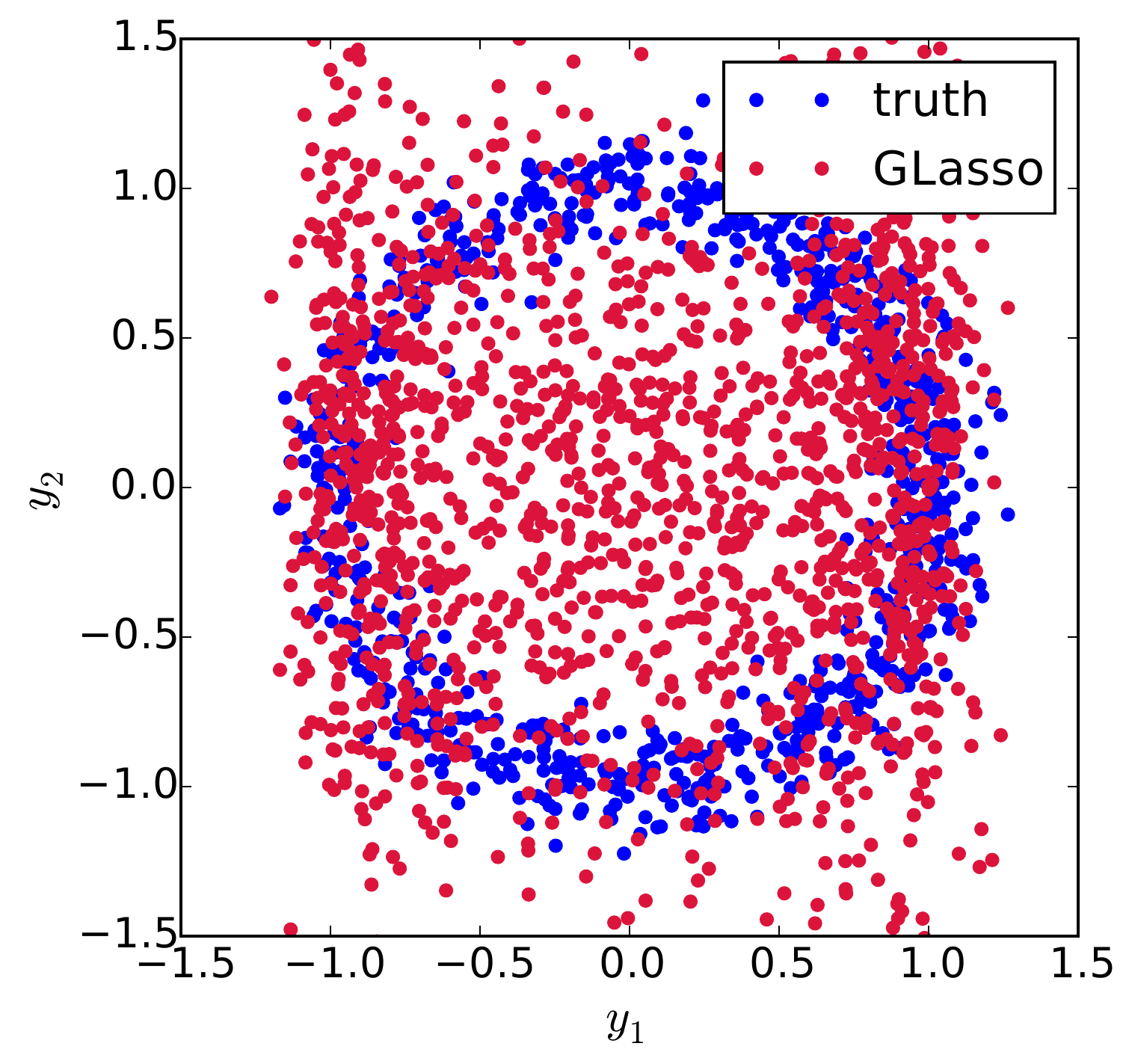} \\
\includegraphics[width=0.3\textwidth]{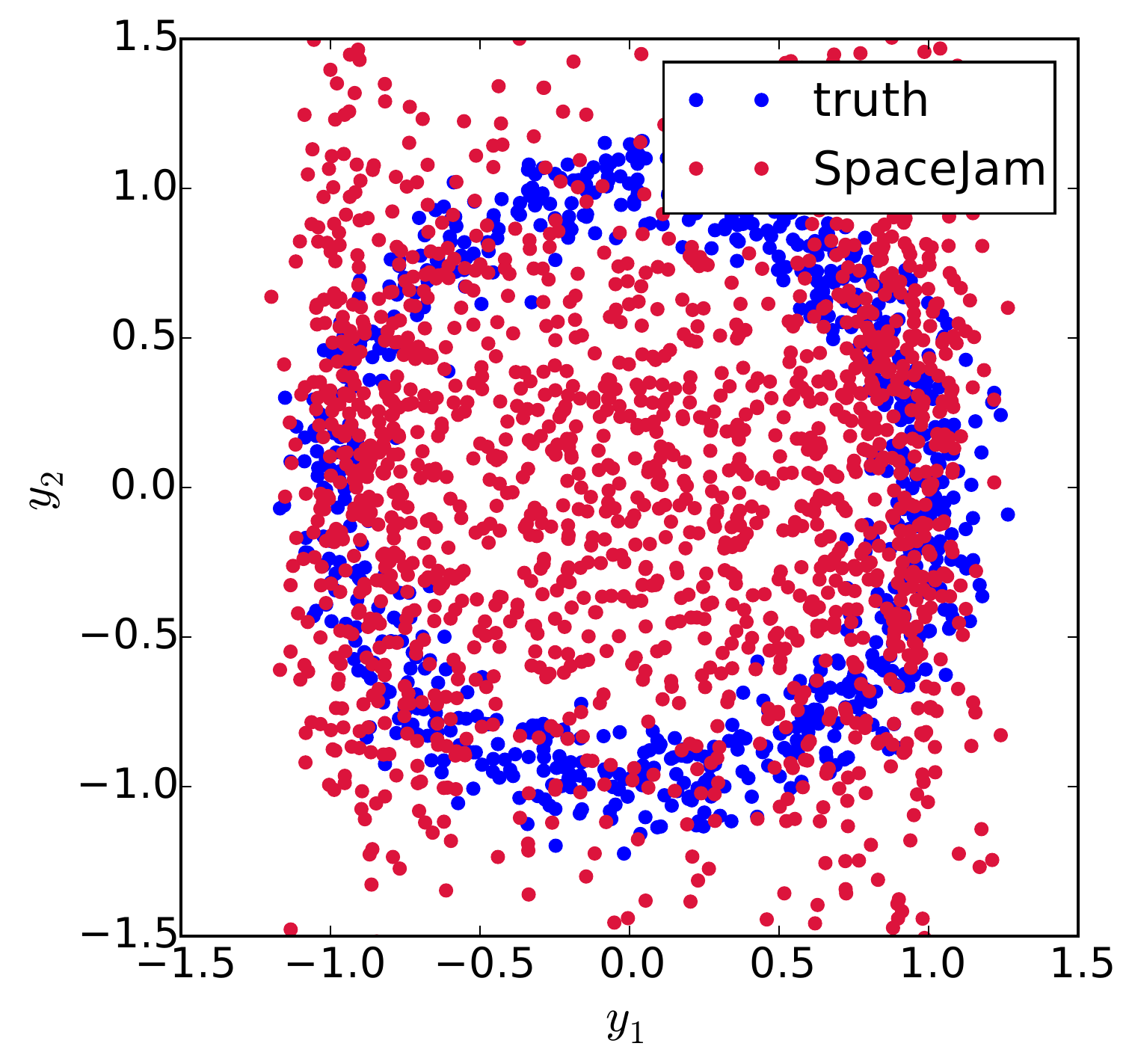}
\hfill
\includegraphics[width=0.3\textwidth]{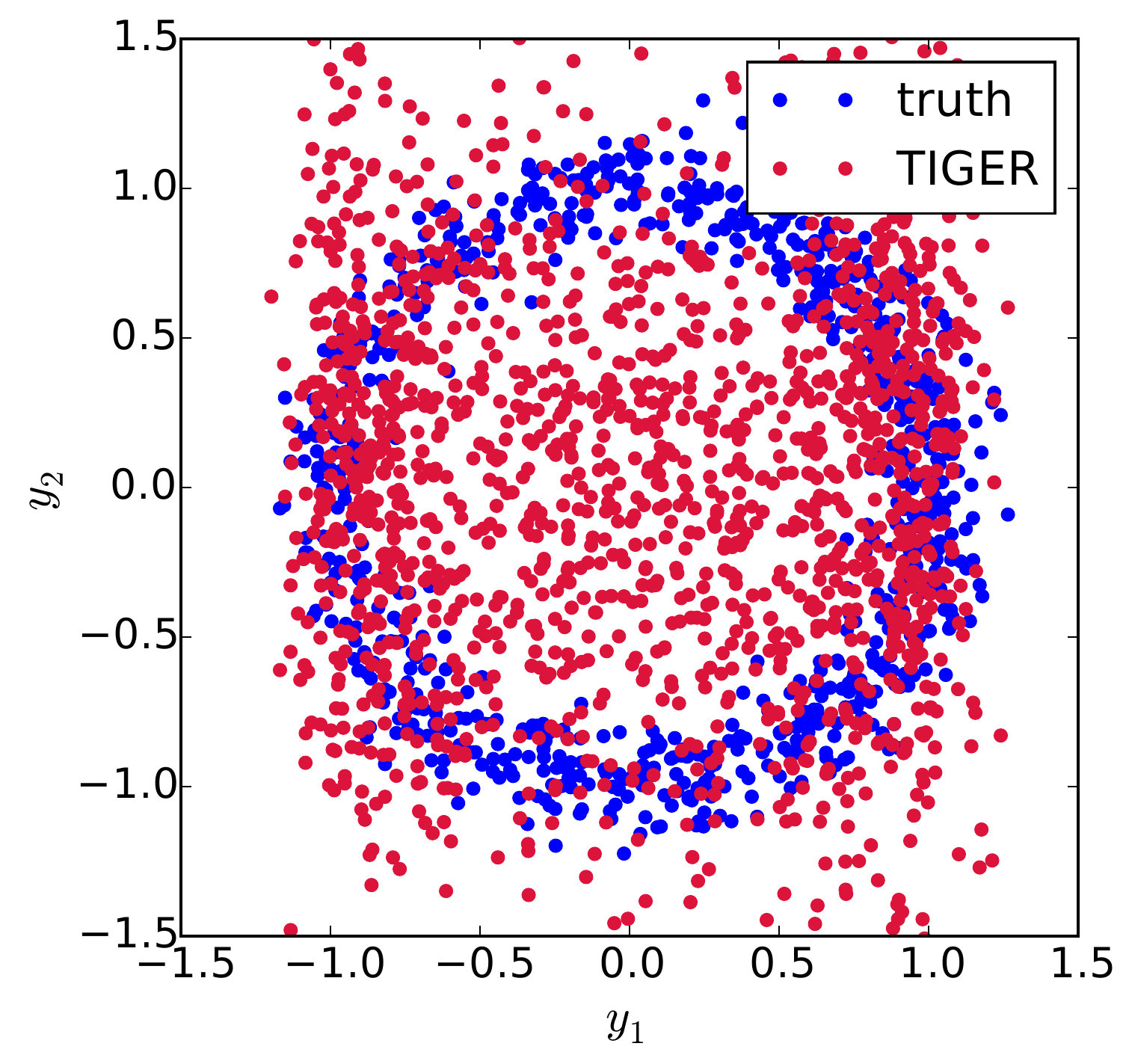}
\hfill
\includegraphics[width=0.3\textwidth]{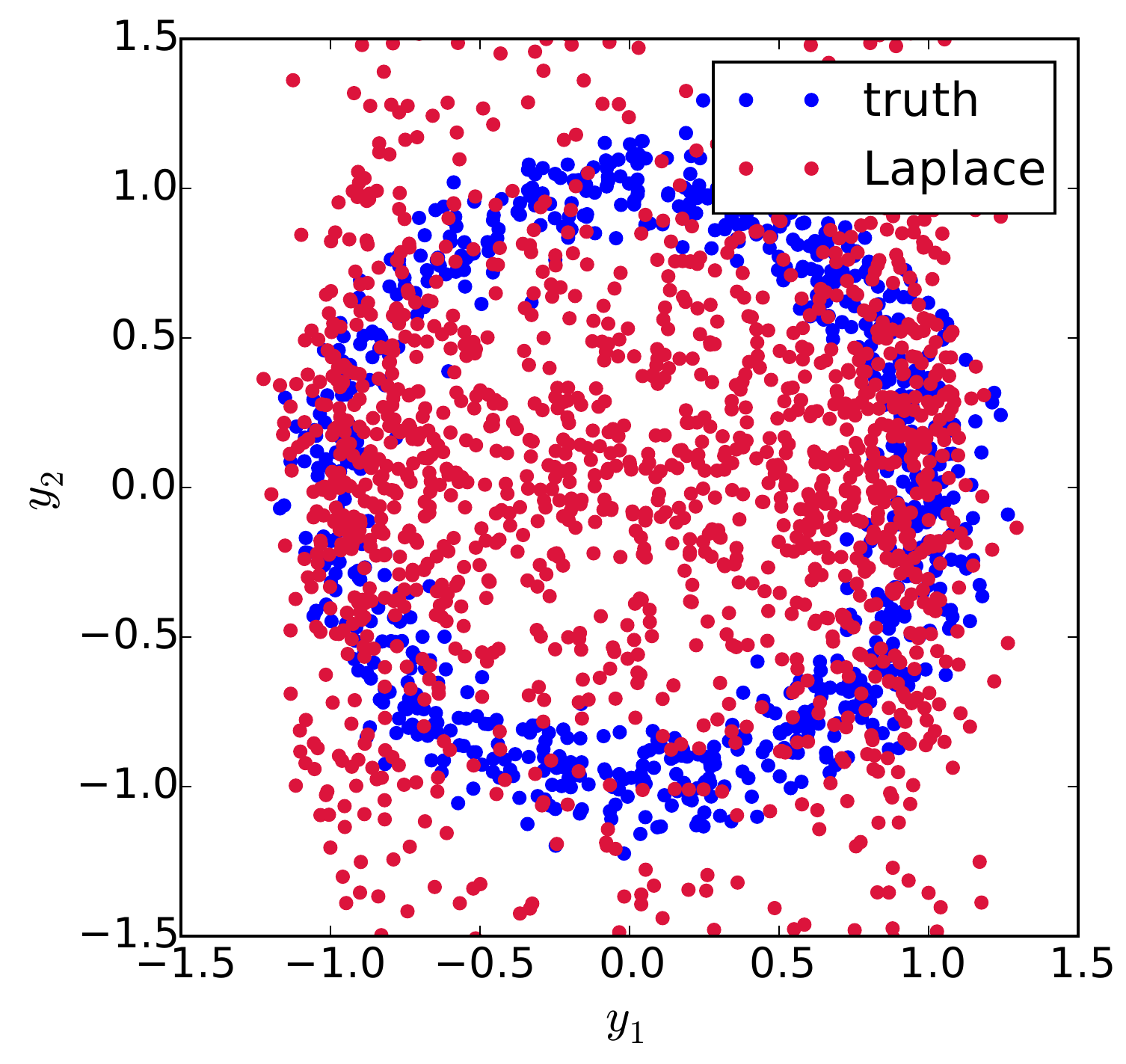}
\caption{Data from the ring distribution (blue) as well as new samples (red) from the \mqgm, \mb, \glasso, \spacejam, \tiger, and \laplace~fitted to the same (blue) data; the samples from the \mqgm~were obtained by using our Gibbs sampler.}
\label{fig:ring:dist:left}
\end{figure}

Figure \ref{fig:ring:dist:right} shows the conditional independencies recovered by the \mqgm, \mb, \glasso, \spacejam, the \nonparas, \tiger, and \laplace, when run on the ring data. We visualize these independencies by forming a $d \times d$ matrix with the cell $(j,k)$ set to white if $j,k$ are conditionally independent given the others, and black otherwise.  Across a range of tuning parameters for each method,
the \mqgm~is the only one that successfully recovers the underlying conditional dependencies.
\begin{figure}[t!]
\centering
\includegraphics[width=0.6\textwidth]{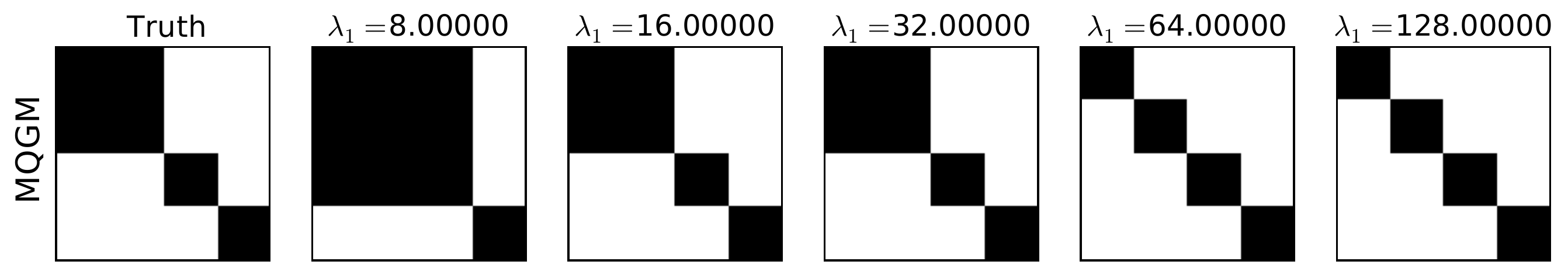} \\
\includegraphics[width=0.6\textwidth]{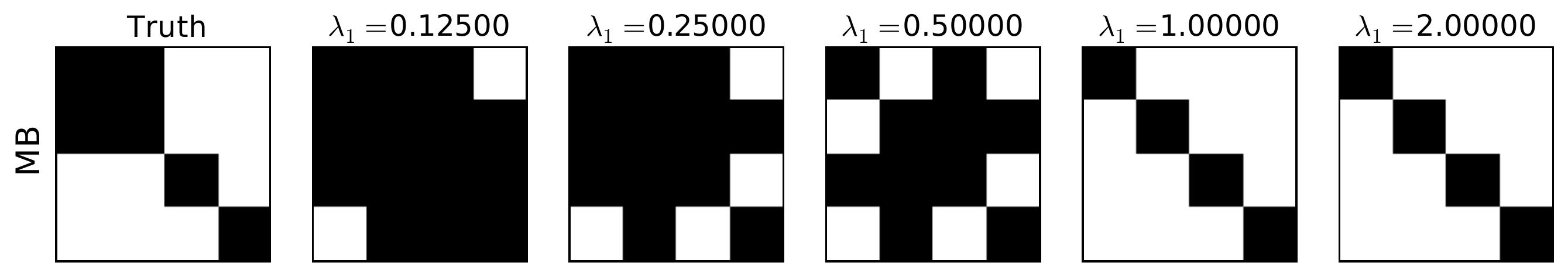} \\
\includegraphics[width=0.6\textwidth]{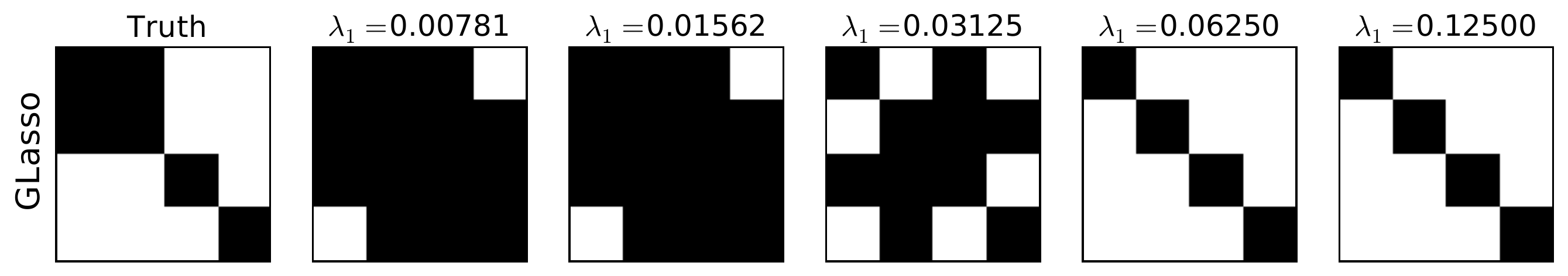} \\
\includegraphics[width=0.6\textwidth]{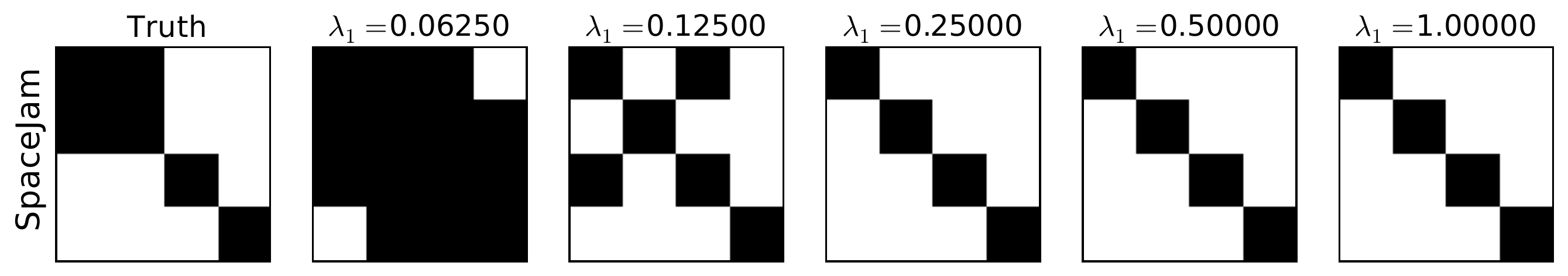} \\
\includegraphics[width=0.6\textwidth]{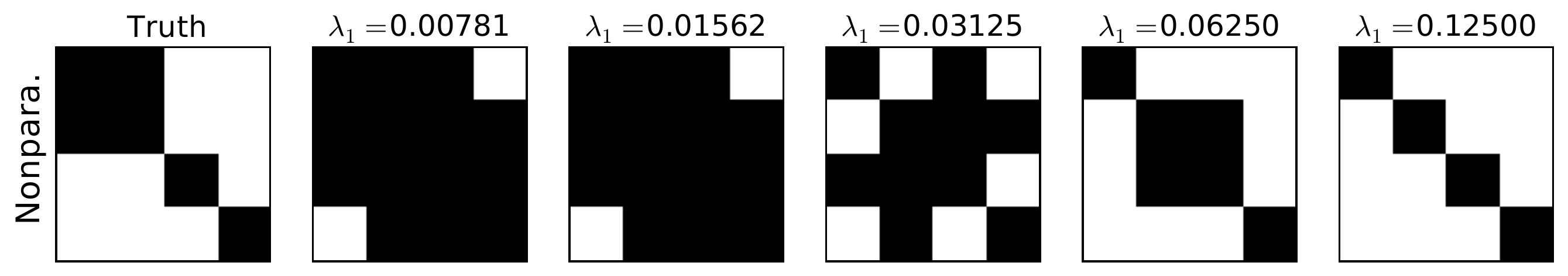} \\
\includegraphics[width=0.6\textwidth]{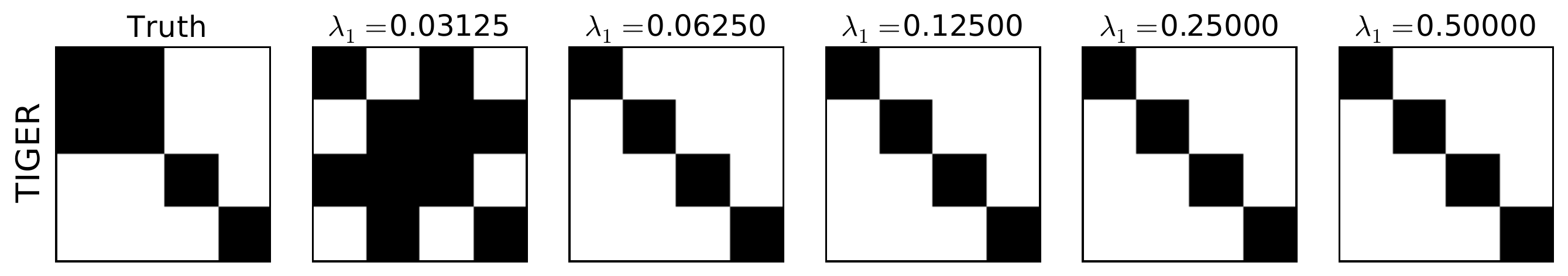} \\
\includegraphics[width=0.6\textwidth]{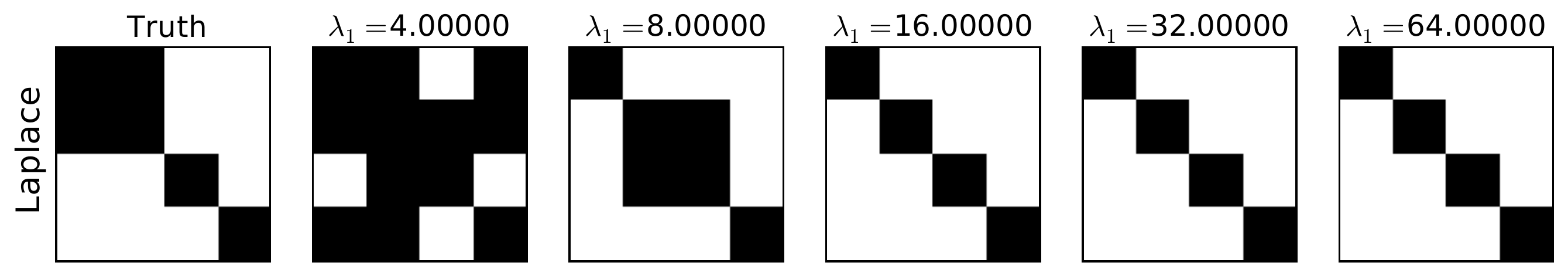}
\caption{Conditional independencies recovered by the \mqgm, \mb, \glasso, \spacejam, the \nonparas, \tiger, and \laplace~on the ring data; black means conditional dependence.  The \mqgm~is the only method that successfully recovers the underlying conditional dependencies.}
\label{fig:ring:dist:right}
\end{figure}

Table \ref{tab:calib} presents an evaluation of the conditional CDFs given by the \mqgm, \mb, \glasso, \spacejam, \tiger, and \laplace~when run on the ring data.  For each method, we averaged the total variation distances and Kolmogorov-Smirnoff statistics between the fitted and true conditional CDFs across all variables, and then reported the best values obtained across a range of tuning parameters (further details are given in the supplement); the \mqgm~outperforms all its competitors, in both metrics.
\begin{table}[]
\caption{Total variation (TV) distance and Kolmogorov-Smirnoff (KS) statistic values for the \mqgm, \mb, \glasso, \spacejam, \tiger, and \laplace~on the ring data; lower is better, best in {\bf bold}.}
\label{tab:calib}
\centering
\begin{tabular}{lcc}
& TV & KS \\
\mqgm & \textbf{20.873} & \textbf{0.760}\\
\mb & 92.298 & 1.856\\
\glasso & 92.479 & 1.768\\
\spacejam & 91.568 & 1.697\\
\tiger & 88.284 & 1.450\\
\laplace & 127.406 & 1.768\\

\end{tabular}
\end{table}

\paragraph{Larger examples}
To investigate performance at larger scales, we drew $n \in \{50, 100, 300\}$ samples from a multivariate normal and Student $t$-distribution (with 3 degrees of freedom), both in $d=100$ dimensions, and parameterized by a random, sparse, diagonally dominant $d \times d$ inverse covariance matrix, following the procedure in \cite{peng2009partial,khare2014convex,NIPS2014_5576,alnur}. 
Over the same set of sample sizes, with $d=100$, we also considered an autoregressive setup in which we drew samples of pairs of adjacent variables from the ring distribution.
In all three data settings (normal, $t$, and autoregressive), we used $m=10$ and $r=20$ for the \mqgm. To summarize the performances, we considered a range of tuning parameters for each method, computed corresponding false and true positive rates (in detecting conditional dependencies),
and then computed the corresponding area under the curve (AUC), following, \eg, \cite{peng2009partial,khare2014convex,NIPS2014_5576,alnur}. 
Table \ref{tab:synth} reports the median
AUCs (across 50 trials) for all three of these examples; the \mqgm~outperforms all other methods on the autoregressive example, as well as on the small-$n$ normal and Student $t$ examples.
\begin{table}
\centering
\caption{AUC values for the \mqgm, \mb, \glasso, \spacejam, the \nonparas, \tiger, and \laplace~for the normal, $t$, and autoregressive data settings; higher is better, best in {\bf bold} (standard errors are $\approx 10^{-4}$ or smaller).}
\label{tab:synth}
\begin{scriptsize}
\begin{tabular}{lcccccccccc}
& \multicolumn{3}{c}{Normal} & \multicolumn{3}{c}{Student $t$} & \multicolumn{3}{c}{Autoregressive} \\
& $n = 50$ & $n = 100$ & $n = 300$ & $n = 50$ & $n = 100$ & $n = 300$ & $n = 50$ & $n = 100$ & $n = 300$ \\
\mqgm & \textbf{0.953} & \textbf{0.976} & 0.988 & \textbf{0.928} & 0.947 & 0.981 & \textbf{0.726} & \textbf{0.754} & \textbf{0.955}\\
\mb & 0.850 & 0.959 & 0.994 & 0.844 & 0.923 & 0.988 & 0.532 & 0.563 & 0.725\\
\glasso & 0.908 & 0.964 & \textbf{0.998} & 0.691 & 0.605 & 0.965 & 0.541 & 0.620 & 0.711\\
\spacejam & 0.889 & 0.968 & 0.997 & 0.893 & \textbf{0.965} & 0.993 & 0.624 & 0.708 & 0.854\\
\nonpara & 0.881 & 0.962 & 0.996 & 0.862 & 0.942 & \textbf{0.998} & 0.545 & 0.590 & 0.612\\
\tiger & 0.732 & 0.921 & 0.996 & 0.420 & 0.873 & 0.989 & 0.503 & 0.518 & 0.718\\
\laplace & 0.803 & 0.931 & 0.989 & 0.800 & 0.876 & 0.991 & 0.530 & 0.554 & 0.758\\

\end{tabular}
\end{scriptsize}
\end{table}

\subsection{Modeling flu epidemics}
\label{sec:flu}
We study $n=937$ weekly flu incidence reports from September 28, 1997 through August 30, 2015, across 10 regions in the United States (see the left panel of Figure \ref{fig:flu:heatmap:left}), obtained from \cite{cdc}.
We considered $d=20$ variables: the first 10 encode the current week's flu incidence (precisely, the percentage of doctor's visits in which flu-like symptoms are presented) in the 10 regions, and the last 10 encode the same but for the prior week.
We set $m = 5$, $r=99$, and also introduced exogenous variables to encode the week numbers, so $p=1$.
Thus, learning the \mqgm~here corresponds to learning the structure of a spatiotemporal graphical model, and reduces to solving 20 multiple quantile regression subproblems, each of dimension $(19 \times 5 + 1) \times 99 = 9504$. All subproblems took about 1 minute on a 6 core 3.3 Ghz Core i7 X980 processor.

The left panel of Figure \ref{fig:flu:heatmap:right} plots the wallclock time (seconds) for solving one subproblem using ADMM versus SCS \citep{o2015conic}, a cone solver that has been advocated as a reasonable choice for a class of problems encapsulating \eqref{eq:mqgm}; ADMM outperforms SCS by roughly two orders of magnitude.  The right panel of Figure \ref{fig:flu:heatmap:left} presents the conditional independencies recovered by the \mqgm. 
Nonzero entries in the upper left $10 \times 10$ submatrix correspond to dependencies between the $y_k$ variables for $k=1,\ldots,10$; \eg, the nonzero (0,2) entry suggests that region 1 and 3's flu reports are dependent.  The lower right $10 \times 10$ submatrix corresponds to the $y_k$ variables for $k=11,\ldots,20$, and the nonzero banded entries suggest that at any region the previous week's flu incidence (naturally) influences the next week's.
The left panel of Figure \ref{fig:flu:heatmap:left} visualizes these relationships by drawing an edge between dependent regions; region 6 is highly connected, suggesting that it is a bellwether for other regions, which is a qualitative observation also made by the CDC.
To draw samples from the fitted distributions, we ran our Gibbs sampler over the year, generating 1000 total samples, making 5 passes over all coordinates between each sample, and with a burn-in period of 100 iterations.  The right panel of Figure \ref{fig:flu:heatmap:right} plots samples from the marginal distribution of the percentages of flu reports at region six (other regions are in the supplement) throughout the year, revealing the heteroskedastic nature of the data; we also see that flu incidence (naturally) increases towards the end of the year.
\begin{figure}
\centering
\hspace{0.1in}
\includegraphics[width=0.48\textwidth]{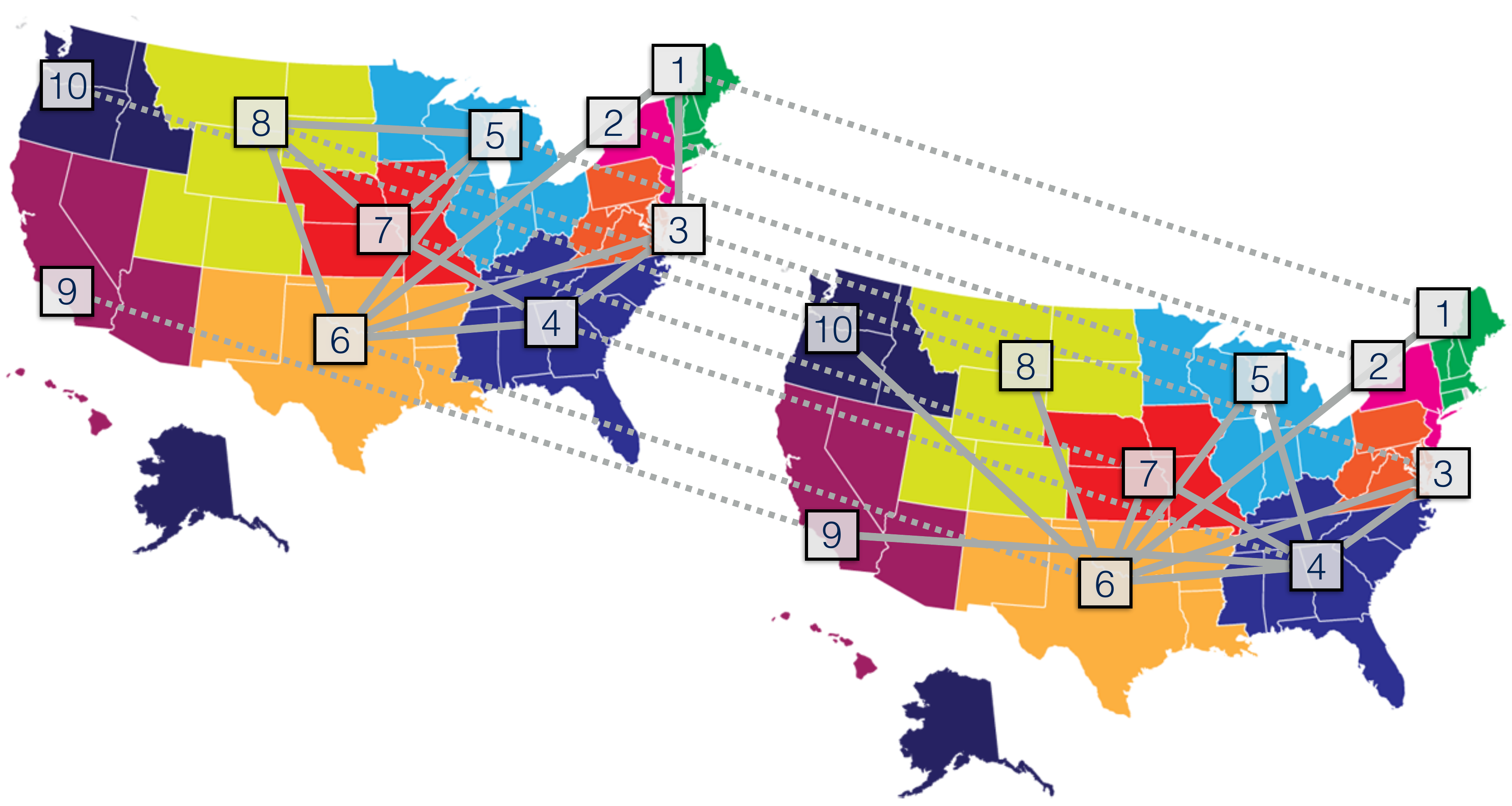}
\hfill
\includegraphics[width=0.3\textwidth]{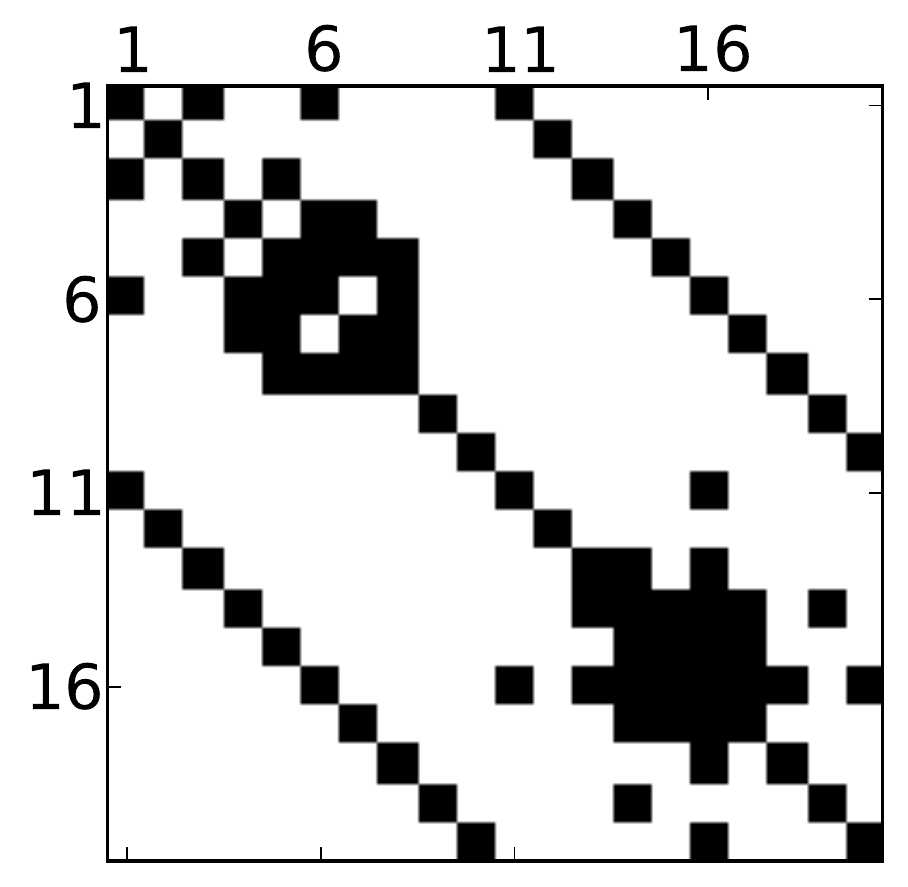}
\hspace{0.2in}
\caption{Conditional dependencies recovered by the \mqgm~on the flu data; each of the first ten cells on the right corresponds to a region of the U.S. on the left, and black means dependence.}
\label{fig:flu:heatmap:left}
\end{figure}
\begin{figure}
\centering
\includegraphics[width=0.48\textwidth]{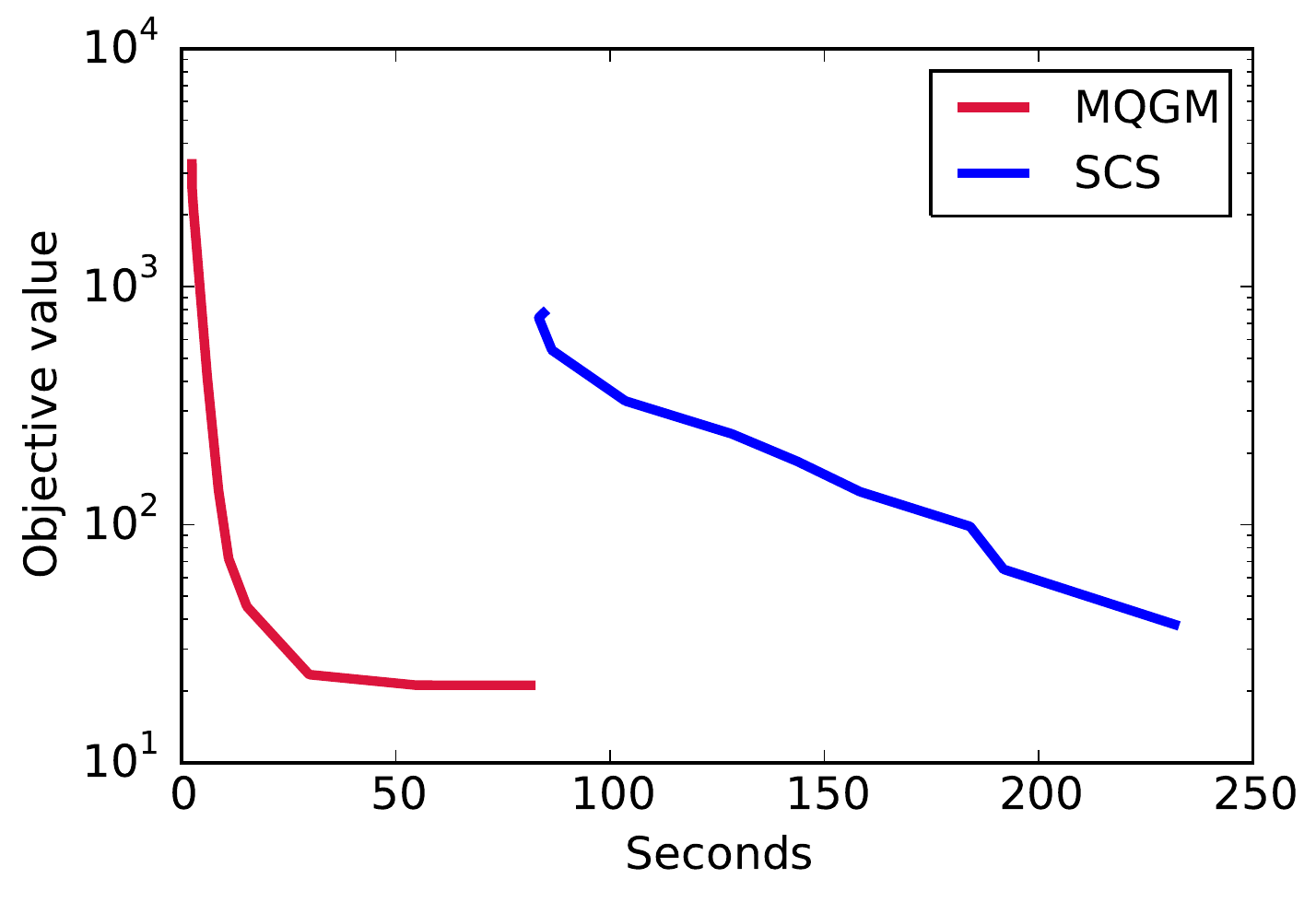}
\hfill
\includegraphics[width=0.44\textwidth]{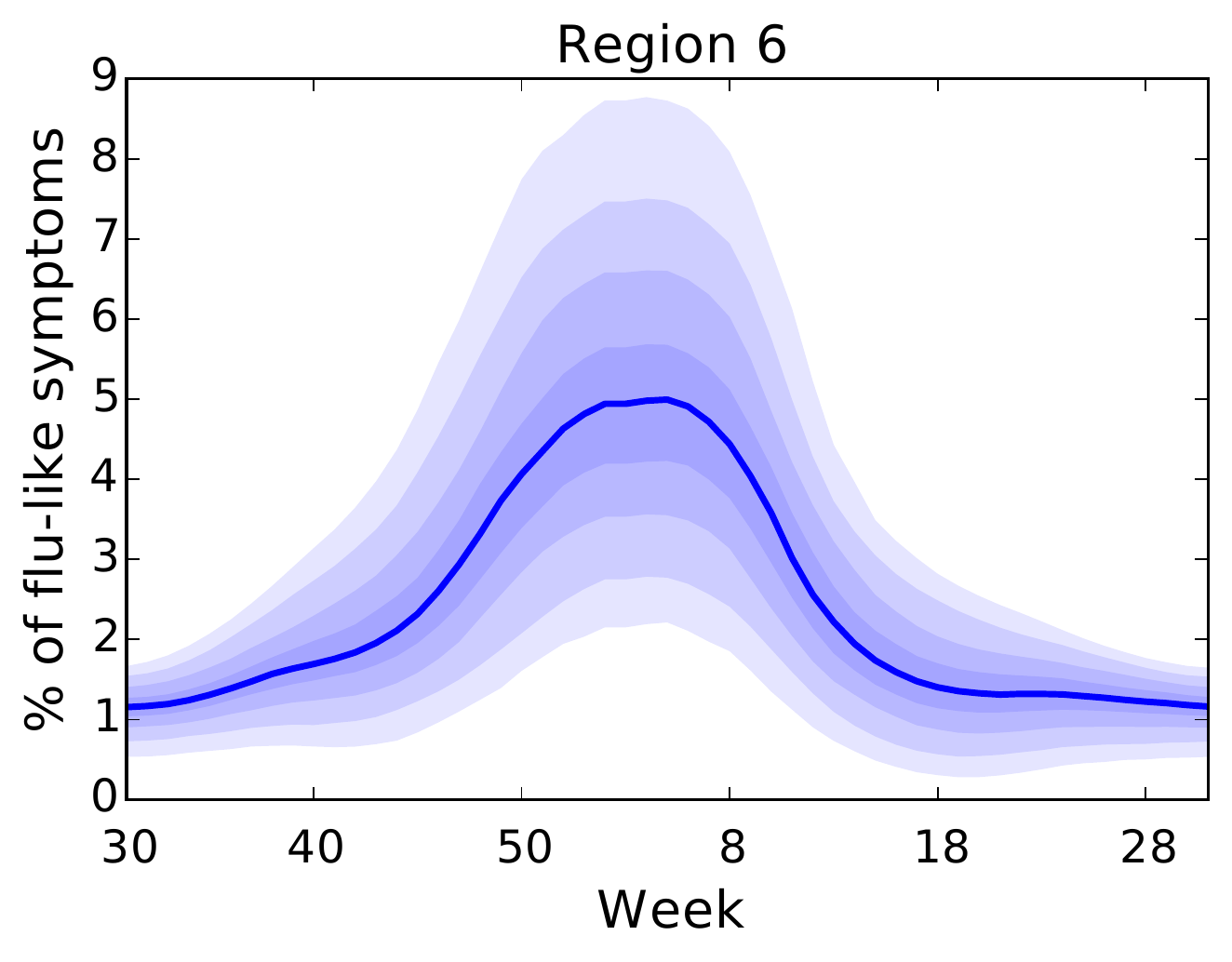}
\caption{Wallclock time (seconds) for solving one subproblem using ADMM versus SCS, on the left. Samples from the fitted marginal  distribution of the weekly flu incidence rates at region six, on the right. Samples at larger quantiles are shaded lighter; the median is in darker blue.} 
\label{fig:flu:heatmap:right}
\end{figure}

Our last example, on wind power data, is presented in the supplement.

\section{Discussion}
\label{sec:disc}

We proposed and studied the Multiple Quantile Graphical Model (\mqgm).  We established theoretical and empirical backing to the claim that the \mqgm~is capable of compactly representing relationships between heteroskedastic, non-Gaussian variables.
We developed efficient algorithms for estimation and sampling in the \mqgm. All in all, we believe that our work represents a step forward in the design of flexible yet tractable graphical models.

\subsubsection*{Acknowledgments}
This work was supported by a Department of Energy Computational Science Graduate Fellowship, under grant number DE-FG02-97ER25308.

\bibliographystyle{plainnat}
\bibliography{refs}

\begin{thebibliography}{50}
\providecommand{\natexlab}[1]{#1}
\providecommand{\url}[1]{\texttt{#1}}
\expandafter\ifx\csname urlstyle\endcsname\relax
  \providecommand{\doi}[1]{doi: #1}\else
  \providecommand{\doi}{doi: \begingroup \urlstyle{rm}\Url}\fi

\bibitem[Ali et~al.(2016)Ali, Khare, Oh, and Rajaratnam]{alnur}
Alnur Ali, Kshitij Khare, Sang-Yun Oh, and Bala Rajaratnam.
\newblock Generalized pseudolikelihood methods for inverse covariance
  estimation.
\newblock Technical report, 2016.
\newblock Available at \url{http://arxiv.org/pdf/1606.00033.pdf}.

\bibitem[Banerjee et~al.(2008)Banerjee, El~Ghaoui, and d'Aspremont]{BEA:08}
Onureena Banerjee, Laurent El~Ghaoui, and Alexandre d'Aspremont.
\newblock Model selection through sparse maximum likelihood estimation for
  multivariate {Gaussian} or binary data.
\newblock \emph{Journal of Machine Learning Research}, 9:\penalty0 485--516,
  2008.

\bibitem[Belloni and Chernozhukov(2011)]{belloni2011}
Alexandre Belloni and Victor Chernozhukov.
\newblock $\ell_1$-penalized quantile regression in high-dimensional sparse
  models.
\newblock \emph{Annals of Statistics}, 39\penalty0 (1):\penalty0 82--130, 2011.

\bibitem[Besag(1974)]{besag1974}
Julian Besag.
\newblock Spatial interaction and the statistical analysis of lattice systems.
\newblock \emph{Journal of the Royal Statistical Society: Series B},
  36\penalty0 (2):\penalty0 192--236, 1974.

\bibitem[Bishop(2006)]{Bishop:2006}
Christopher Bishop.
\newblock \emph{Pattern Recognition and Machine Learning}.
\newblock Springer, 2006.

\bibitem[Boyd et~al.(2011)Boyd, Parikh, Chu, Peleato, and
  Eckstein]{boyd2011distributed}
Stephen Boyd, Neal Parikh, Eric Chu, Borja Peleato, and Jonathan Eckstein.
\newblock Distributed optimization and statistical learning via the alternating
  direction method of multipliers.
\newblock \emph{Foundations and Trends in Machine Learning}, 3\penalty0
  (1):\penalty0 1--122, 2011.

\bibitem[{Centers for Disease Control and Prevention (CDC)}(2015)]{cdc}
{Centers for Disease Control and Prevention (CDC)}.
\newblock Influenza national and regional level graphs and data, August 2015.
\newblock URL \url{http://gis.cdc.gov/grasp/fluview/fluportaldashboard.html}.

\bibitem[Chen et~al.(2015)Chen, Witten, and Shojaie]{chen2015}
Shizhe Chen, Daniela Witten, and Ali Shojaie.
\newblock Selection and estimation for mixed graphical models.
\newblock \emph{Biometrika}, 102\penalty0 (1):\penalty0 47--64, 2015.

\bibitem[Dempster(1972)]{dempster1972covariance}
Arthur Dempster.
\newblock Covariance selection.
\newblock \emph{Biometrics}, 28\penalty0 (1):\penalty0 157--175, 1972.

\bibitem[Fan et~al.(2014)Fan, Fan, and Barut]{fan2014adaptive}
Jianqing Fan, Yingying Fan, and Emre Barut.
\newblock Adaptive robust variable selection.
\newblock \emph{Annals of Statistics}, 42\penalty0 (1):\penalty0 324--351,
  2014.

\bibitem[Finegold and Drton(2011)]{finegold2011}
Michael Finegold and Mathias Drton.
\newblock Robust graphical modeling of gene networks using classical and
  alternative t-distributions.
\newblock \emph{Annals of Applied Statistics}, 5\penalty0 (2A):\penalty0
  1057--1080, 2011.

\bibitem[Friedman et~al.(2008)Friedman, Hastie, and
  Tibshirani]{friedman2008sparse}
Jerome Friedman, Trevor Hastie, and Robert Tibshirani.
\newblock Sparse inverse covariance estimation with the graphical lasso.
\newblock \emph{Biostatistics}, 9\penalty0 (3):\penalty0 432--441, 2008.

\bibitem[Friedman et~al.(2010)Friedman, Hastie, and
  Tibshirani]{friedman2010applications}
Jerome Friedman, Trevor Hastie, and Robert Tibshirani.
\newblock Applications of the lasso and grouped lasso to the estimation of
  sparse graphical models.
\newblock Technical report, 2010.
\newblock Available at \url{http://statweb.stanford.edu/~tibs/ftp/ggraph.pdf}.

\bibitem[Gyrofi et~al.(2002)Gyrofi, Kohler, Krzyzak, and Walk]{gyorfi}
Laszlo Gyrofi, Michael Kohler, Adam Krzyzak, and Harro Walk.
\newblock \emph{A Distribution-Free Theory of Nonparametric Regression}.
\newblock 2002.

\bibitem[Heckerman et~al.(2000)Heckerman, Chickering, Meek, Rounthwaite, and
  Kadie]{heckerman2000dependency}
David Heckerman, David~Maxwell Chickering, David Meek, Robert Rounthwaite, and
  Carl Kadie.
\newblock Dependency networks for inference, collaborative filtering, and data
  visualization.
\newblock \emph{Journal of Machine Learning Research}, 1:\penalty0 49--75,
  2000.

\bibitem[H{\"o}fling and Tibshirani(2009)]{hoefling2009}
Holger H{\"o}fling and Robert Tibshirani.
\newblock Estimation of sparse binary pairwise {Markov} networks using
  pseudo-likelihoods.
\newblock \emph{Journal of Machine Learning Research}, 10:\penalty0 883--906,
  2009.

\bibitem[Hong et~al.(2014)Hong, Pinson, and Fan]{hong}
Tao Hong, Pierre Pinson, and Shu Fan.
\newblock Global energy forecasting competition 2012.
\newblock \emph{International Journal of Forecasting}, 30:\penalty0 357--363,
  2014.

\bibitem[Johnson(2013)]{nickdp}
Nicholas Johnson.
\newblock A dynamic programming algorithm for the fused lasso and
  {$\ell_0$}-segmentation.
\newblock \emph{Journal of Computational and Graphical Statistics}, 22\penalty0
  (2):\penalty0 246--260, 2013.

\bibitem[Kato(2011)]{kato2011}
Kengo Kato.
\newblock Group lasso for high dimensional sparse quantile regression models.
\newblock Technical report, 2011.
\newblock Available at \url{http://arxiv.org/pdf/1103.1458.pdf}.

\bibitem[Khare et~al.(2014)Khare, Oh, and Rajaratnam]{khare2014convex}
Kshitij Khare, Sang-Yun Oh, and Bala Rajaratnam.
\newblock A convex pseudolikelihood framework for high dimensional partial
  correlation estimation with convergence guarantees.
\newblock \emph{Journal of the Royal Statistical Society: Series B},
  77\penalty0 (4):\penalty0 803--825, 2014.

\bibitem[Koenker(2005)]{koenker2005quantile}
Roger Koenker.
\newblock \emph{Quantile Regression}.
\newblock Cambridge University Press, 2005.

\bibitem[Koenker(2011)]{koenker2011}
Roger Koenker.
\newblock Additive models for quantile regression: Model selection and
  confidence bandaids.
\newblock \emph{Brazilian Journal of Probability and Statistics}, 25\penalty0
  (3):\penalty0 239--262, 2011.

\bibitem[Koenker and Bassett(1978)]{koenker1978}
Roger Koenker and Gilbert Bassett.
\newblock Regression quantiles.
\newblock \emph{Econometrica}, 46\penalty0 (1):\penalty0 33--50, 1978.

\bibitem[Koenker et~al.(1994)Koenker, Ng, and Portnoy]{koenker1994}
Roger Koenker, Pin Ng, and Stephen Portnoy.
\newblock Quantile smoothing splines.
\newblock \emph{Biometrika}, 81\penalty0 (4):\penalty0 673--680, 1994.

\bibitem[Lauritzen(1996)]{lauritzen1996graphical}
Steffen Lauritzen.
\newblock \emph{Graphical models}.
\newblock Oxford University Press, 1996.

\bibitem[Lee and Hastie(2013)]{lee2013}
Jason Lee and Trevor Hastie.
\newblock Structure learning of mixed graphical models.
\newblock In \emph{Proceedings of the 16th International Conference on
  Artificial Intelligence and Statistics}, pages 388--396, 2013.

\bibitem[Liu and Wang(2012)]{liu2012tiger}
Han Liu and Lie Wang.
\newblock {TIGER}: A tuning-insensitive approach for optimally estimating
  {Gaussian} graphical models.
\newblock Technical report, 2012.
\newblock Available at \url{http://arxiv.org/pdf/1209.2437.pdf}.

\bibitem[Liu et~al.(2009)Liu, Lafferty, and Wasserman]{liu2009nonparanormal}
Han Liu, John Lafferty, and Larry Wasserman.
\newblock The nonparanormal: Semiparametric estimation of high dimensional
  undirected graphs.
\newblock \emph{Journal of Machine Learning Research}, 10:\penalty0 2295--2328,
  2009.

\bibitem[Liu et~al.(2012)Liu, Han, Yuan, Lafferty, and Wasserman]{liu2012high}
Han Liu, Fang Han, Ming Yuan, John Lafferty, and Larry Wasserman.
\newblock High-dimensional semiparametric {Gaussian} copula graphical models.
\newblock \emph{The Annals of Statistics}, pages 2293--2326, 2012.

\bibitem[Meier et~al.(2009)Meier, van~de Geer, and
  Buhlmann]{meier2009highdimensional}
Lukas Meier, Sara van~de Geer, and Peter Buhlmann.
\newblock High-dimensional additive modeling.
\newblock \emph{Annals of Statistics}, 37\penalty0 (6):\penalty0 3779--3821,
  2009.

\bibitem[Meinshausen and B{\"u}hlmann(2006)]{meinshausen2006high}
Nicolai Meinshausen and Peter B{\"u}hlmann.
\newblock High-dimensional graphs and variable selection with the lasso.
\newblock \emph{Annals of Statistics}, 34\penalty0 (3):\penalty0 1436--1462,
  2006.

\bibitem[Neville and Jensen(2004)]{neville2004dependency}
Jennifer Neville and David Jensen.
\newblock Dependency networks for relational data.
\newblock In \emph{Proceedings of Fourth IEEE International Conference on the
  Data Mining}, pages 170--177. IEEE, 2004.

\bibitem[O'Donoghue et~al.(2013)O'Donoghue, Chu, Parikh, and Boyd]{o2015conic}
Brendan O'Donoghue, Eric Chu, Neal Parikh, and Stephen Boyd.
\newblock Operator splitting for conic optimization via homogeneous self-dual
  embedding.
\newblock Technical report, 2013.
\newblock Available at \url{https://stanford.edu/~boyd/papers/pdf/scs.pdf}.

\bibitem[Oh et~al.(2014)Oh, Dalal, Khare, and Rajaratnam]{NIPS2014_5576}
Sang-Yun Oh, Onkar Dalal, Kshitij Khare, and Bala Rajaratnam.
\newblock Optimization methods for sparse pseudolikelihood graphical model
  selection.
\newblock In \emph{Advances in Neural Information Processing Systems 27}, pages
  667--675, 2014.

\bibitem[Parikh and Boyd(2013)]{parikh2013proximal}
Neal Parikh and Stephen Boyd.
\newblock Proximal algorithms.
\newblock \emph{Foundations and Trends in Optimization}, 1\penalty0
  (3):\penalty0 123--231, 2013.

\bibitem[Peng et~al.(2009)Peng, Wang, Zhou, and Zhu]{peng2009partial}
Jie Peng, Pei Wang, Nengfeng Zhou, and Ji~Zhu.
\newblock Partial correlation estimation by joint sparse regression models.
\newblock \emph{Journal of the American Statistical Association}, 104\penalty0
  (486):\penalty0 735--746, 2009.

\bibitem[Raskutti et~al.(2012)Raskutti, Wainwright, and
  Yu]{raskutti2012minimax}
Garvesh Raskutti, Martin Wainwright, and Bin Yu.
\newblock Minimax-optimal rates for sparse additive models over kernel classes
  via convex programming.
\newblock \emph{Journal of Machine Learning Research}, 13:\penalty0 389--427,
  2012.

\bibitem[Rocha et~al.(2008)Rocha, Zhao, and Yu]{rocha2008path}
Guilherme Rocha, Peng Zhao, and Bin Yu.
\newblock A path following algorithm for sparse pseudo-likelihood inverse
  covariance estimation {(SPLICE)}.
\newblock Technical report, 2008.
\newblock Available at
  \url{https://www.stat.berkeley.edu/~binyu/ps/rocha.pseudo.pdf}.

\bibitem[Rothman et~al.(2008)Rothman, Bickel, Levina, and Zhu]{rothman2008}
Adam Rothman, Peter Bickel, Elizaveta Levina, and Ji~Zhu.
\newblock Sparse permutation invariant covariance estimation.
\newblock \emph{Electronic Journal of Statistics}, 2:\penalty0 494--515, 2008.

\bibitem[Sohn and Kim(2012)]{sohn2012joint}
Kyung-Ah Sohn and Seyoung Kim.
\newblock Joint estimation of structured sparsity and output structure in
  multiple-output regression via inverse covariance regularization.
\newblock In \emph{Proceedings of the 15th International Conference on
  Artificial Intelligence and Statistics}, pages 1081--1089, 2012.

\bibitem[Takeuchi et~al.(2006)Takeuchi, Le, Sears, and
  Smola]{takeuchi2006nonparametric}
Ichiro Takeuchi, Quoc Le, Timothy Sears, and Alexander Smola.
\newblock Nonparametric quantile estimation.
\newblock \emph{Journal of Machine Learning Research}, 7:\penalty0 1231--1264,
  2006.

\bibitem[Varin and Vidoni(2005)]{varin2005}
Cristiano Varin and Paolo Vidoni.
\newblock A note on composite likelihood inference and model selection.
\newblock \emph{Biometrika}, 92\penalty0 (3):\penalty0 519--528, 2005.

\bibitem[Voorman et~al.(2014)Voorman, Shojaie, and Witten]{voorman2014}
Arend Voorman, Ali Shojaie, and Daniela Witten.
\newblock Graph estimation with joint additive models.
\newblock \emph{Biometrika}, 101\penalty0 (1):\penalty0 85--101, 2014.

\bibitem[Wainwright(2009)]{wainwright2009sharp}
Martin Wainwright.
\newblock Sharp thresholds for high-dimensional and noisy sparsity recovery
  using-constrained quadratic programming {(Lasso)}.
\newblock \emph{IEEE Transactions on Information Theory}, 55\penalty0
  (5):\penalty0 2183--2202, 2009.

\bibitem[Wang and Ip(2008)]{wang2008}
Yuchung Wang and Edward Ip.
\newblock Conditionally specified continuous distributions.
\newblock \emph{Biometrika}, 95\penalty0 (3):\penalty0 735--746, 2008.

\bibitem[Wytock and Kolter(2013)]{wytock2013sparse}
Matt Wytock and Zico Kolter.
\newblock Sparse {Gaussian} conditional random fields: Algorithms, theory, and
  application to energy forecasting.
\newblock In \emph{Proceedings of the 30th International Conference on Machine
  Learning}, pages 1265--1273, 2013.

\bibitem[Yang et~al.(2012)Yang, Ravikumar, Allen, and Liu]{yang2012}
Eunho Yang, Pradeep Ravikumar, Genevera Allen, and Zhandong Liu.
\newblock Graphical models via generalized linear models.
\newblock In \emph{Advances in Neural Information Processing Systems 25}, pages
  1358--1366, 2012.

\bibitem[Yang et~al.(2015)Yang, Ravikumar, Allen, and Liu]{yang2015}
Eunho Yang, Pradeep Ravikumar, Genevera Allen, and Zhandong Liu.
\newblock Graphical models via univariate exponential family distributions.
\newblock \emph{Journal of Machine Learning Research}, 16:\penalty0 3813--3847,
  2015.

\bibitem[Yuan and Lin(2007)]{yuan2007}
Ming Yuan and Yi~Lin.
\newblock Model selection and estimation in the {Gaussian} graphical model.
\newblock \emph{Biometrika}, 94\penalty0 (1):\penalty0 19--35, 2007.

\bibitem[Yuan and Zhang(2014)]{yuan2014partial}
Xiao-Tong Yuan and Tong Zhang.
\newblock Partial {Gaussian} graphical model estimation.
\newblock \emph{IEEE Transactions on Information Theory}, 60\penalty0
  (3):\penalty0 1673--1687, 2014.

\end{thebibliography}

\newpage
\allowdisplaybreaks

\renewcommand\thefigure{S.\arabic{figure}}    
\renewcommand\thesection{S.\arabic{section}}
\renewcommand\theequation{S.\arabic{equation}}     
\setcounter{figure}{0}
\setcounter{section}{0}
\setcounter{equation}{0}

\title{Supplement to ``The Multiple Quantile Graphical Model''}
\maketitle

\section{Proof of Lemma \ref{lem:indep}}
\label{supp:indep}

If the conditional quantiles satisfy $Q_{U|V,W}(\alpha) =
Q_{U|W}(\alpha)$ for all $\alpha \in [0,1]$, then the conditional  
CDF must obey the same property,
\ie, $F_{U|V,W}(t)=F_{U|W}(t)$ for all $t$ in the support of $U$.
This is simply because any CDF may be expressed in terms
of its corresponding quantile function (\ie, inverse CDF), as in
\begin{equation*}
F_{U|V,W}(t) = \sup\{\alpha \in [0,1] : Q_{U|V,W}(\alpha) \leq t \},
\end{equation*}
and the right-hand side does not depend on $V$, so neither can the
left-hand side.   But this precisely implies that the distribution of
$U|V,W$ equals that of $U|W$, \ie, $U$ and $V$ are conditionally
independent given $W$.  We note that the converse of the statement in
the lemma is true as well, by just reversing all the arguments here.  
\hfill\qedsymbol

\section{Proof of Lemma \ref{lem:gibbs}}
This result can be seen as a generalization of Theorem 3 in \cite{heckerman2000dependency}.

First, we define an {\it iteration} of Gibbs sampling to be a single pass through all the variables (without a loss of generality, we take this order to be $y_1,\ldots,y_d$).  Now, consider a particular iteration of Gibbs sampling; let $\tilde{y}_1, \ldots, \tilde{y}_d$ be the values assigned to the variables on the previous iteration.  Then the transition kernel for our Gibbs sampler is given by
\begin{align}
\prob(y_1,\ldots,y_d | \tilde{y}_1,\ldots,\tilde{y}_d) & = \prob(y_d | y_{d-1}, \ldots, y_1, \tilde{y}_1, \ldots, \tilde{y}_d) \prob(y_{d-1}, \ldots, y_1 | \tilde{y}_1, \ldots, \tilde{y}_d) \label{eq:1} \\
& = \prob(y_d | y_{d-1}, \ldots, y_1) \prob(y_{d-1}, \ldots, y_1 | \tilde{y}_1, \ldots, \tilde{y}_d) \label{eq:2} \\
& = \prob(y_d | y_{d-1}, \ldots, y_1) \prob(y_{d-1} | y_{d-2}, \ldots y_{1}, \tilde{y}_d) \cdots \prob(y_1 | \tilde{y}_2, \ldots, \tilde{y}_d), \label{eq:3}
\end{align}
where \eqref{eq:1} follows by the definition of conditional probability, \eqref{eq:2} by conditional independence, and \eqref{eq:3} by repeated applications of these tools.  Since each conditional distribution is assumed to be (strictly) positive, we have that the transition kernel is also positive, which in turn implies \citep[page 544]{Bishop:2006} that the induced Markov chain is ergodic with a unique stationary distribution that can be reached from any initial point.
\hfill\qedsymbol

\section{Statement and discussion of regularity conditions for Theorem \ref{thm:recovery}}

For each $k=1,\ldots,r$, $\ell=1,\ldots,r$, let us define the ``effective'' (independent) error terms \smash{$\epsilon_{\ell ki}=y_k^{(i)} - b_{\ell k}^* - \sum_{j \neq k} \phi(y_j^{(i)})^T \theta_{\ell kj}^*$}, over $i=1,\ldots,n$.  Denote by \smash{$F_{\epsilon_{\ell k}}$} the conditional CDF of 
\smash{$\epsilon_{\ell ki} | y_{\neg k}^{(i)}$}, $i=1,\ldots,n$, which by construction satisfies \smash{$F_{\epsilon_{\ell k}}(0) = \alpha_\ell$}.  Also define the underlying support
\begin{equation*}
S_{\ell k} = \big\{ j \in \{1,\ldots,d\} : \theta_{\ell k j}^* \not= 0
\big\}.  
\end{equation*}

Here we take a moment to explain a somewhat subtle indexing issue with
the columns of the feature matrix \smash{$\Phi \in \R^{n \times dm}$}.
For a  single fixed index $j=1,\ldots,d$, we will extract an
appropriate block of columns of \smash{$\Phi \in \R^{n \times dm}$}, 
corresponding to the basis expansion of variable $j$, by
writing $\Phi_j$.  More precisely, we use $\Phi_j$ to denote the block 
of $m$ columns  
\begin{equation}
\label{eq:basis_exp}
[\Phi_{(j-1)m+1}, \Phi_{(j-1)m+2}, \ldots, \Phi_{jm}].
\end{equation}
We do this because it simplifies notation considerably. (Occasionally,
to be transparent, we will use the more exhaustive notation on the
right-hand side in \eqref{eq:basis_exp}, but this is to be treated as
an exception, and the default is to use the concise notation as in
$\Phi_j$.) The same rule will be used for subsets of indices among
$1,\ldots,d$, so that \smash{$\Phi_{S_{\ell k}}$} denotes the
appropriate block of \smash{$m|S_{\ell k}|$} columns corresponding to
the basis expansions of the variables in $S_{\ell k}$.

For all $k=1,\ldots,d$, $\ell=1,\ldots,r$, we will 
assume the following regularity conditions.

\begin{enumerate}
\item[\bf A1.] {\it Groupwise irrepresentability:} for $j \in S_{\ell k}^c$, we require that \smash{$\| \Phi_j^T \Phi_{S_{\ell k}}\|_F < \lambda_1 / (6 f_{\epsilon_{\ell k}}(0) \gamma)$}, where $S_{\ell k} = \{ j \in \{1,\ldots,dm\} : \theta_{\ell kj}^* \neq 0 \}$, \smash{$f_{\epsilon_{\ell k}}$} is the density of \smash{$F_{\epsilon_{\ell k}}$}, and $\gamma > 0$ is a quantity prescribed by Lemma \ref{supp:estimation}.



\item[\bf A2.] {\it Distributional smoothness:} we assume that \smash{$| F_{\epsilon_{\ell k}}(x) - F_{\epsilon_{\ell k}}(0) - x f_{\epsilon_{\ell k}}(0) | \leq C_1 x^2$} for all $|x| \leq C_2$, where $C_1,C_2>0$ are constants.

\item[\bf A3.] {\it Correlation restriction:} we assume that \smash{$C_3 \leq (f_{\epsilon_{\ell k}}(0)/n) \lambda_{\min} (\Phi_{S_{\ell k}}^T \Phi_{S_{\ell k}} ) \leq C_4$} for constants $C_3, C_4 > 0$, where $\lambda_{\min}(A)$ denotes the minimum eigenvalue of $A$.

\item[\bf A4.] {\it Basis and support size restrictions:} we assume that $m = O(n^{1/9})$ and $s = O(n^{1/21})$, where $s=|S_{\ell k}|$. We also assume, with probability tending to one, that $\Phi_{\max}=\Omega(1)$ and $\Phi_{\max} = o(n^{1/21} / \log^{1/2} n)$, where we write $\Phi_{\max}$ to denote the maximum absolute entry of the basis matrix $\Phi$. 
\end{enumerate}

Next, we provide some intuition for these conditions.

\bigskip
{\it Condition A1.}
Fix some $j \in S_{\ell k}^c$.  For notational convenience, we let
\[
A = \Phi_j^T \Phi_{S_{\ell k}} \in \reals^{m \times sm}.
\]
Observe that each entry of $A$ can be expressed as
\begin{equation}
A_{ip} = n \rho_{i,p} \| \Phi_{(j-1)m + i} \|_2 \| \Phi_{p} \|_2, \label{eq:aip}
\end{equation}
for $i=1,\ldots,m$, $p$ denoting an index into the basis expansion of the columns $\Phi_{S_{\ell k}}$, and $\rho_{i,p}$ denoting the sample correlation coefficient for the columns $\Phi_i$ and $\Phi_p$.  Since $\| A_p \|_F \leq \sqrt{m} \| A_p \|_\infty$, we have that
\[
\max_{i,p} \; \rho_{i,p} < \frac{\lambda_1}{6 n^2 f_{\epsilon_{\ell k}}(0) \sqrt{m}} 
\]
is sufficient for condition A1; here, we have also used the column scaling assumption $\| \Phi_p \|_2 \leq \sqrt{n}$.

So, roughly speaking, bounded correlation between each pair of columns in the submatrices $\Phi_j$ and $\Phi_{S_{\ell k}}$ is enough for condition A1 to hold; note that this is trivially satisfied when $\Phi_i^T \Phi_p = 0$, for $i=1,\ldots,m$, and $p$ as defined above. Condition A1 is therefore similar to, \eg, the mutual incoherence condition of \cite{wainwright2009sharp} for the lasso, which is given by
\begin{align*}
\left\| \Phi_{S^c}^T \Phi_S \left( \Phi_S^T \Phi_S \right)^{-1} \right\|_\infty \leq 1-\tilde{\gamma} \iff \max_{j \in S^c} \left\| \left( \Phi_S^T \Phi_S \right)^{-1} \Phi_S^T \Phi_{j} \right\|_1 \leq 1-\tilde{\gamma}, \label{eq:inco}
\end{align*}
where again $\Phi_S$ extracts the appropriate block of columns of $\Phi$, $\| \cdot \|_\infty$ here denotes the $\ell_\infty$ operator norm (maximum $\ell_1$ norm of a row), $\| \cdot \|_1$ here denotes the elementwise $\ell_1$ norm, and $\tilde{\gamma} \in (0,1]$ is a constant.  This condition can be seen as requiring bounded correlation between each column in the submatrix $\Phi_{S^c}$ and all columns in the submatrix $\Phi_S$.

\bigskip
{\it Condition A2.}
This condition is similar to requiring that $f_{\epsilon_{\ell k}}(x)$ be Lipschitz, over some $x$ in a neighborhood of 0. We can show that the Laplace distribution, \eg, satisfies this condition.

The density and distribution functions for the Laplace distribution with location zero and unit scale are given by
\[
f_{\epsilon_{\ell k}}(x) = (1/2) \exp ( -|x| )
\]
and
\[
F_{\epsilon_{\ell k}}(x) =
\begin{cases}
1 - (1/2) \exp ( - x ) & \textrm{if } x \geq 0 \\
(1/2) \exp ( x ) & \textrm{if } x < 0,
\end{cases}
\]
respectively. 

Now, suppose $0 \leq x \leq C_2$.  Then we can express condition A2 as
\begin{align*}
& | f_{\epsilon_{\ell k}}(x) - f_{\epsilon_{\ell k}}(0) - x f_{\epsilon_{\ell k}}(0) | \leq C_1 x^2 \iff -2 C_1 x^2 \leq \exp(-x) + x - 1 \leq 2 C_1 x^2.
\end{align*}

For the first inequality, since $1-x \leq \exp(-x)$, it is sufficient to check that $0 \leq C_1 x^2$, which is true for $C_1 > 0$ and all $x$.  For the second inequality, by differentiating and again using $1-x \leq \exp(-x)$, we have that the function
\begin{equation}
2 C_1 x^2 - \exp(-x) - x + 1
\end{equation}
is nondecreasing in $x \geq 0$; thus, it is sufficient to check that this function is nonnegative for $x = 0$, which is true.

Now, suppose $-C_2 \leq x < 0$.  Then we can express condition A2 as
\begin{align*}
& | f_{\epsilon_{\ell k}}(x) - f_{\epsilon_{\ell k}}(0) - x f_{\epsilon_{\ell k}}(0) | \leq C_1 x^2 \iff -2 C_1 x^2 \leq \exp(x) - x - 1 \leq 2 C_1 x^2.
\end{align*}

By symmetry with the preceding case, the first inequality here holds.  The second inequality here also holds, since $\exp(x) - 2 C_1 x^2 - x - 1$ is continuous and increasing in $x < 0$; taking the limit as $x \uparrow 0$ gives that this function is nonpositive as required.

\bigskip
{\it Condition A3.}
This condition is a generalization of the minimum eigenvalue condition of \cite{wainwright2009sharp}, \ie, 
$
c_{\min} \leq \lambda_{\min} \left( (1/n) \Phi_S^T \Phi_S \right),
$
for some constant $c_{\min} > 0$, and where we write $\Phi_S$ to extract the appropriate block of columns of $\Phi$.

\bigskip
{\it Condition A4.} This condition allows the number of basis functions $m$ in the expansion to grow with $n$, at a polynomial rate (with fractional exponent). This is roughly in line with standard nonparametric regression; \eg, when estimating a continuous differentiable function via a spline expansion, one typically takes the number of basis functions $m$ to scale as $n^{1/3}$, and the more derivatives that are assumed, the smaller the fractional exponent \citep{gyorfi}.  The condition also restricts, for any given variable, the number of variables $s$ that contribute to its neighborhood model to be polynomial in $n$ (with a smaller fractional exponent).  

Finally, the condition assumes that the entries of the basis matrix $\Phi$ (\ie, the matrix of transformed variables) to be at least of constant order, and at most of polynomial order (with small fractional exponent), with $n$. We note that this implicitly places a restriction on the tails of distribution governing the data \smash{$y^{(i)}_j$}, $i=1,\ldots,n$, $j=1,\ldots,d$.  However, the restriction is not a strong one, because it allows the maximum to grow polynomially large with $n$ (whereas a logarithmic growth would be expected, \eg, for normal data).  Furthermore, it is possible to trade off the restrictions on $m$, $s$, $\Phi_{\max}$, and $d$ (presented in the statement of the theorem), making each of these restrictions more or less stringent, if required.

\section{Proof of Theorem \ref{thm:recovery}}
\label{supp:recovery}

The general strategy that we use here for support recovery is inspired
by that in \cite{fan2014adaptive}, for $\ell_1$-penalized quantile 
regression. 

Fix some $k=1,\ldots,d$ and $\ell=1,\ldots,r$.   We consider the
conditional distribution \smash{$y_k | y_{\neg k}$}, whose 
$\alpha_\ell$-quantile is assumed to satisfy \eqref{eq:invcdf}. 
Hence, to be perfectly clear, all expectations and probability
statements in what follows are to be interpreted with respect to the 
observations \smash{$y_k^{(i)}$, $i=1,\ldots,n$} conditional on 
\smash{$y_j^{(i)}$, $i=1,\ldots,n$}, for $j \not= k$ (and thus we can
treat the feature matrix $\Phi$ as fixed throughout). In the setting 
assumed by the theorem, the conditional quantile model in
\eqref{eq:invcdf} is, more explicitly,
\begin{equation*}
Q_{y_k | y_{\neg k}}(\alpha_\ell) = b_{\ell k}^* + \sum_{j \neq k}^d
(\theta_{\ell k j}^*)^T \phi^j(y_j),   
\end{equation*}
for some unknown parameters \smash{$b^*_{\ell k}$} and  
\smash{$\theta_{\ell k j}^*$}, $j=1,\ldots,d$. For simplicity, in this
proof, we will drop the intercept term 
completely both from the model (denoted \smash{$b_{\ell k}^*$}) and
the optimization problem in \eqref{eq:mqgm} (here denoted
\smash{$b_{\ell k}$}) that defines  
the estimator in question. Including the intercept is not at all
difficult, and it just requires some extra bookkeeping at various
places.  Recall that we define  
\begin{equation*}
S_{\ell k} = \big\{ j \in \{1,\ldots,d\} : \theta_{\ell k j}^* \not= 0
\big\},  
\end{equation*}
and analogously define
\begin{equation*}
\hat{S}_{\ell k} = \big\{ j \in \{1,\ldots,d\} : \hat\theta_{\ell k j} 
\not= 0 \big\},  
\end{equation*}
where \smash{$\hat\theta_{\ell k}=(\hat\theta_{\ell k 1}, \ldots, 
\hat\theta_{\ell k d}) \in \R^{dm}$} is the solution in
\eqref{eq:mqgm_theta}.  

We will show that, with probability at least $1-\delta/(dr)$, it holds 
that \smash{$S_{\ell k} = \hat{S}_{\ell k}$}.  A union bound (over all
choices $k=1,\ldots,d$ and $\ell=1,\ldots,r$) will then tell us that 
\smash{$E^* = \hat{E}$} with probability at least $1-\delta$,
completing the proof.

To certify that \smash{$S_{\ell k} = \hat{S}_{\ell k}$}, we will show
that the unique solution in \eqref{eq:mqgm_theta} is given by
\begin{equation}
\label{eq:proposed_sol}
\hat\theta_{\ell k (S_{\ell k})} = \tilde\theta_{\ell k (S_{\ell k})},
\quad 
\hat\theta_{\ell k (S_{\ell k}^c)} = 0,
\end{equation}
where \smash{$\tilde\theta_{\ell k (S_{\ell k})}$}  
solves the ``restricted'' optimization problem:
\begin{equation}
\minimize_{\theta_{\ell k (S_{\ell k})}} \;
\psi_{\alpha_\ell} \Big(Y_k - \Phi_{S_{\ell k}} \theta_{\ell k
  (S_{\ell k})}\Big) + 
\lambda_1 \sum_{j \in S_{\ell k}} \|\theta_{\ell k j} \|_2 + 
\frac{\lambda_2}{2} \| \theta_{\ell k (S_{\ell k})} \|_2^2.
\label{eq:restricted}
\end{equation}


Now, to prove that \smash{$\hat\theta_{\ell k}$} as defined above in 
\eqref{eq:proposed_sol} indeed the solution in \eqref{eq:mqgm_theta},  
we need to check that it satisfies the KKT conditions for
\eqref{eq:mqgm_theta}, namely
\begin{align}
\Phi_{S_{\ell k}}^T v_{\ell} \Big(Y_k -
\Phi_{S_{\ell k}} \tilde\theta_{\ell k (S_{\ell k})}\Big) 
- \lambda_2 \tilde{\theta}_{\ell k (S_{\ell k})} & = 
\lambda_1 u_{\ell k (S_{\ell k})},   
\label{eq:opt1} \\ 
\Phi_{S_{\ell k}^c}^T v_{\ell} \Big(Y_k - 
\Phi_{S_{\ell k}} \tilde{\theta}_{\ell k (S_{\ell k})} \Big) & = 
\lambda_1 u_{\ell k (S_{\ell k}^c)}, 
\label{eq:opt2} 
\end{align}
where \smash{$v_{\ell} (Y_k - \Phi_{S_{\ell k}} \tilde\theta_{\ell k
    (S_{\ell k})} ) \in \R^n$} is a subgradient of 
\smash{$\psi_{\alpha_\ell}(\cdot)$} at \smash{$Y_k - \Phi_{S_{\ell k}} 
  \tilde{\theta}_{\ell k (S_{\ell k})}$}, \ie,    
\begin{equation*}
\Big[v_{\ell} \Big(Y_k - \Phi_{S_{\ell k}} 
\tilde\theta_{\ell k (S_{\ell k})}\Big)\Big]_i   
= \alpha_\ell - I_{-}\Big(y_k^{(i)} - \Phi_{i (S_{\ell k})}  
\tilde\theta_{\ell k (S_{\ell k})}\Big), 
\quad i=1,\ldots,n
\end{equation*}
where $I_-(\cdot)$ is the indicator function of the nonpositive 
real line, and where each $u_{\ell k j} \in \R^m$ is a subgradient of   
$\|\cdot\|_2$ at \smash{$\tilde\theta_{\ell k j}$}, \ie,    
\begin{equation*}
u_{\ell k j}  \in
\begin{cases}
\{\tilde\theta_{\ell k j} / \| \tilde\theta_{\ell k j} \|_2\} &
\text{if} \;\, \theta_{\ell k j} \neq 0 \\ 
\{ x \in \R^m : \|x\|_2 \leq 1 \} & \text{if} \;\, 
\theta_{\ell k j} = 0, 
\end{cases}
\end{equation*}
for $j=1,\ldots,d$.  Note that, since 
\smash{$\tilde\theta_{\ell k(S_{\ell k})}$} is optimal for the
restricted problem \eqref{eq:restricted}, we know that there exists a 
collection of subgradients \smash{$u_{\ell k(S_{\ell k})}$} to satisfy
\eqref{eq:opt1}, from the KKT conditions for \eqref{eq:restricted} itself.  

It remains to satisfy \eqref{eq:opt2}, and for this, we can use   
\smash{$u_{\ell k j} = 
\Phi_j^T v_{\ell}(Y_k - \Phi_{S_{\ell k}} \tilde{\theta}_{\ell k
  (S_{\ell k})})$} as a valid choice of subgradient, for each
\smash{$j \in S_{\ell k}^c$}, provided that  
\begin{equation}
\label{eq:subg_bound}
\Big\|\Phi_j^T v_{\ell} \Big(Y_k - \Phi_{S_{\ell k}}
\tilde\theta_{\ell k (S_{\ell k})} \Big)\Big\|_2 < \lambda_1, \quad
\text{for}\;\, j \in S_{\ell k}^c. 
\end{equation}
Define \smash{$z_j(\vartheta) = 
\Phi_j^T v_{\ell}(Y_k - \Phi_{S_{\ell k}} \vartheta)$, for 
$j \in S_{\ell k}^c$}, and define a ball  
\begin{equation*}
B^* = \big\{ \vartheta \in \R^{sm} : 
\| \vartheta - \theta_{\ell k (S_{\ell k})}^* \|_2 \leq \gamma \big\},
\end{equation*}
where we write $s=|S_{\ell k}|$. 
To show \eqref{eq:subg_bound}, then, it suffices to show that
\begin{equation}
\label{eq:ball_bound}
\underbrace{\tilde\theta_{\ell k(S_{\ell k})} 
\in B^* \vphantom{\sup_{\vartheta \in B^*}}}_{E_1},
\quad \text{and} \quad
\underbrace{\max_{j \in S_{\ell k}^c} \; 
\sup_{\vartheta \in B^*} \;
\|z_j(\vartheta)\|_2 < \lambda_1}_{E_2}.
\end{equation}
In Lemma \ref{lem:estimation}, given in Section \ref{supp:estimation}, it is 
shown that the event $E_1$ defined above occurs with probability at
least $1-\delta/(2dr)$, with a choice of radius 
\newcommand{\ourgamma}{C \left( \frac{\lambda_1 s \sqrt{m}}{n} + \sqrt{\frac{s \log n}{n}} \right)}
\begin{equation*}
\gamma = \ourgamma,
\end{equation*}
for a constant $C>0$. Below we show that $E_2$ occurs with probability at least $1-\delta/(2dr)$, as 
well. 

For $j=1,\ldots,d$, let us expand 
\begin{multline}
\label{eq:z_expand}
z_j(\vartheta) = 
\underbrace{ \Phi_j ^T v_{\ell} (\epsilon_{\ell k}) }_{ \Delta_1^j } +
\underbrace{ \Phi_j^T \E \Big[ v_{\ell} \Big(Y_k -
 \Phi_{S_{\ell k}} \vartheta \Big) - v_{\ell} (\epsilon_{\ell k})
 \Big]}_{ \Delta_2^j } + {} \\
\underbrace{ \Phi_j^T \bigg(
    v_{\ell} \Big(Y_k -\Phi_{S_{\ell k}} \vartheta) -
    v_{\ell}(\epsilon_{\ell k} \Big) - 
    \E \Big[ v_{\ell} (Y_k - 
    \Phi_{S_{\ell k}} \vartheta) - v_{\ell} (\epsilon_{\ell k}) \Big]
    \bigg)}_{ \Delta_3^j}, 
\end{multline}
where \smash{$\epsilon_{\ell k}=
(\epsilon_{\ell k 1},\ldots,\epsilon_{\ell k n}) \in \R^n$} is a
vector of the effective error terms, which recall, is defined by  
\smash{$\epsilon_{\ell k}=Y_k - \Phi \theta_{\ell k}^*$}. Therefore,
to show that the event $E_2$ in \eqref{eq:ball_bound} holds, we can
show that for each $p=1,2,3$,
\begin{equation*}
\max_{j \in S_{\ell k}^c} \; 
\sup_{\vartheta \in B^*} \;
\| \Delta^j_p \|_2 < \frac{\lambda_1}{3}.
\end{equation*}
Further, to show that $E_2$ holds with probability at least
$1-\delta/(2dr)$, we can show that the above holds for $p=1,3$ each
with probability at least $1-\delta/(4dr)$, as the statement for $p=2$
is deterministic. We now bound the terms
\smash{$\Delta^j_1,\Delta^j_2, \Delta^j_3$} one by one.  

\bigskip
{\it Bounding \smash{$\| \Delta_1^j \|_2$}.}  
Fix \smash{$j \in S_{\ell k}^c$}, and write
\begin{equation*}
\Phi_j^T v_{\ell}(\epsilon_{\ell k}) = 
\bigg(\sum_{i=1}^n \Phi_{i,(j-1)m+1} 
v_{\ell} (\epsilon_{\ell k i}),\ldots,
\sum_{i=1}^n \Phi_{i,jm} v_{\ell} (\epsilon_{\ell k i}) \bigg),
\end{equation*}
where, as a reminder that the above quantity is a vector, we have
returned momentarily to the more exhaustive notation for 
indexing the columns of $\Phi$, as in the right-hand side of
\eqref{eq:basis_exp}. 

Straightforward calculations reveal that, for each $i=1,\ldots,n$, and
$p=1,\ldots,m$,
\begin{equation*}
\Expect \Phi_{i,(j-1)m+p} v_{\ell} (\epsilon_{\ell k i}) = 0, 
\quad \text{and} \quad
-| \Phi_{i,(j-1)m+p} | \leq \Phi_{i,(j-1)m+p} v_{\ell} (\epsilon_{\ell
  k i}) \leq | \Phi_{i,(j-1)m+p} |.  
\end{equation*}
Hence,
\begin{align*}
\P \Big( \| \Phi_j^T v_\ell(\epsilon_{\ell k i}) \|_2 \geq 
\sqrt{m} t \Big) &\leq
\P \bigg( \bigg|\sum_{i=1}^n \Phi_{i,(j-1)m+p} 
v_{\ell} (\epsilon_{\ell k i}) \bigg| \geq t, \;\,
\text{some $p=1,\ldots,m$} \bigg) \\
&\leq \sum_{p=1}^m
2 \exp \bigg( -\frac{t^2}{2 \sum_{i=1}^n 
\Phi_{i,(j-1)m+p}^2} \bigg) \\
&\leq 2m \exp\bigg( -\frac{t^2}{2n} \bigg).
\end{align*}
Above, the first inequality used the simple fact that $\|x\|_2 \leq
\sqrt{m} \|x\|_\infty$ for \smash{$x \in \R^m$}; the second used
Hoeffding's bound and the union bound; and the third used our
assumption that the columns of $\Phi$ have norm at most $\sqrt{n}$.
Therefore, taking \smash{$t=\lambda_1/(3\sqrt{m})$}, we see that, by
the above and the union bound,
\begin{equation*}
\P \bigg(\max_{j \in S_{\ell k}^c} \; 
\| \Delta^j_1 \|_2 < \frac{\lambda_1}{3} \bigg) \geq 1 - 2dm 
\exp\bigg(-\frac{\lambda_1^2}{18mn}\bigg). 
\end{equation*}
By choosing 
\smash{$\lambda_1 = C' \sqrt{ 18mn \log(8 d^2mr / \delta) }$} for a constant $C'>0$, we  
see that the probability in question is at least $1-\delta/(4dr)$, as desired.

\bigskip
{\it Bounding \smash{$\| \Delta_2^j \|_2$}.}  Recall that
\smash{$F_{\epsilon_{\ell k}}(\cdot)$} is used to denote the CDF of 
the effective error distribution, and 
\smash{$f_{\epsilon_{\ell k}}(\cdot)$} 
is used for its its density.  By construction,
\smash{$F_{\epsilon_{\ell k}} (0) = \alpha_\ell$}.  Direct
calculation, using the definition of $v_\ell(\cdot)$, shows that, 
for any $\vartheta \in B^*$, and each $i=1,\ldots,n$,
\begin{equation*}
\Expect \Big[ v_{\ell} (\epsilon_{\ell k}) - v_{\ell} 
\Big(Y_k - \Phi_{S_{\ell k}} \vartheta \Big) \Big] = 
F_{\epsilon_{\ell k}} \Big( \Phi_{S_{\ell k}}\big(\vartheta -
\theta_{\ell k (S_{\ell k})}^*\big) \Big) - F_{\epsilon_{\ell k}} ({\mathbf 0}),
\end{equation*} 
where we apply $F_{\epsilon_{\ell k}}$ componentwise, and so
\begin{equation*}
\Phi_j^T \Expect \Big[ v_{\ell} (\epsilon_{\ell k}) - v_{\ell} 
\Big(Y_k - \Phi_{S_{\ell k}} \vartheta \Big) \Big] =
f_{\epsilon_{\ell k}} (0) \Phi_j^T \Phi_{S_{\ell k}} \big(\vartheta - 
\theta_{\ell k (S_{\ell k})}^*\big) + \Delta_4^j 
\end{equation*}
with \smash{$\Delta_4^j \in \reals^m$} being the appropriate
remainder term, \ie, 
\begin{equation*}
\left[ \Delta_4^j \right]_t = \sum_{i=1}^n \Phi_{i t}
\Big[ F_{\epsilon_{\ell k}} \Big( \Phi_{i (S_{\ell k})} \big(
\vartheta - \theta_{\ell k (S_{\ell k})}^* \big) \Big) - 
F_{\epsilon_{\ell k}} (0) - f_{\epsilon_{\ell k}} (0) 
\Phi_{i (S_{\ell k})} 
\big ( \vartheta - \theta_{\ell k (S_{\ell k})}^*\big) \Big],
\end{equation*}
for $t = j(m-1)+1, \ldots, jm$.

Now, we have that
\begin{equation*}
\big\| 
f_{\epsilon_{\ell k}} (0) \Phi_j^T \Phi_{S_{\ell k}} \big(\vartheta - 
\theta_{\ell k (S_{\ell k})}^*\big) \big\|_2 \leq 
f_{\epsilon_{\ell k}} (0) \big\|\Phi_j^T \Phi_{S_{\ell k}}\big\|_F
\big\|\vartheta - \theta_{\ell k (S_{\ell k})}^*\big\|_2 \leq 
\frac{\lambda_1}{6},
\end{equation*}
where we have used 
\smash{$\|\vartheta - \theta_{\ell k (S_{\ell k})}^*\|_2 \leq \gamma$}
and the groupwise irrepresentability condition in A1.

We also have the following two facts, which we will use momentarily:
\begin{align}
\Phi_{\max}^3 n s \gamma^2 & = o(\lambda_1) \label{eq:4} \\
\sqrt{s} \Phi_{\max} \gamma & \to 0. \label{eq:5}
\end{align}
Note that \eqref{eq:4} can be obtained as follows.  Since $(1/2)(x+y)^2 \leq x^2 + y^2$ for $x,y \in \reals$, we can plug in
\[
\gamma = \ourgamma,
\]
and check that both terms on the right-hand side of
\[
\frac{\Phi_{\max}^3 n s}{\lambda_1} \left( \frac{\lambda_1^2 s^2 m}{n^2} + \frac{s \log n}{n} \right) = \frac{\Phi_{\max}^3 s^3 \lambda_1 m}{n} + \frac{\Phi_{\max}^3 s^2 \log n}{\lambda_1}
\]
tend to zero.  For the first term on the right-hand side, it is enough to show that
\[
\Phi_{\max}^6 s^6 m^3 \log(d^2 m r) (\log^3 n) / n \to 0,
\]
where we have plugged in \smash{$\lambda_1 = C'\sqrt{mn \log( d^2mr / \delta) \log^3 n}$}.  Using the assumptions in condition A4, we get that $\log(d^2 m r) = O(\log d + \log m) = O(n^{2/21})$, and furthermore that
\[
\Phi_{\max}^6 s^6 m^3 \log(d^2 m r) (\log^3 n) / n = o\left( \frac{n^{1/3} \cdot n^{2/21} \cdot n^{6/21} \cdot n^{6/21}}{\log^3 n} \right) \frac{\log^3 n}{n} \to 0,
\]
as required.  A similar calculation shows that the second term on the right-hand side also tends to zero, \ie, $\Phi_{\max}^3 s^2 (\log n) / \lambda_1 \to 0$, which establishes \eqref{eq:4}.  Lastly, \eqref{eq:5} follows since its left-hand side is dominated by the left-hand side of \eqref{eq:4}.

So, we now compute
\begin{align*}
\|\Delta_4^j\|_2 &\leq \sqrt{m} \max_{t} \; \sum_{i=1}^n \bigg|\Phi_{i t}  
\Big[ F_{\epsilon_{\ell k}} \Big( \Phi_{i (S_{\ell k})} \big(
\vartheta - \theta_{\ell k (S_{\ell k})}^* \big) \Big) - {} \\
& \hspace{200pt}
F_{\epsilon_{\ell k}} (0) - f_{\epsilon_{\ell k}} (0) \Phi_{i (S_{\ell k})} 
\big ( \vartheta - \theta_{\ell k (S_{\ell k})}^*\big) \Big] \bigg| 
\\ & \leq C_1 \Phi_{\max} \sqrt{m} \sum_{i=1}^n \Big(
\Phi_{i (S_{\ell k})} \big(\vartheta - \theta_{\ell k (S_{\ell k})}^*
\big) \Big)^2 \\
& \leq C_1 \Phi_{\max} \sqrt{m} \sum_{i=1}^n 
\| \Phi_{i (S_{\ell k})} \|_2^2 
\| \vartheta - \theta_{\ell k (S_{\ell k})}^* \|_2^2 \\  
& \leq C_1 \Phi_{\max}^3 \sqrt{m} n s \gamma^2 \\ 
& = o ( \lambda_1 ).
\end{align*}
Here the first inequality follows from the fact that $\|x\|_2 \leq 
\sqrt{m} \|x\|_\infty$ for \smash{$x \in \R^m$}, and the triangle
inequality; the second follows from the distributional smoothness 
condition in A2, which is applicable since 
\eqref{eq:5} holds; the third
uses Cauchy-Schwarz; the fourth uses our column norm assumption on
$\Phi$, 
and \smash{$\|\vartheta - \theta_{\ell k (S_{\ell k})}^*\|_2 \leq
  \gamma$}; the last uses \eqref{eq:4}.  As
\smash{$\|\Delta_4^j\|_2 = o(\lambda_1)$}, it will certainly be
strictly less than $\lambda_1/6$ for $n$ large enough.  We have hence
shown, noting that none of our above arguments have depended on the
particular choice of $j=1,\ldots,d$ or $\vartheta \in B^*$,
\begin{equation*}
\max_{j \in S_{\ell k}^c} \; \sup_{\vartheta \in B^*} \;
\| \Delta^j_2 \|_2 < \frac{\lambda_1}{3}.
\end{equation*}

\bigskip
{\it Bounding \smash{$\| \Delta_3^j \|_2$}.}  For this part, we can
use the end of the proof of Lemma 2 in \cite{fan2014adaptive}, which
uses classic entropy-based techniques to establish a bound very
similar to that which we are seeking. By carefully
looking at the conditions required for this lemma, we see that under
the distributional smoothness condition in A2, condition A3, and also 
\begin{align*}
\sqrt{n \log(dm)} & = o(\lambda_1) \\
n \Phi_{\max} \gamma^2 & = o(\lambda_1) \\
(1 + \gamma \Phi_{\max}^2 s^{3/2}) \log^2 n & = o(\lambda_1^2 / n),
\end{align*}
all following directly from condition A4 by calculations similar to the ones we used when bounding $\| \Delta_2^j \|$, we have 
\begin{equation*}
\P \bigg(
\max_{j \in S_{\ell k}^c} \; \sup_{\vartheta \in B^*} \;
\|\Delta_3^j\|_2 \geq \frac{\lambda_1}{3} \bigg) \leq 
\P \bigg(
\max_{j \in S_{\ell k}^c} \; \sup_{\vartheta \in B^*} \;
\|\Delta_3^j\|_\infty \geq \frac{\lambda_1}{3\sqrt{m}} \bigg);
\end{equation*} 
the probability on the right-hand side can be made arbitrarily
small for large $n$, by the arguments at the end of Lemma 2 in
\cite{fan2014adaptive}, and hence clearly smaller than the desired
\smash{$\delta/(4dr)$} level.

\bigskip
{\it Putting it together.}  Returning to the logic in
\eqref{eq:subg_bound}, \eqref{eq:ball_bound}, \eqref{eq:z_expand}, we
have shown that the subgradient condition in \eqref{eq:subg_bound}
holds with probability at least $1-(\delta/(2dr) + \delta/(4dr) +
\delta/(4dr)) = 1-\delta/(dr)$.  Taking a union bound over
$k=1,\ldots,d$ and $\ell=1,\ldots,r$, which were considered fixed at
the start of our analysis, gives the result stated in the theorem.
\hfill\qedsymbol

\section{Statement and proof of Lemma \ref{lem:estimation}}
\label{supp:estimation}

\newcommand{\subs}{{\ell k (S_{\ell k})}}

We show that with probability at least $1 - \delta / (2dr)$, it holds that \smash{$\tilde{\theta}_\subs \in B^*$}, where \smash{$\tilde{\theta}_\subs$} is the solution to the restricted problem \eqref{eq:restricted}, for some fixed $k=1,\ldots,d$ and $\ell=1,\ldots,r$, and $B^*$ is a ball defined in the proof of Theorem \ref{thm:recovery} in Section \ref{supp:recovery}.  This fact is used a few times in the proof of Theorem \ref{thm:recovery}.

\begin{lemma}
\label{lem:estimation}
Fix some $k=1,\ldots,d$ and $\ell=1,\ldots,r$.  Define the ball
\[
B^* = \{ \vartheta \in \reals^{sm} : \| \vartheta - \theta^*_\subs \|_2 \leq \gamma \}
\]
centered at the underlying coefficients \smash{$\theta^*_\subs$} with radius 
\[
\gamma = \ourgamma,
\]
for some constant $C > 0$.  Then, with probability at least $1 - \delta / (2dr)$, it holds that \smash{$\tilde{\theta}_\subs \in B^*$}, where \smash{$\tilde{\theta}_\subs$} is the solution to the restricted problem \eqref{eq:restricted}.
\end{lemma}

\begin{proof}
We will follow the strategy for the proof of Theorem 1 in \cite{fan2014adaptive} closely.  We begin by considering the ball
\[
B = \{ \vartheta \in \reals^{sm} : \| \vartheta - \theta^*_\subs \|_2 \leq R \}
\]
with center \smash{$\theta^*_\subs$} and radius $R$.  We also introduce some useful notational shorthand, and write the quantile loss term in the restricted problem \eqref{eq:restricted} as
\[
L_{\ell k}( \vartheta ) = \psi_{\alpha_\ell} \left( Y_k - \Phi_{S_{\ell k}} \vartheta \right).
\]
Below, we show that a particular function of $R$ serves as an upper bound for the quantity \smash{$\Expect [ L_{\ell k}( \tilde{\vartheta}_\subs ) - L_{\ell k}( \theta^*_\subs ) ]$}, where the expectation here is taken over draws of the data, and \smash{$\tilde{\vartheta}_\subs$} is a particular point in $B$ that we define in a moment.  This in turn implies, with probability at least $1 - \delta/(2dr)$, that \smash{$\tilde{\theta}_\subs \in B^*$}, as claimed.

\bigskip
First, we define \smash{$\tilde{\vartheta}_\subs$} more precisely: it is a point on the line segment between the solution to the restricted problem \smash{$\tilde{\theta}_\subs$} and the underlying coefficients \smash{$\theta^*_\subs$}, \ie,
\[
\tilde{\vartheta}_\subs = \beta \tilde{\theta}_\subs + (1-\beta) \theta^*_\subs,
\]
for a particular choice
\[
\beta = \frac{R}{R + \| \tilde{\theta}_\subs - \theta^*_\subs \|_2},
\]
which guarantees that \smash{$\tilde{\vartheta}_\subs \in B$} even if \smash{$\tilde{\theta}_\subs \notin B$} , as  we establish next.  Observe that we always have
\begin{align*}
& \| \tilde{\theta}_\subs - \theta^*_\subs \|_2  \leq R + \| \tilde{\theta}_\subs - \theta^*_\subs \|_2 \\
\iff \; & R \frac{\| \tilde{\theta}_\subs - \theta^*_\subs \|_2}{R + \| \tilde{\theta}_\subs - \theta^*_\subs \|_2}  \leq R \\
\iff \; & \beta \| \tilde{\theta}_\subs - \theta^*_\subs \|_2 \leq R \\
\iff \; & \| \beta \tilde{\theta}_\subs - \beta \theta^*_\subs + \theta^*_\subs - \theta^*_\subs \|_2 \leq R \\
\iff \; & \| \tilde{\vartheta}_\subs - \theta^*_\subs \|_2 \leq R,
\end{align*}
as claimed.  The second line here follows by rearranging and multiplying through by $R$; the third by using the definition of $\beta$ above; the fourth by adding and subtracting the underlying coefficients; and the fifth by using the definition of \smash{$\tilde{\vartheta}_\subs$}.

\bigskip
Now, the beginning of the proof of Theorem 1 in \cite{fan2014adaptive} establishes, for any \smash{$\tilde{\vartheta}_\subs \in B$}, for some constant $C_5 > 0$, and using condition A3, that
\begin{equation}
\Expect \left[ L_{\ell k}( \tilde{\vartheta}_\subs ) - L_{\ell k}( \theta^*_\subs ) \right] \geq C_5 n \| \tilde{\vartheta}_\subs - \theta^*_\subs \|_2^2, \label{eq:balls}
\end{equation}
and so, by direct calculation, since
\begin{equation}
\| \tilde{\vartheta}_\subs - \theta^*_\subs \|_2 \leq R \iff \beta \| \tilde{\theta}_\subs - \theta^*_\subs \|_2 \leq R \iff \| \tilde{\theta}_\subs - \theta^*_\subs \|_2 \leq R / 2, \label{eq:balls2}
\end{equation}
it suffices to obtain a suitable upper bound for \smash{$\Expect [ L_{\ell k}( \tilde{\vartheta}_\subs ) - L_{\ell k}( \theta^*_\subs ) ]$}, in order to get the result in the statement of the lemma.  To this end, we introduce one more piece of shorthand, and denote the objective for the restricted problem \eqref{eq:restricted} as $J_{\ell k}(\vartheta)$.

\bigskip
We proceed with the following chain of (in)equalities:
\begin{align}
& \Expect \left[ L_{\ell k}(\tilde{\vartheta}_\subs) - L_{\ell k}(\theta^*_\subs) \right] \notag \\
& = \Expect \left[ L_{\ell k}(\tilde{\vartheta}_\subs) - L_{\ell k}(\theta^*_\subs) \right] + J_{\ell k}(\tilde{\vartheta}_\subs) - J_{\ell k}(\tilde{\vartheta}_\subs) + {} \notag \\
& \hspace{270pt} J_{\ell k}(\theta^*_\subs) - J_{\ell k}(\theta^*_\subs) \label{eq:expand} \\
& = \underbrace{ L_{\ell k}(\theta^*_\subs) - \Expect L_{\ell k}(\theta^*_\subs) - L_{\ell k}(\tilde{\vartheta}_\subs) + \Expect L_{\ell k}(\tilde{\vartheta}_\subs) }_{ \Delta(\theta^*_\subs, \tilde{\vartheta}_\subs) } + {} \notag \\
& \qquad  J_{\ell k}(\tilde{\vartheta}_\subs) - J_{\ell k}(\theta^*_\subs) + 
\lambda_1 \sum_{j \in S_{\ell k}} \| \theta_{\ell k j}^* \|_2 - \lambda_1 \sum_{j \in S_{\ell k}} \| \tilde{\vartheta}_{\ell k j} \|_2 {} \notag \\
& \hspace{210pt}
- (\lambda_2 / 2) \| \tilde{\vartheta}_\subs \|_2^2 + (\lambda_2 / 2) \| \theta^*_\subs \|_2^2 \label{eq:rearr} \\
& \leq \Delta(\theta^*_\subs, \tilde{\vartheta}_\subs) + J_{\ell k}(\tilde{\vartheta}_{\ell k (S_{\ell k})}) - J_{\ell k}(\theta_{\ell k (S_{\ell k})}^*) + \lambda_1 \sum_{j \in S_{\ell k}} \| \theta_{\ell k j}^* - \tilde{\vartheta}_{\ell k j} \|_2 + o(1) \label{eq:trieq} \\
& \leq \Delta(\theta^*_\subs, \tilde{\vartheta}_\subs) + J_{\ell k}(\tilde{\vartheta}_\subs) - J_{\ell k}(\theta_{\ell k (S_{\ell k})}^*) + \lambda_1 s R \sqrt{m} + o(1) \label{eq:ball} \\
& \leq \Delta(\theta^*_\subs, \tilde{\vartheta}_\subs) + 2\lambda_1 s R \sqrt{m} \label{eq:cvxity} \\
& \leq \sup_{\tilde{\vartheta}_\subs \in B} | \Delta(\theta^*_\subs, \tilde{\vartheta}_\subs) | + 2\lambda_1 s R \sqrt{m}. \label{eq:conclusion}
\end{align}
Here, \eqref{eq:expand} follows by adding and subtracting like terms, and \eqref{eq:rearr} by rearranging \eqref{eq:expand}. In \eqref{eq:trieq} we use the triangle inequality and the following argument to show that the terms involving $\lambda_2$ are $o(1)$.  Under the assumption that \smash{$\lambda_2=o(n^{41/42} / \theta_{\max}^*)$}, combined with the restriction that $s=o(n^{1/21})$, we have \smash{$\lambda_2=o(n/(\sqrt{s}\theta_{\max}^*))$}. Therefore, under our choice of $R=1/n$ (as specified below), we have 
\begin{equation*}
\lambda_2 \sqrt{s}\theta_{\max}^* R \to 0.
\end{equation*}
This in turn is used to argue that
\begin{align*}
- (\lambda_2 / 2) \| \tilde{\vartheta}_\subs \|_2^2 + (\lambda_2 / 2) \| \theta^*_\subs \|_2^2 &= 
(\lambda_2/2) \| \tilde{\vartheta}_\subs  - \theta^*_\subs \|_2^2 {} \\
& \hspace{85pt} - \lambda_2 \| \tilde{\vartheta}_\subs \|_2^2 + \lambda_2 \tilde{\vartheta}_\subs^T \theta^*_\subs \\
&\leq (\lambda_2/2) R^2 - \lambda_2 \| \tilde{\vartheta}_\subs \|_2 ( \| \tilde{\vartheta}_\subs \|_2 - \| \theta^*_\subs \|_2) \\
& \leq (\lambda_2/2) R^2 + \lambda_2 \| \tilde{\vartheta}_\subs \|_2 R \\
& \leq (\lambda_2/2) R^2 + \lambda_2 \| \theta^*_\subs \|_2 R \\
&\leq  (\lambda_2/2) R^2 + \lambda_2 \sqrt{s} \theta^*_{\max} R  \to 0.
\end{align*}
In the second to last line, we have applied \smash{$\| \tilde{\vartheta}_\subs \|_2 \leq \| \theta^*_\subs \|_2$}, as, outside of this case, the term in question \smash{$- (\lambda_2 / 2) \| \tilde{\vartheta}_\subs \|_2^2 + (\lambda_2 / 2) \| \theta^*_\subs \|_2^2$} would be negative, anyway.

Continuing on, \eqref{eq:ball} holds because \smash{$\| \theta^*_\subs - \tilde{\vartheta}_\subs \|_2 \leq R$} implies \smash{$\| \theta^*_{\ell k j} - \tilde{\vartheta}_{\ell k j} \|_2 \leq R$}. Finally, \eqref{eq:cvxity} follows because of the following argument.  Since $J_{\ell k}$ is convex, we can use the definition of \smash{$\tilde{\vartheta}_\subs$} and get
\[
J_{\ell k}(\tilde{\vartheta}_\subs) \leq \beta J_{\ell k}(\tilde{\theta}_\subs) + (1-\beta) J_{\ell k}(\theta^*_\subs) = J_{\ell k}(\theta^*_\subs) + \beta ( J_{\ell k}(\tilde{\theta}^*_\subs) - J_{\ell k}(\theta^*_\subs) );
\]
notice that the last term here is nonpositive, since \smash{$\tilde{\theta}_\subs$} is the solution to the restricted problem \eqref{eq:restricted}, and thus we have that
\[
J_{\ell k}(\tilde{\vartheta}_\subs) \leq J_{\ell k}(\theta^*_{\ell k (S_{\ell k})}),
\]
which lets us move from \eqref{eq:ball} to \eqref{eq:cvxity}.

\bigskip
Lemma 1 in \cite{fan2014adaptive} states, with probability at least $1-\delta$, where $\delta = \exp(-C_6 s \log n)$ and $C_6 > 0$ is some constant, that
\[
\sup_{\tilde{\vartheta}_\subs \in B} | \Delta(\theta^*_\subs, \tilde{\vartheta}_\subs) | \leq 6 R \sqrt{ s n \log n },
\]
so from \eqref{eq:conclusion}, with probability at least $1-\delta$, we see that
\[
\Expect \left[ L_{\ell k}(\tilde{\vartheta}_\subs) - L_{\ell k}(\theta^*_\subs) \right] \leq 6 R \sqrt{s n \log n} + 2 \lambda_1 s R \sqrt{m}
\]
and, using \eqref{eq:balls}, that
\[
n \| \tilde{\vartheta}_\subs - \theta^*_\subs \|_2^2 \leq C' \left( R \sqrt{s n \log n} + \lambda_1 s R \sqrt{m} \right),
\]
for some constant $C' > 0$.

\bigskip
Plugging in $R = 1/n$, dividing through by $n$, and using the fact that the square root function is subadditive, we get, with probability at least $1-\delta$, that
\begin{align*}
\| \tilde{\vartheta}_\subs - \theta^*_\subs \|_2 & \leq C' \left( \frac{(s \log n)^{1/4}}{n^{3/4}} + \frac{(\lambda_1 s)^{1/2} m^{1/4}}{n} \right) \\
& \leq C' \left( \sqrt{\frac{s \log n}{n}} + \frac{\lambda_1 s \sqrt{m}}{n} \right).
\end{align*}

Finally, we complete the proof by applying \eqref{eq:balls2}, in order to get that
\[
\| \tilde{\theta}_\subs - \theta^*_\subs \|_2 \leq \gamma,
\]
where we have defined
\[
\gamma = \ourgamma,
\]
and $C > 0$ is some constant, with probability at least $1-\delta/(2dr)$, for large enough $n$.
\end{proof}

\section{Proof of Lemma \ref{lem:prox}}
\label{supp:prox}

The prox operator \smash{$\prox_{\lambda \psi_{\mathcal{A}}}(A)$} is separable in the entries of its minimizer $X$, so we focus on minimizing over $X_{ij}$ the expression 
\begin{align}
& \max \{ \alpha_j X_{ij}, (\alpha_j - 1) X_{ij} \} + (1 / (2 \lambda)) \left( X_{ij} - A_{ij} \right)^2 \notag \\
& \quad = \alpha_j \max \{ 0, X_{ij} \} + (1 - \alpha_j) \max \{ 0, - X_{ij} \} + (1 / (2 \lambda)) \left( X_{ij} - A_{ij} \right)^2. \label{eq:prox} 
\end{align}
Suppose $X_{ij} > 0$.  Then differentiating \eqref{eq:prox} gives $X_{ij} = A_{ij} - \lambda \alpha_j$ and the sufficient condition $A_{ij} > \lambda \alpha_j$.  Similarly, assuming $X_{ij} < 0$ gives $ X_{ij} = A_{ij} + \lambda (1 - \alpha_j)$ when $A_{ij} < \lambda (\alpha_j - 1)$.  Otherwise, we can take $X_{ij} = 0$.  Putting these cases together gives the result.
\hfill\qedsymbol


\section{ADMM for the \mqgm}
\label{sec:admm}

A complete description of our ADMM-based algorithm for fitting the \mqgm~to data is given in Algorithm \ref{alg:admm}.

\begin{algorithm}
\caption{ADMM for the \mqgm}
\label{alg:admm}
\begin{algorithmic}
	\STATE {\bf Input:} observations \smash{$y^{(1)},\ldots,y^{(n)}
		\in \reals^d$}, feature matrix $\Phi \in \reals^{n \times dm}$,
	quantile levels $\mathcal{A}$, constants $\lambda_1,\lambda_2>0$
	\STATE {\bf Output:} fitted coefficients 
	\smash{$\hat\Theta = (\hat\theta_{\ell kj},
		\hat b_{\ell k})$} 
	\FOR{$k=1,\ldots,d$ (in parallel, if possible)}
	\STATE {\bf initialize} $\Theta_k, B_k, V, W, Z, U_V,
	U_W, U_Z$ 
	\REPEAT
	
	\STATE update $\Theta_k$ using \eqref{eq:theta}
	\STATE update $B_k$ using \eqref{eq:theta}
	\STATE update $V$ using \eqref{eq:V}
	\STATE update $W$ using \eqref{eq:W}
	\STATE update $Z$ using \eqref{eq:Z} and Lemma \ref{lem:prox}
	\STATE update $U_V, U_W, U_Z$:
	\begin{align*}
	U_V & \leftarrow U_V + (\mathbf{1} B_k^T + 
	\Phi_k \Theta - V) \\ 
	U_W & \leftarrow U_W + (\Theta_k - W) \\
	U_Z & \leftarrow U_Z + (Y_k\mathbf{1}^T -
	\mathbf{1} B_k^T - \Phi_k \Theta - Z) 
	\end{align*}
	\UNTIL converged
	\ENDFOR
\end{algorithmic}
\end{algorithm}

\section{Additional details on Gibbs sampling}

\label{sec:gibbs:more}
In the \mqgm, there is no analytic solution
for parameters like the mean, median, or quantiles of these marginal
and conditional distributions, but the pseudolikelihood approximation
makes for a very efficient Gibbs sampling procedure, which we
highlight in this section.  As it is relevant to the computational
aspects of the approach, in this subsection we will make explicit the
conditional random field, where $y_k$ depends on both $y_{\neg k}$ and
fixed input features $x$. 

First, note that since we are representing the distribution of $y_k |
y_{\neg k},x$ via its inverse CDF, to sample from from this
conditional distribution we can simply generate a random $\alpha \sim
\textrm{Uniform}(0,1)$. We then compute
\begin{equation*}
\begin{split}
\hat{Q}_{y_k | y_{\neg k}}(\alpha_\ell) & = \phi(y)^T \theta_{\ell k } + x^T \theta^x_{\ell k } \\
\hat{Q}_{y_k | y_{\neg k}}(\alpha_{\ell+1}) & = \phi(y)^T \theta_{(\ell+1) k } + x^T \theta^x_{(\ell+1) k }
\end{split}
\end{equation*}
\vspace{-0.1in}
for some pair $\alpha_\ell \leq \alpha \leq \alpha_{\ell+1}$ and set $y_k$ to be a linear interpolation of the two values
\begin{equation*}
y_k \leftarrow \hat{Q}_{y_k | y_{\neg k}}(\alpha_\ell) + \frac{\left(\hat{Q}_{y_k | y_{\neg k}}(\alpha_{\ell+1}) - \hat{Q}_{y_k | y_{\neg k}}(\alpha_\ell) \right)(\alpha - \alpha_\ell)}{{\alpha_{\ell+1}} - \alpha_\ell}. 
\end{equation*}

This highlights the desirability of having a range of non-uniformly spaced $\alpha$ terms that reach values close to zero and one as otherwise we may not be able to find a pair of $\alpha$'s that lower and upper bound our random sample $\alpha$.  However, in the case that we model a sufficient quantity of $\alpha$, a reasonable approximation (albeit one that will not sample from the extreme tails) is also simply to pick a random $\alpha_\ell \in \mathcal{A}$ and use just the corresponding column $\theta_{\ell k}$ to generate the random sample.

Computationally, there are a few simple but key points involved in
making the sampling efficient.  First, when sampling from a
conditional distribution, we can precompute $x^T \Theta^x_k$ for each
$k$, and use these terms as a constant offset.  Second, we maintain a
``running'' feature vector $\phi(y) \in \reals^{dm}$, \ie, the
concatenation of features corresponding to each coordinate
$\phi(y_k)$. Each time we sample a new coordinate $y_k$, we generate
just the new features in the $\phi(y_k)$ block, leaving the remaining
features untouched.  Finally, since the $\Theta_k$ terms are sparse,
the inner product $\phi(y)^T \theta_{\ell k}$ will only contain
a few nonzeros terms in the sum, and will be computed more efficiently
if the $\Theta_k$ are stored as a sparse matrices.

\section{Additional details on the evaluation of fitted conditional CDFs}

Here, we elaborate on the evaluation of each method's conditional CDFs that we first presented in Section \ref{sec:synth}.  For simplicity, we describe everything that follows in terms of the conditional CDF $y_1 | y_2$ only, with everything being extended in the obvious way to other conditionals.  (We omit the \nonparas~from our evaluation as it is not clear how to sample from its conditionals, due to the nature of a particular transformation that it uses.)

\bigskip
First, we carried out the following steps in order to compute the true (empirical) conditional CDF.
\begin{enumerate}
\item We drew $n = 400$ samples from the ring distribution, by following the procedure described in Section \ref{sec:synth}; these observations are plotted across the top row of Figure \ref{fig:calib}. 
\item We then partitioned the $y_2$ samples into five equally-sized bins, and computed the true empirical conditional CDF of $y_1$ given each bin of $y_2$ values.
\end{enumerate}

\bigskip
Next, we carried out the following steps in order to compute the estimated (empirical) conditional CDFs, for each method.
\begin{enumerate}
\setcounter{enumi}{2}
\item We fitted each method to the samples obtained in step (1) above.
\item Then, for each method, we drew a sample of $y_1$ given {\it each} $y_2$ sample, using the method's conditional distribution; these conditionals are plotted across the second through fifth rows of Figure \ref{fig:calib} (for representative values of $\lambda_1$).

Operationally, we drew samples from each method's conditionals in the following ways.
\begin{itemize}
\item \mqgm: we used the Gibbs sampler described in Section \ref{sec:gibbs:more}.
\item \mb: we drew $y_1 \sim \mathcal{N}( \hat \theta_1^T y_2^{(i)}, \hat \sigma_{1|2}^2 )$, where $\hat \theta_1$ is the fitted lasso regression coefficient of $y_1$ on $y_2$; $y_2^{(i)}$ for $i=1,\ldots,n$ is the $i$th observation of $y_2$ obtained in step (1) above; and $\hat \sigma_{1|2}^2 = \var( Y_1 - Y_2 \hat \theta_1 )$ denotes the sample variance of the underlying error term $Y_1 - Y_2 \hat \theta_1$ with $Y_i = ( y_i^{(1)}, \ldots, y_i^{(n)} ) \in \reals^n$ collecting all observations along variable $i$.
\item \spacejam: we drew $y_1 \sim \mathcal{N}( \hat \theta_1^T \phi(y_2^{(i)}), \hat \sigma_{1|2}^2 )$, where $\phi$ is a suitable basis function, and $\hat \theta_1$ as well as $\hat \sigma^2_{1|2}$ are defined in ways analogous to the neighborhood selection setup.
\item \glasso: we drew $y_1 \sim \mathcal{N}( \hat \mu_{1|2}, \hat \sigma^2_{1|2} )$, where
\begin{align*}
\hat \mu_{1|2} & = \hat \mu_1 + \hat \Sigma_{12} \hat \Sigma_{22}^{-1} ( y_2^{(i)} - \hat \mu_2 ) \\
\hat \sigma^2_{1|2} & = \hat \Sigma_{11} - \hat \Sigma_{12} \hat \Sigma_{22}^{-1} \hat \Sigma_{21}
\end{align*}
with $\hat \mu_i$ denoting the sample mean of $Y_i$, and $\hat \Sigma$ denoting the estimate of the covariance matrix given by \glasso~(subscripts select blocks of this matrix).
\end{itemize}

\item Finally, we partitioned the $y_2$ samples into five equally-sized bins (just as when computing the true conditional CDF), and computed the estimated empirical conditional CDF of $y_1$ given each bin of $y_2$ values.
\end{enumerate}

\bigskip
Having computed the true as well as estimated conditional CDFs, we measured the goodness of fit of each method's conditional CDFs to the true conditional CDFs, by computing the total variation (TV) distance, \ie, 
\[
(1/2) \sum_{i=1}^q \left| \hat F^{\textrm{method}_j}_{y_1 | y_2}( z^{(i)} | \zeta ) - \hat F^{\textrm{true}}_{y_1 | y_2}( z^{(i)} | \zeta ) \right|,
\]
as well as the (scaled) Kolmogorov-Smirnoff (KS) statistic, \ie,
\[
\max_{i=1,\ldots,q} \left| \hat F^{\textrm{method}_j}_{y_1 | y_2}( z^{(i)} | \zeta ) - \hat F^{\textrm{true}}_{y_1 | y_2}( z^{(i)} | \zeta ) \right|.
\]
Here, \smash{$\hat F^{\textrm{true}}_{y_1 | y_2}( z^{(i)} | \zeta )$} is the true empirical conditional CDF of $y_1 | y_2$, evaluated at $y_1 = z^{(i)}$ and given $y_2 = \zeta$, and \smash{$\hat F^{\textrm{method}_j}_{y_1 | y_2}( z^{(i)} | \zeta )$} is a particular method's (``$\textrm{method}_j$'' above) estimated empirical conditional CDF, evaluated at $y_1 = z^{(i)}$ and given $y_2 = \zeta$.  For each method, we averaged these TV and KS values across the method's conditional CDFs.  Table \ref{tab:calib} reports the best (across a range of tuning parameters) of these averaged TV and KS values.

\begin{figure}[p]
\centering
\includegraphics[width=0.32\columnwidth]{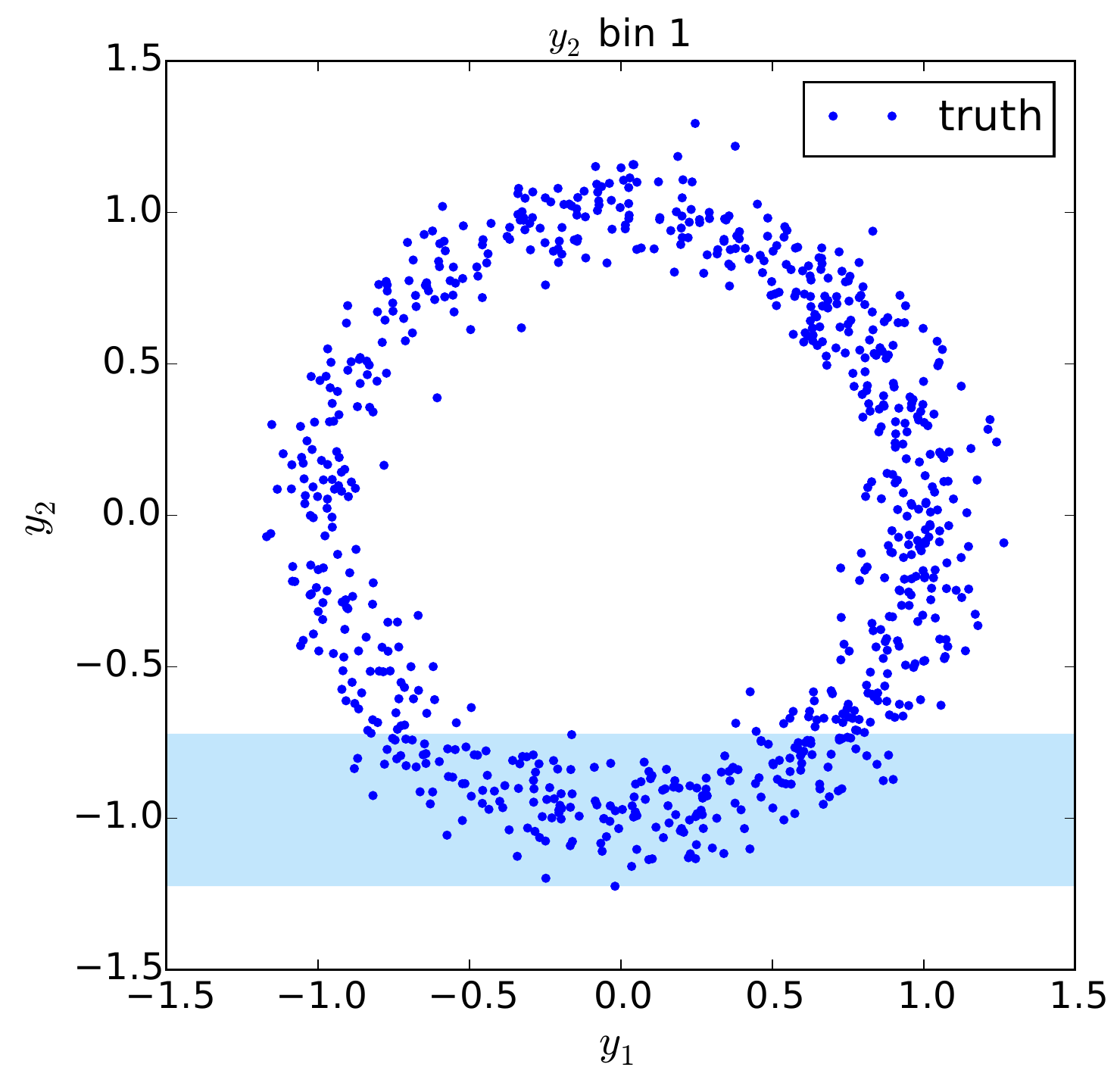}
\includegraphics[width=0.32\columnwidth]{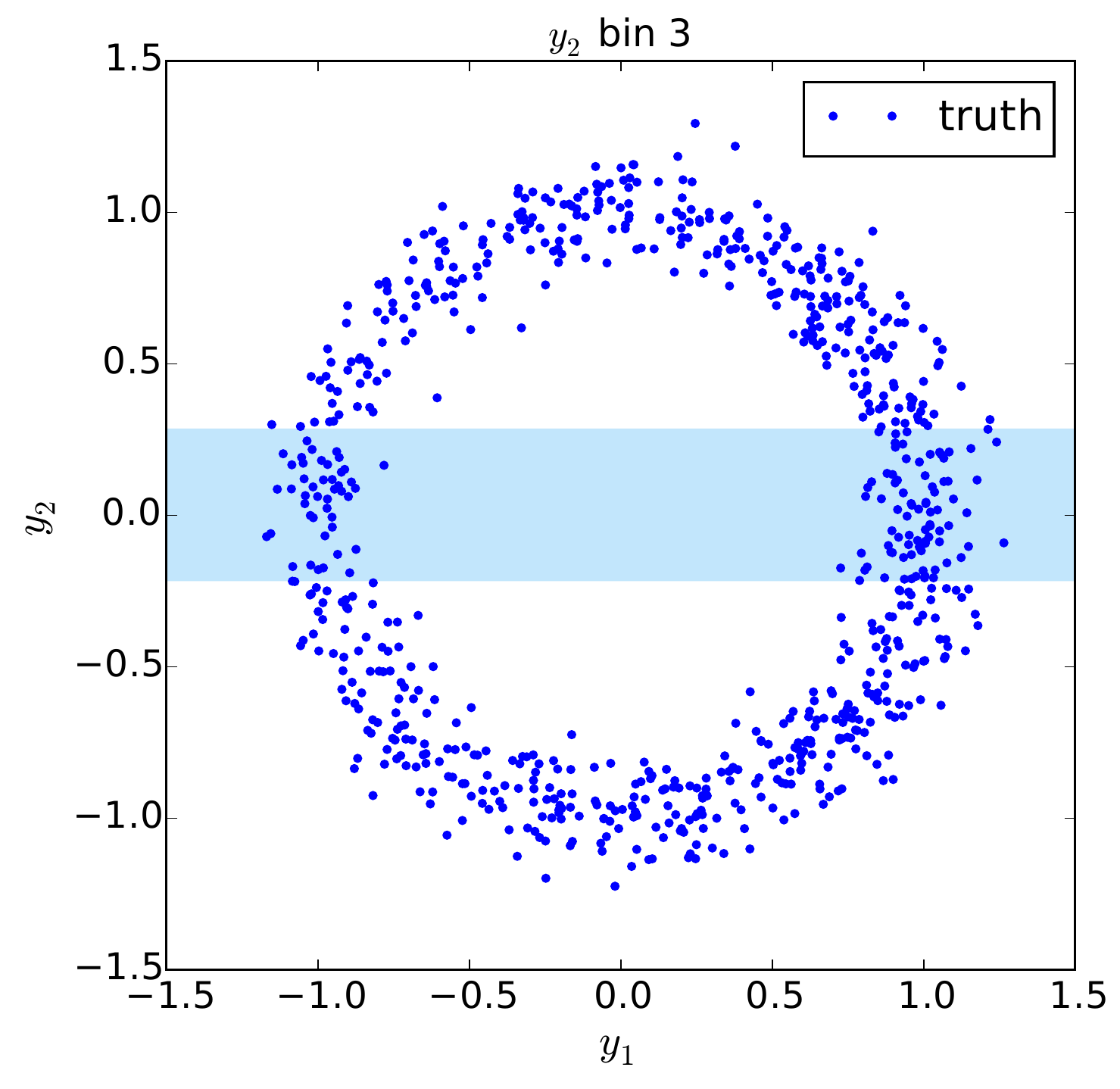}
\includegraphics[width=0.32\columnwidth]{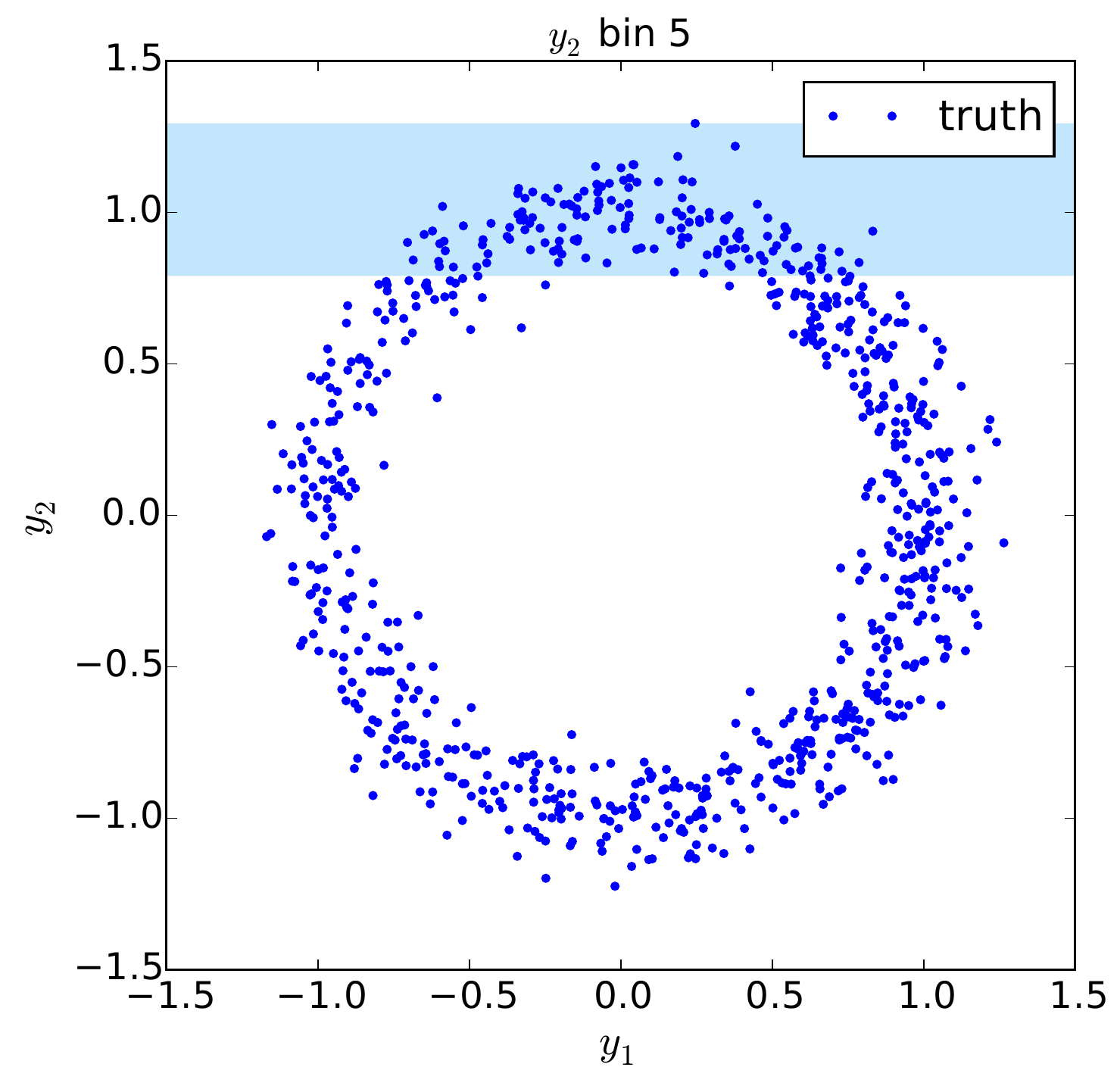}\\
\includegraphics[width=0.32\columnwidth]{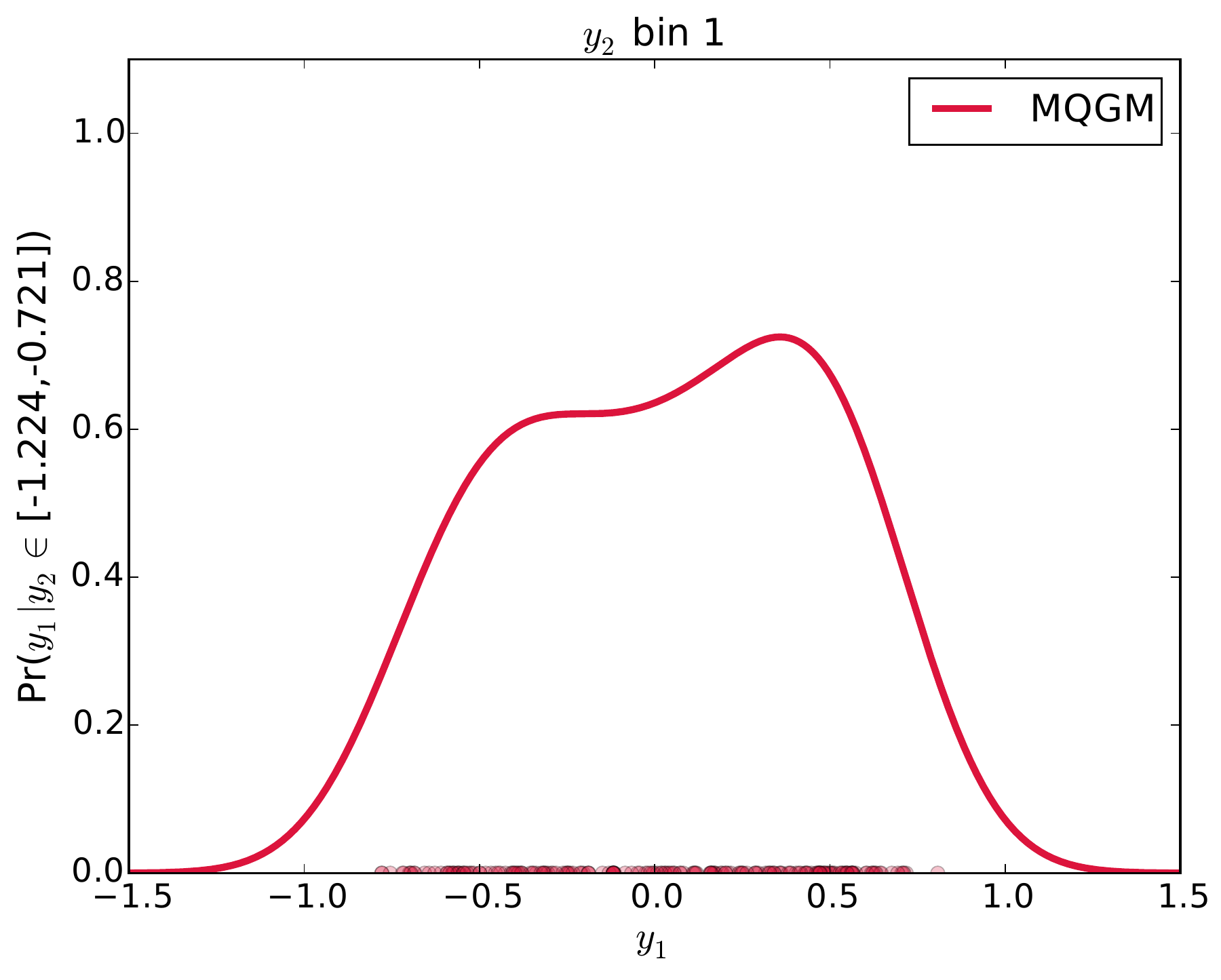}
\includegraphics[width=0.32\columnwidth]{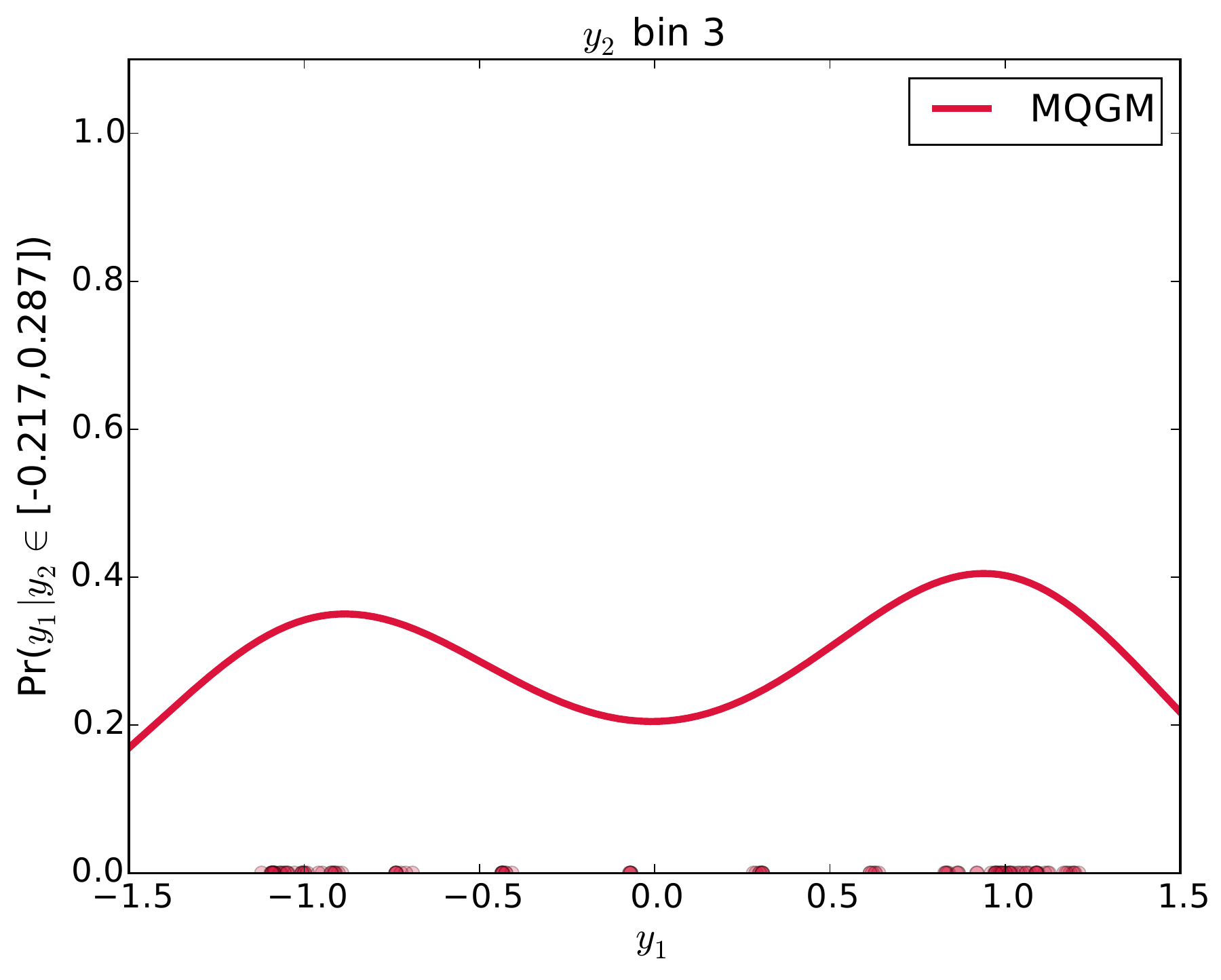}
\includegraphics[width=0.32\columnwidth]{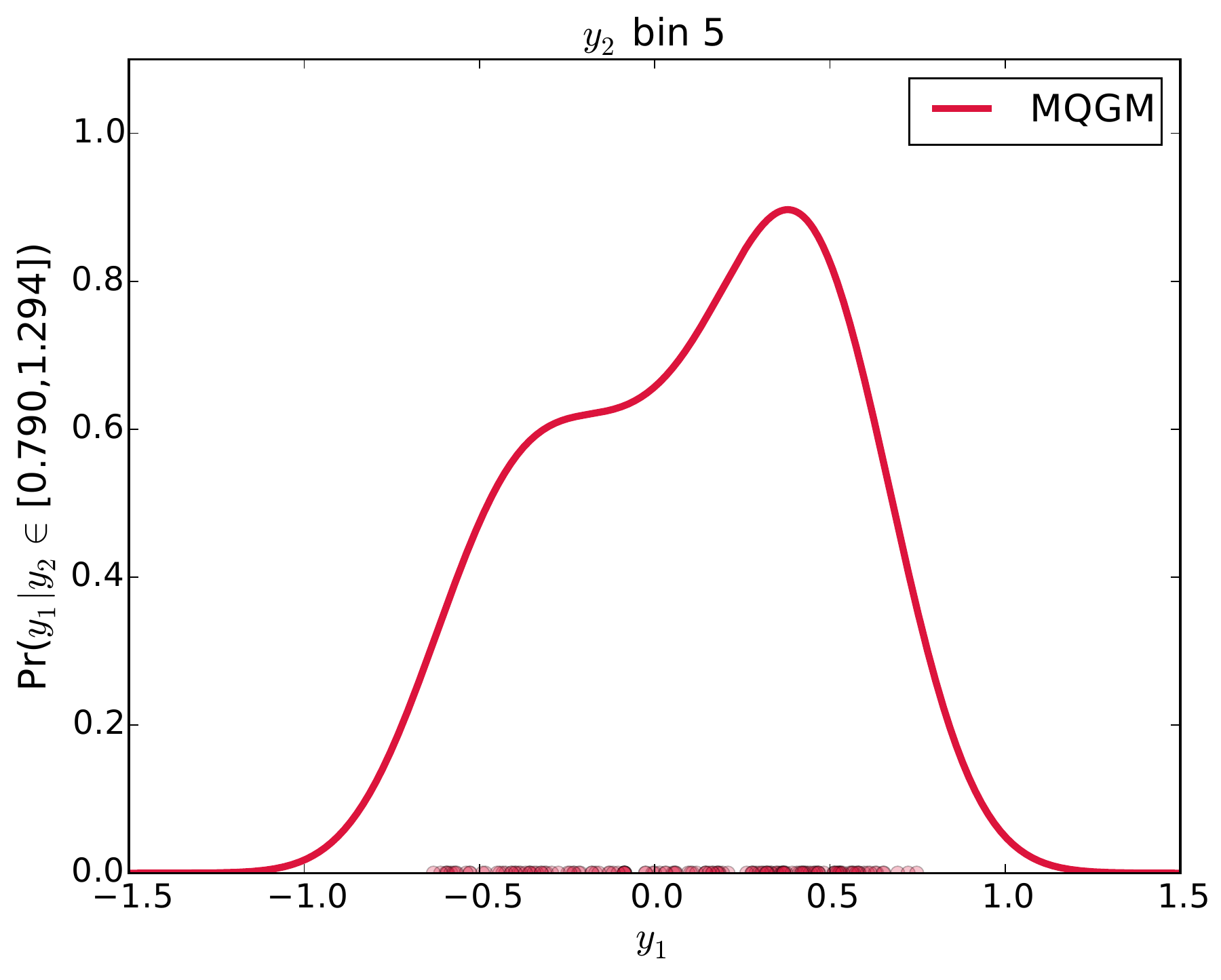}\\
\includegraphics[width=0.32\columnwidth]{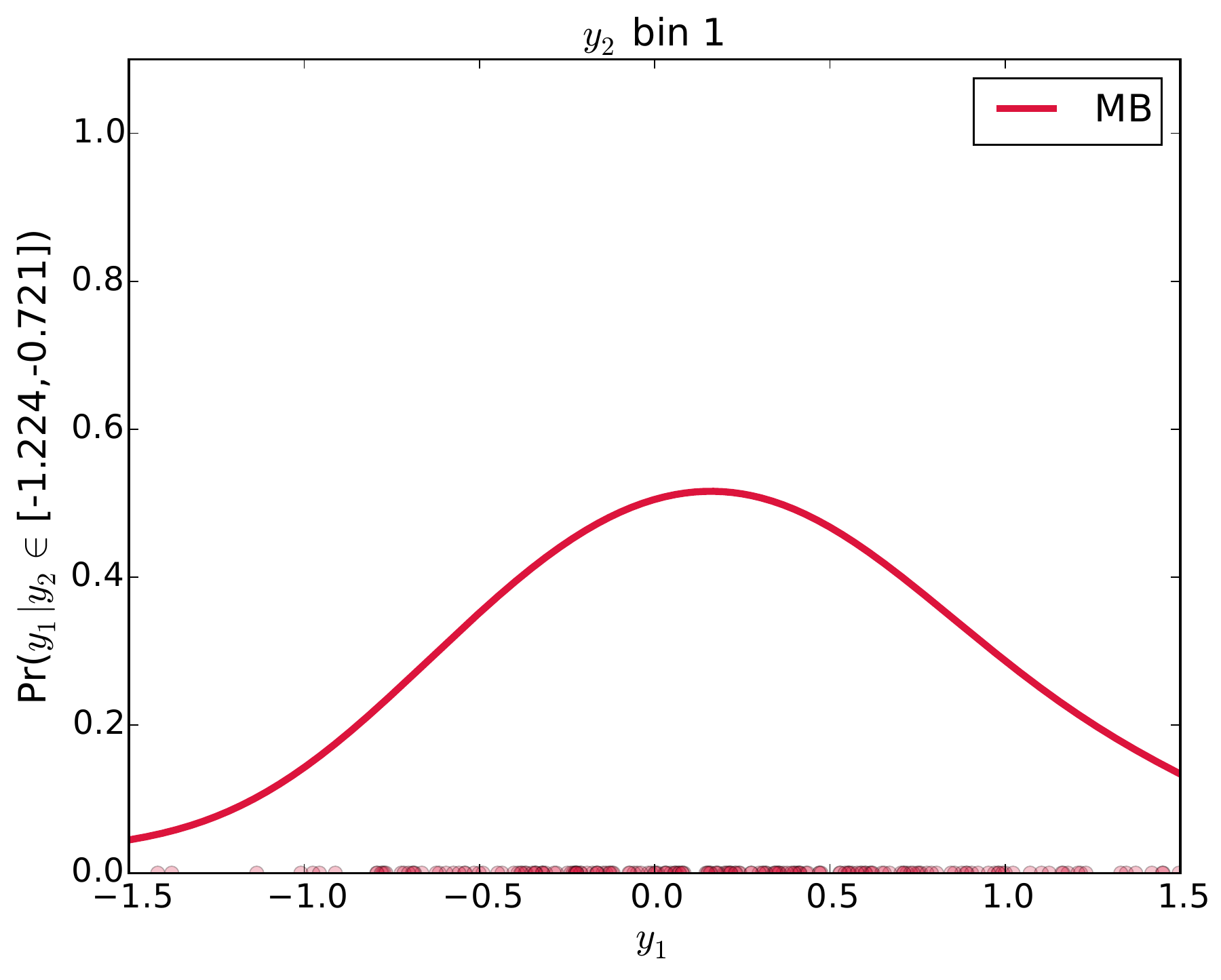}
\includegraphics[width=0.32\columnwidth]{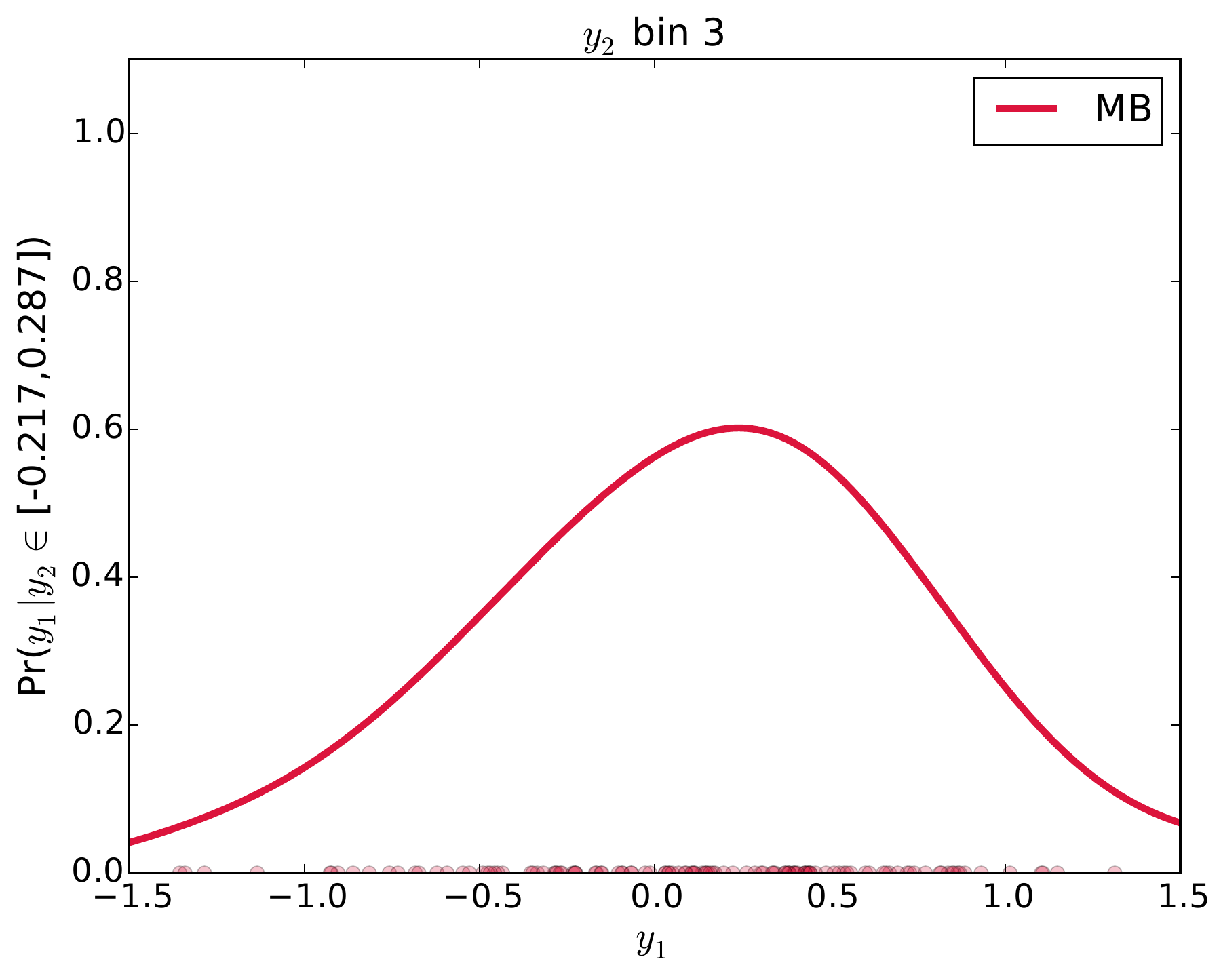}
\includegraphics[width=0.32\columnwidth]{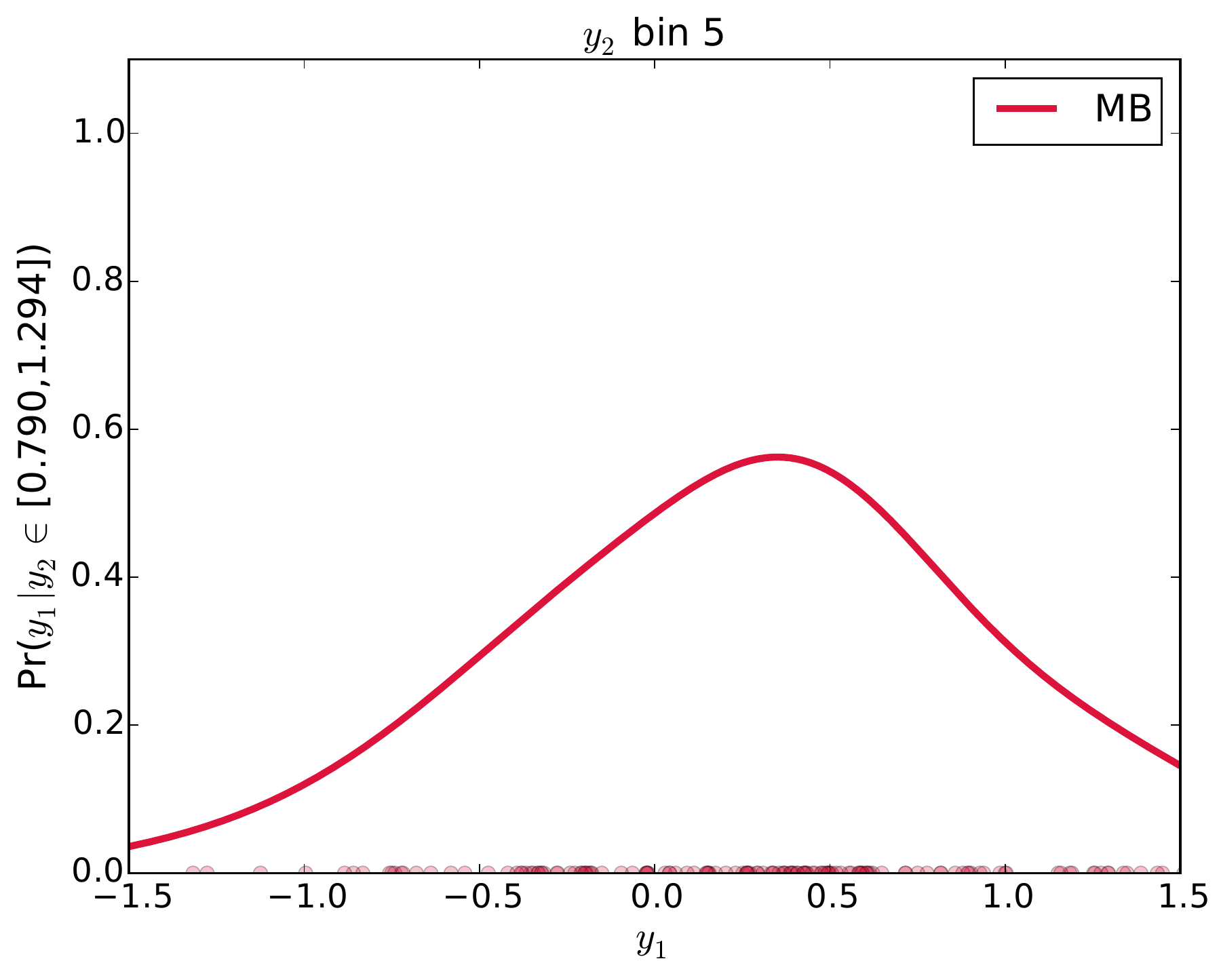} \\
\includegraphics[width=0.32\columnwidth]{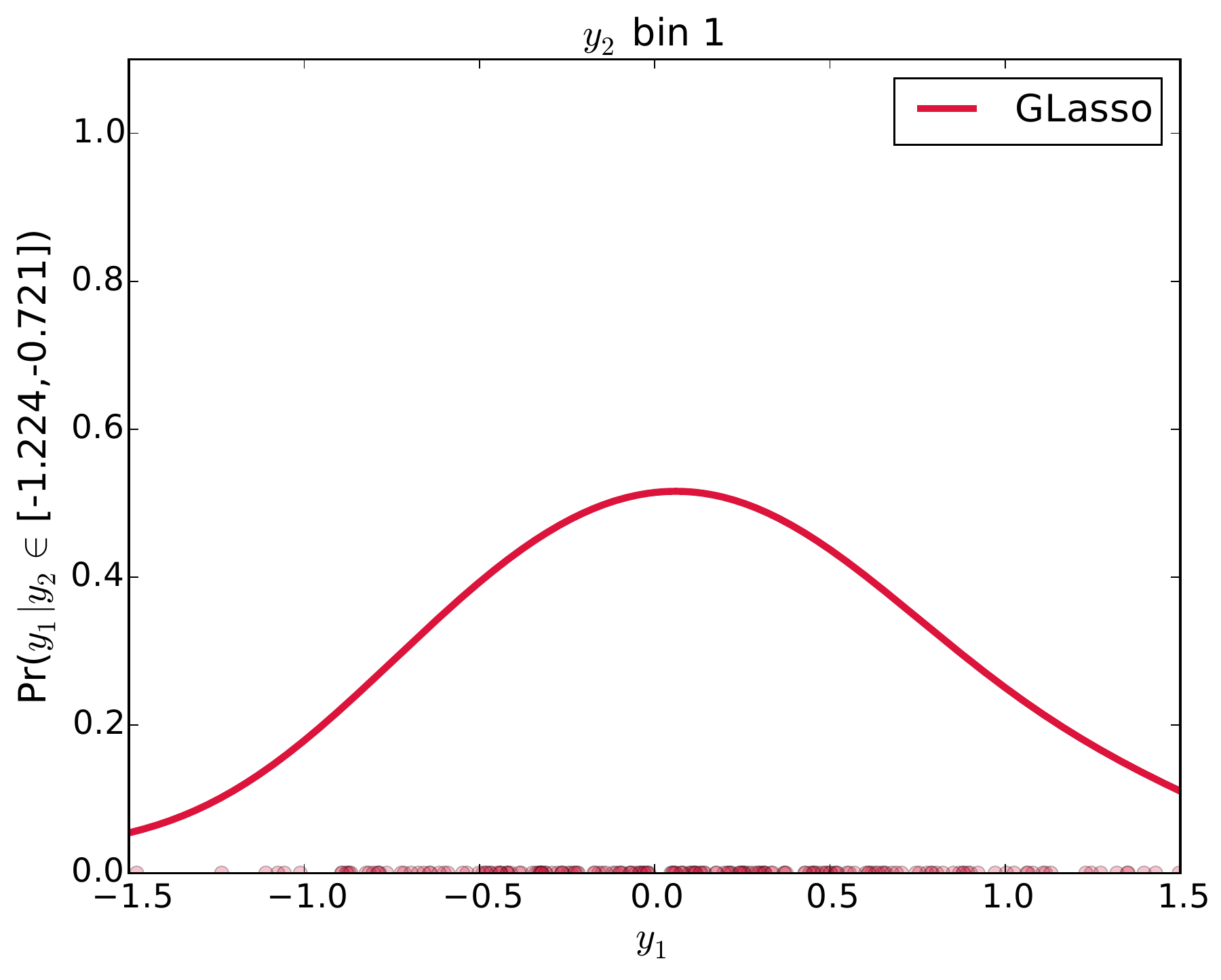}
\includegraphics[width=0.32\columnwidth]{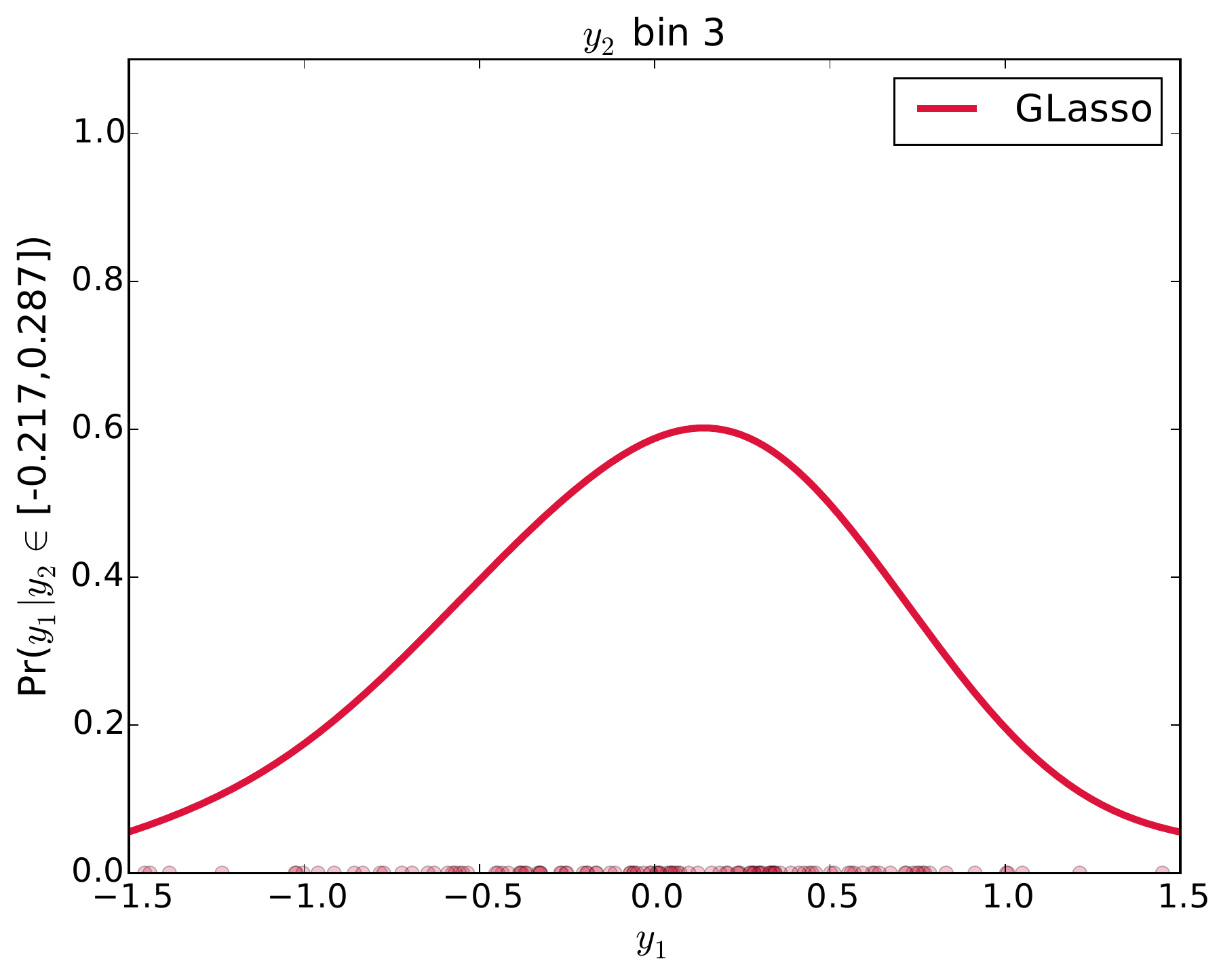}
\includegraphics[width=0.32\columnwidth]{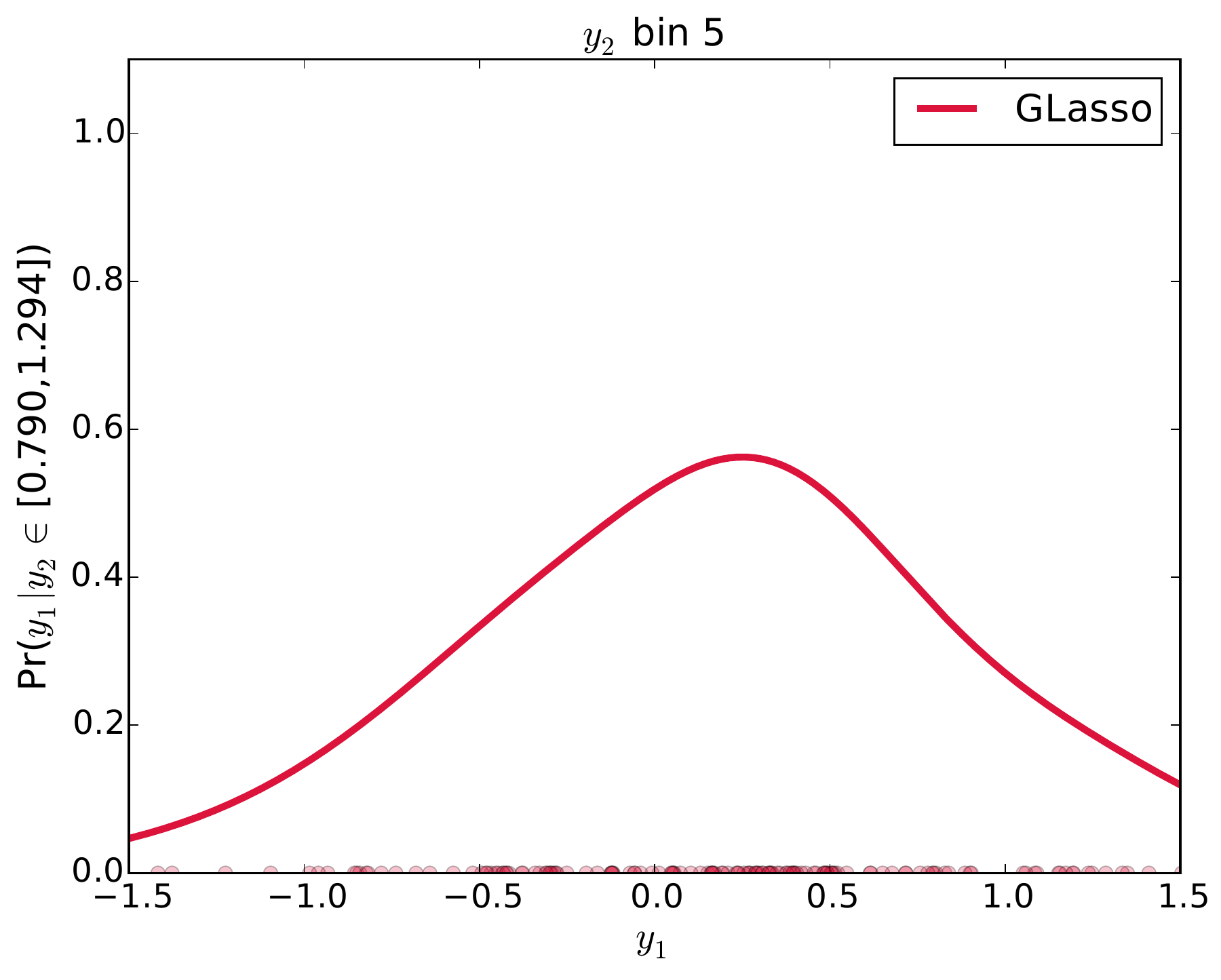} \\
\includegraphics[width=0.32\columnwidth]{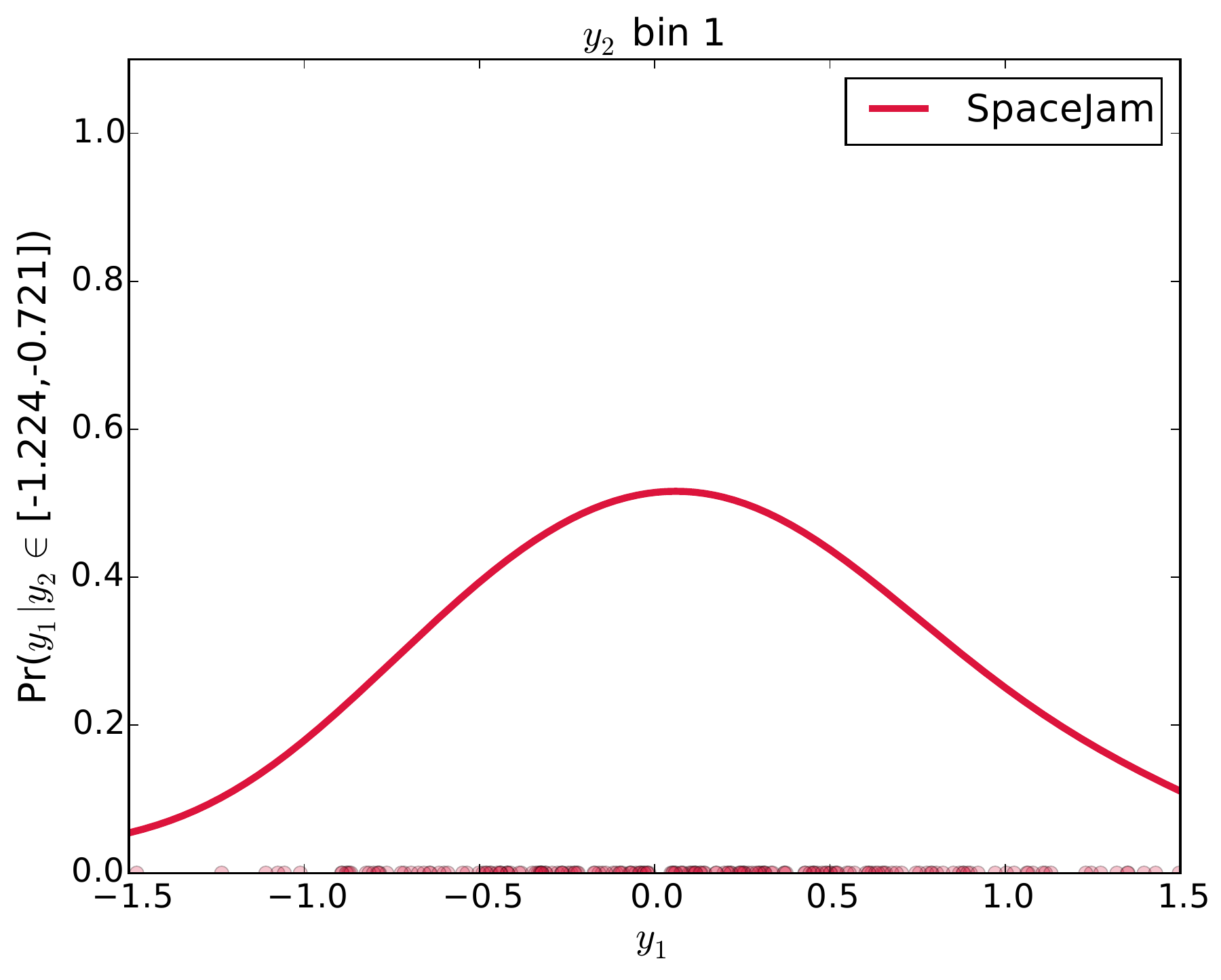}
\includegraphics[width=0.32\columnwidth]{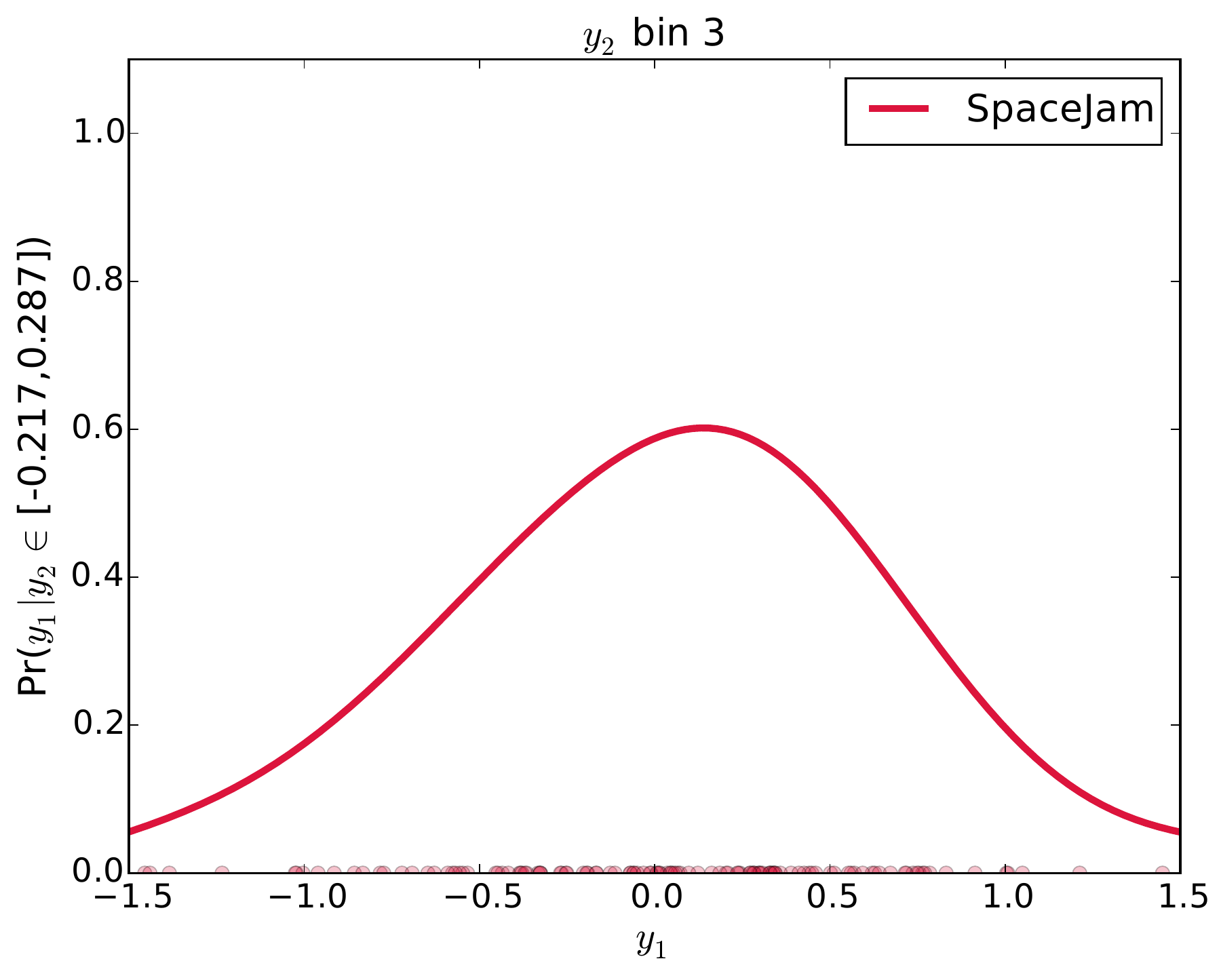}
\includegraphics[width=0.32\columnwidth]{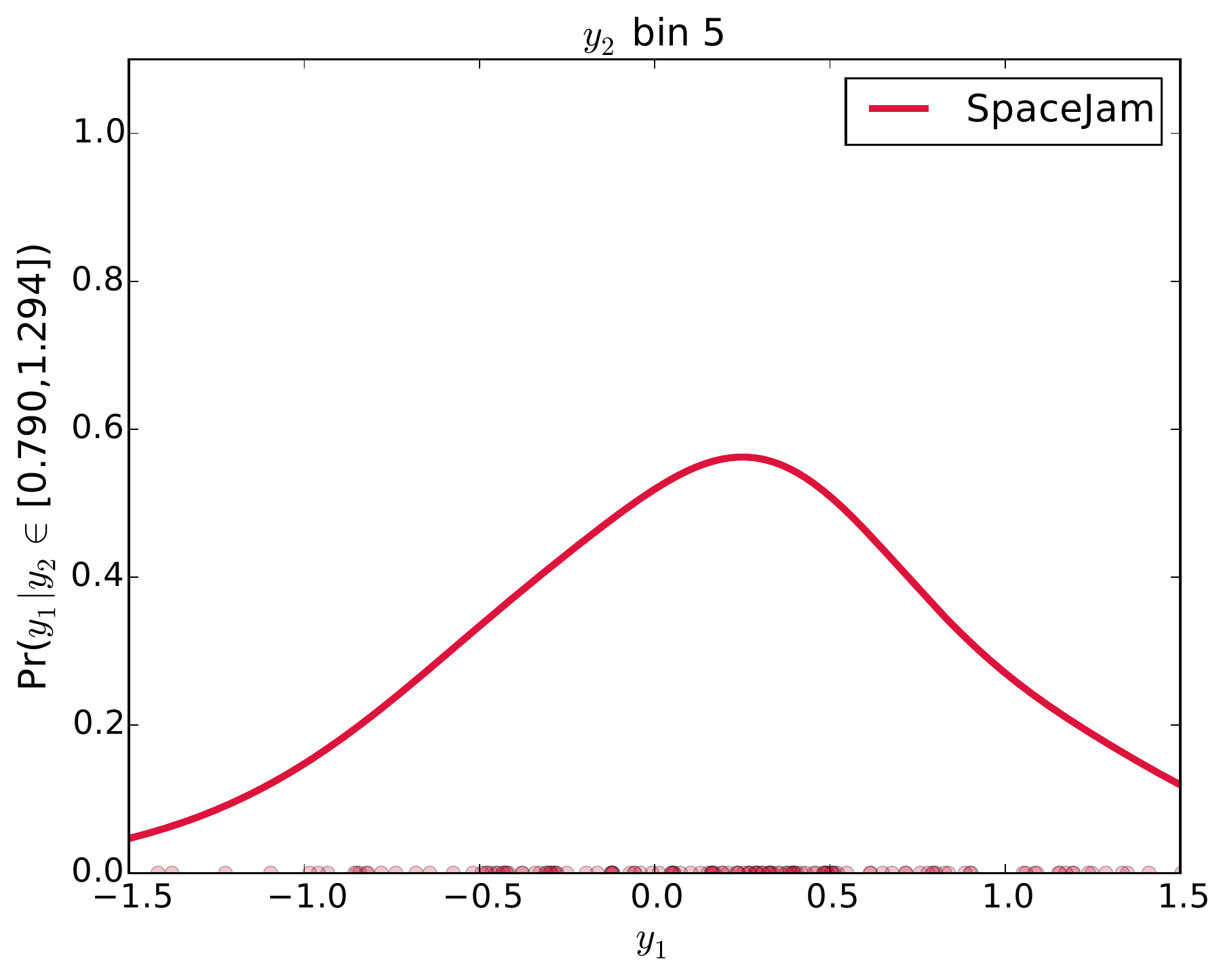}
\caption{Conditional distributions for \mqgm, \mb, \glasso, and \spacejam, fitted to samples from the ring distribution (\tiger~and \laplace's conditionals both look similar to \mb's).  First row: samples from the ring distribution, where each plot highlights the samples falling into a particular shaded bin on the $y_2$ axis. Second through fifth rows: conditional distributions of $y_1$ given $y_2$ for each method, where each plot conditions on the appropriate $y_2$ bin as highlighted in the first row. The \mqgm's conditional distributions are intuitive, appearing bimodal for bin 3, and more peaked for bins 1 and 5. \mb, \glasso, and \spacejam's densities appear (roughly) Gaussian, as expected.}
\label{fig:calib}
\end{figure}

\section{Additional numerical results for modeling flu epidemics}

Here, we plot samples from the marginal distributions of the percentages of flu reports at regions one, five, and ten throughout the year, which reveals the heteroskedastic nature of the data (just as in Section \ref{sec:flu}, for region six).

\begin{figure}[tb]
\centering
\includegraphics[width=1.8in]{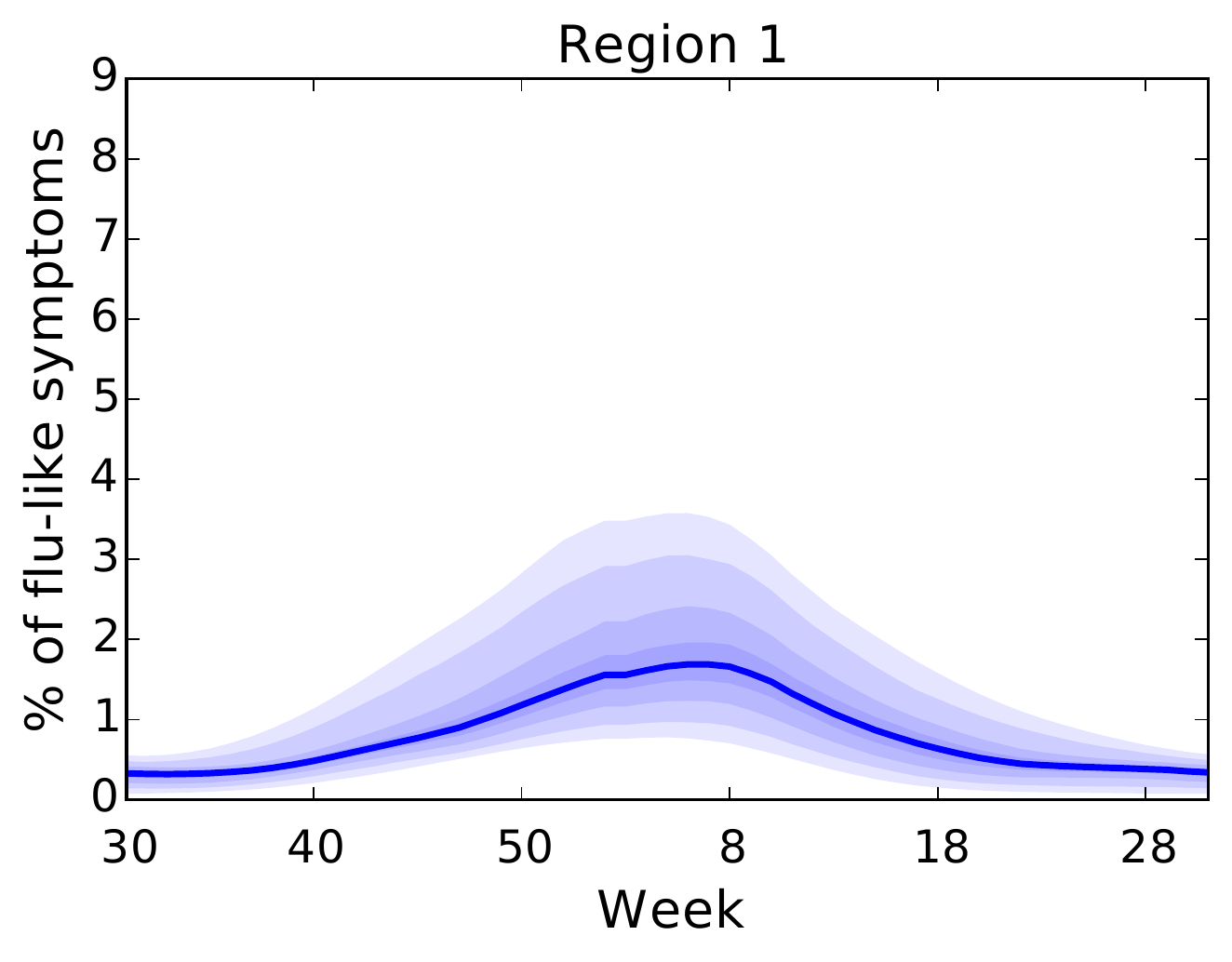}
\hfill
\includegraphics[width=1.8in]{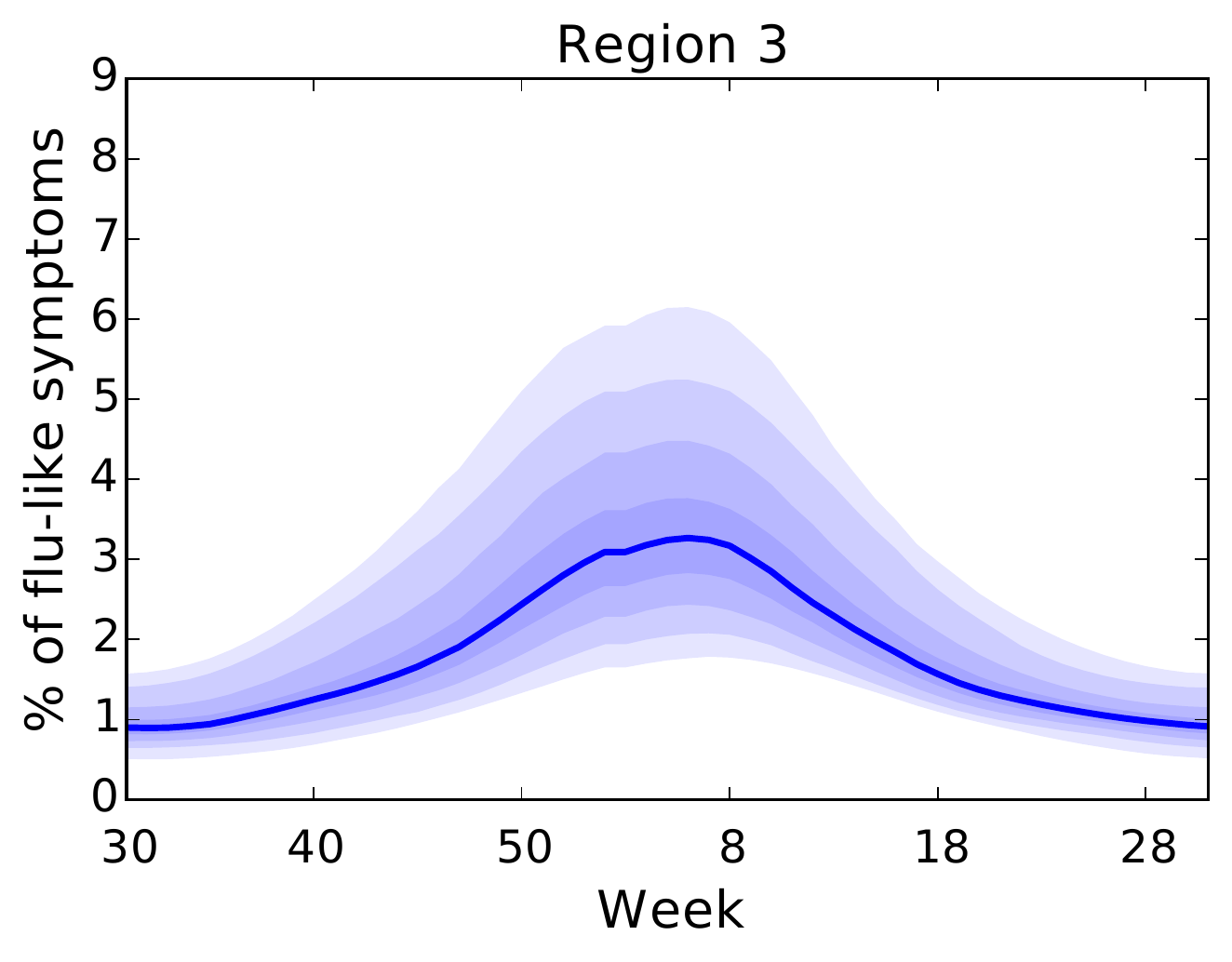} 
\hfill
\includegraphics[width=1.8in]{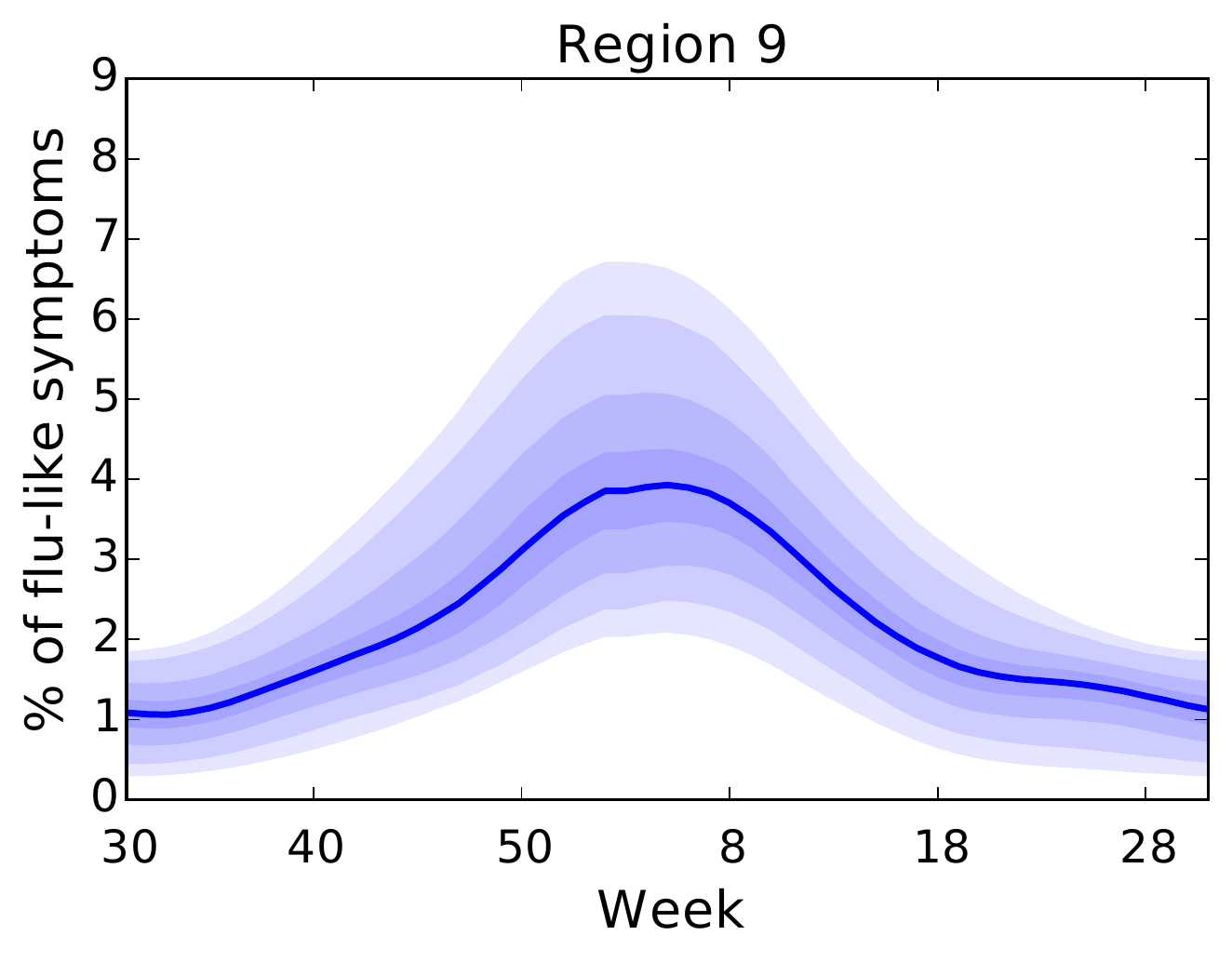}
\caption{Samples from the fitted marginal distributions of the weekly flu incidence rates at several regions of the U.S.; samples at larger quantile levels shaded lighter, median in darker blue.}
\end{figure}

\begin{figure}[tb]
\centering
\includegraphics[width=0.5\textwidth]{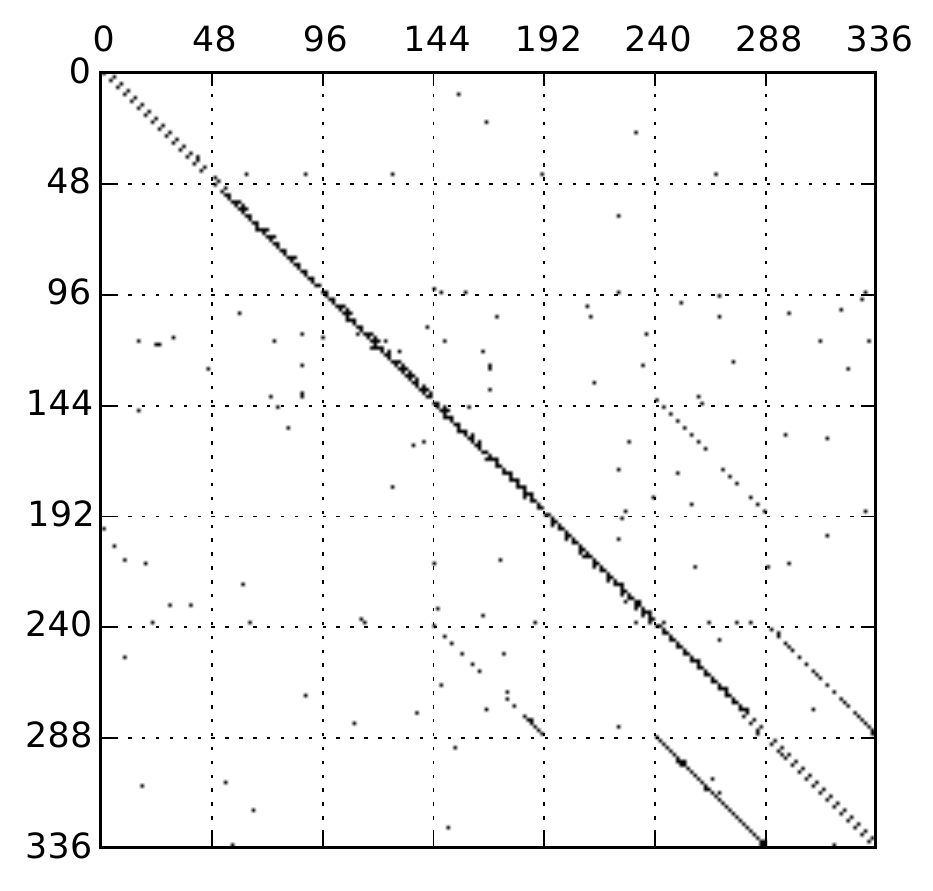} 
\caption{Conditional independencies recovered by the \mqgm~on the wind
farms data; each block corresponds to a wind farm, and black indicates
dependence.} 
\label{fig:wind}
\end{figure}

\section{Sustainable energy application}

We evaluate the ability of \mqgm~to recover the conditional
independencies between several wind farms on the basis of large-scale,
hourly wind power measurements; wind power is intermittent, and thus
understanding the relationships between wind farms can help
farm operators plan.  We obtained hourly wind power measurements from
July 1, 2009 through September 14, 2010 at seven wind farms ($n =
877$, see \cite{hong,wytock2013sparse,alnur} for details).  The primary
variables here encode the hourly wind power at a farm over two days
(\ie, 48 hours), thus $d = 7 \times 48 = 336$.
Exogenous variables were used to encode forecasted wind power and
direction as well as other historical measurements, for a total of 
$q = 3417$. We set $m=5$ and $r=20$. Fitting the 
\mqgm~here hence requires solving $48 \times 7 = 
336$ multiple quantile regression subproblems each of
dimension $\left( (336-1) \times 5 + 3417 \right) \times 20 
= 101,840$.  Each subproblem took roughly 87 minutes, comparable to the
algorithm of \cite{wytock2013sparse}.

Figure \ref{fig:wind} presents
the recovered conditional independencies; the nonzero super- and
sub-diagonal entries suggest that at any wind farm, the previous hour's
wind power (naturally) influences the next hour's, while the nonzero
off-diagonal entries, \eg, in the (4,6) block, uncover farms that may
influence one another. \cite{wytock2013sparse}, whose method
placed fifth in a Kaggle competition, as well as \cite{alnur} report similar findings (see
the left panels of Figures 7 and 3 in these papers, respectively).

\end{document}